%% file: QRMLogic.tex
\documentclass[11pt,letterpaper,onecolumn,noarxiv]{quantumarticle}
\input{macro}

\newcommand{\nnote}[1]{}
\newcommand{\dnote}[1]{}
\newcommand{\snote}[1]{}
\newcommand{\cnote}[1]{}
\usepackage{xprintlen}
\usepackage[normalem]{ulem}
\newcommand\redout{\bgroup\markoverwith{\textcolor{red}{\rule[0.5ex]{2pt}{0.8pt}}}\ULon}
\title{
Geometric structure and transversal logic of quantum Reed--Muller codes}
\author[1,2]{Alexander Barg
}
\email{abarg@umd.edu}
\author[1,3]{Nolan J. Coble}
\email{nolanjcoble@gmail.com}
\author[1,4]{Dominik Hangleiter}
\email{mail@dhangleiter.eu}
\author[5]{Christopher Kang}
\email{ctkang@uchicago.edu}

\affil[1]{Joint Center for Quantum Information and Computer Science, University of Maryland \& NIST}
\affil[2]{Institute for Systems Research, Department of Electrical \& Computer Engineering, University of Maryland}
\affil[3]{Department of Computer Science, University of Maryland, College Park}
\affil[4]{Simons Institute for the Theory of Computing, University of California at Berkeley}
\affil[5]{Department of Computer Science, University of Chicago}

\date{}

\begin{document}

\maketitle

\begin{abstract}

{\textbf{Abstract:}} 
Designing efficient and noise-tolerant quantum computation protocols generally begins with an understanding of quantum error-correcting codes and their native logical operations. The simplest class of native operations are transversal gates, which are naturally fault-tolerant. In this paper, we aim to characterize the transversal gates of quantum Reed--Muller (RM) codes by exploiting the well-studied properties of their classical counterparts. 

We start our work by establishing a new geometric characterization of quantum RM codes via the Boolean hypercube and its associated subcube complex. More specifically, a set of stabilizer generators for a quantum RM code can be described via transversal $X$ and $Z$ operators acting on subcubes of particular dimensions. This characterization leads us to define \emph{subcube operators} composed of single-qubit $\pi/2^k$ $Z$-rotations that act on subcubes of given dimensions. We first characterize the action of subcube operators on the code space: depending on the dimension of the subcube, these operators either (1) act as a logical identity on the code space, (2) implement non-trivial logic, or (3) rotate a state away from the code space. Second, and more remarkably, we uncover that the logic implemented by these operators corresponds to circuits of multi-controlled-$Z$ gates that have an explicit and simple combinatorial description. Overall, this suite of results yields a comprehensive understanding of a class of natural transversal operators for quantum RM codes.

\end{abstract}

\newpage
{\small
\begin{spacing}{1}
\tableofcontents 
\end{spacing}
}
 \newpage

\part{Overview and examples}
\section{Introduction}\label{sec:intro}

\subsection{Logical operations and Reed--Muller codes}

Designing fault-tolerant quantum logic is central to constructing scalable and reliable quantum computers. While early work has shown that storing quantum information is feasible at reasonably low resource costs \cite{aharonov1997fault,knill1998resilient,kitaev1997quantum},     
developing resource-efficient schemes for performing universal fault-tolerant logic remains a challenge. 
The need for such techniques is exacerbated by recent experimental progress on storing information \cite{bluvstein2024logical, acharya2024quantum, da2024demonstration} and performing rudimentary logic \cite{bluvstein2024logical,ryan-anderson_high-fidelity_2024,reichardt_demonstration_2024}. 
A variety of code-theoretic approaches to fault-tolerance have been proposed \cite{beverland2021cost}, including transversal logic, code switching and deformation \cite{paetznick_universal_2013,anderson_fault-tolerant_2014}, and magic-state distillation and injection \cite{bravyi_universal_2005, Bravyj2012magic, campbell2017unifying, ye_logical_2023}.
 All of these endeavors require an intimate understanding of quantum codes and their logical operators. For example, the geometric or algebraic structure of specific families of codes is often exploited when constructing native transversal logic for the codes \cite{zeng2011tranversality,bravyi2013classification,pastawski2015fault,jochym2018disjointness,yoder2016universal}.

\eczoo[Quantum Reed--Muller (RM) codes]{quantum_reed_muller} have been a popular candidate for implementing universal quantum computation ever 
since their introduction a quarter century ago in the work of Steane \cite{Steane1999}. Quantum RM codes 
use the Calderbank-Shor-Steane (CSS) construction to design a family of qubit codes with the
embedded structure of two classical RM codes 
$RM(q,m) \subseteq RM(r,m)$ with the property that $q\le r\le m$. 
These codes encode $k = \sum_{l=q+1}^r\binom ml$ logical qubits into $n = 2^m$ physical qubits and have distance $d = 2^{\min(q+1,m-r)}$, 
or, using short notation, have parameters $[[n,k,d]] = [[ 2^m, \sum_{l=q+1}^r\binom ml, 2^{\min(q+1,m-r)} ]]$. 
Quantum RM codes have many variations, 
e.g.,\ qudit codes \cite{sarvepalli2005nonbinary} or entanglement-assisted RM codes \cite{nadkarni2024entanglement}.
Interestingly, quantum RM codes and their descendants are able to fault-tolerantly realize non-Clifford logic: for example, 
morphed/punctured codes can realize $T$ gates transversally \cite{Bombin2007,Terhal2015,kubica2015universal, vasmer2022morphing}. 
Much attention has been devoted to meticulously puncturing codes to achieve specific logical operators, for instance, for $T$ state distillation \cite{campbell2012magic,Bravyj2012magic,campbell2017unified,campbell2017unifying,haah2018codes}.

Despite the proliferation of quantum RM codes and their descendants, a precise characterization of the code's logical operators remains elusive. Some numerical solutions exhaustively search for logical operators \cite{rengaswamy2020optimality, Webster2022xpstabiliser} but are intractable beyond even modest code sizes. 
To scalably study larger codes, we require a deeper theoretical intuition of the structure of logical operators. 
We address this problem and study logical gates of quantum RM codes ``from first principles,'' investigating both the structure of the codes and their logical operators in higher levels of the Clifford Hierarchy. To this end, we first we elucidate the geometric structure of quantum RM codes. 
Then, using this geometry, we isolate and characterize a class of transversal physical operators which effect $k$-th level Clifford logicals. Our results provide a deeper intuition underlying the geometry of logic in quantum RM codes, enabling direct enumeration of logic that can be realized transversally for specific code parameters. 

Apart from the problem areas discussed above, we were motivated to study quantum RM codes for a number of other reasons. Geometric structure of classical RM codes, which we use in an essential way, is an established area
of classical coding theory with multiple links to discrete geometry (e.g., \cite{AK98}). Recently, RM codes have gained renewed prominence in classical coding theory due to 
spectacular advances in understanding their performance in noisy channels. It has been proved that they attain Shannon capacity on the binary erasure channel \cite{kudekar2017}, and \cite{Kumar2016qRM} gave a quantum analog of this result.
They are also intimately connected to the celebrated family of polar codes \cite{arikan2009channel,arikan2010survey,mondelli2014polar}, which have not only advanced classical coding theory in several aspects \cite{sasoglu2012polarization} but also have been adopted into industrial communication standards \cite{Hui2018polar}.




\subsection{Our results}
In this paper, we characterize a broad class of logical operators of the quantum RM codes using geometric intuition.  We begin by constructing the family of quantum RM codes using the structure of the $m$-dimensional hypercube, generalizing the well-known quantum \eczoo[hypercube code family]{hypercube_quantum}.
In this picture, physical qubits are associated with the vertices of an $m$-dimensional hypercube, and $X$ and $Z$ stabilizers are defined using subcubes (or rather sub-hypercubes) with specified dimensions. We address the question of whether transversal operators---in particular, diagonal and transversal operators living in the Clifford Hierarchy---acting on subcubes can act as logical operations on quantum RM codes. We answer in the affirmative, proving two
related but distinct categories of results:
\begin{enumerate}
    \item \textbf{Validity:} In \cref{sec: sufficiency conditions}, we give necessary and sufficient conditions for when such ``subcube operators'' are valid logical operations.
    \item \textbf{Logic:} In \cref{sec: signed logic,sec: unsigned logic} we show exactly what logical circuits these subcube operators implement.
\end{enumerate}
Our geometric description of quantum RM codes draws on, and develops, a geometric presentation of classical RM codes \cite{MS77}. 
While the classical version is relatively well known, the corresponding quantum characterization has received little to no attention.

Quantum RM codes are defined by three parameters: $m$, the dimension of the hypercube of physical qubits, $q$, the codimension of the $X$ stabilizers, and $r$, where $r + 1$ is the dimension of the $Z$ stabilizers:
\begin{restatable}[Quantum RM codes]{definition}{quantumrmcode}\label{def: quantum RM code hypercube}
     Let $0\leq q\leq r\leq m$ be non-negative integers. The \emph{quantum Reed--Muller code} of order $(q,r)$ and length $2^m$, denoted by $QRM_m(q,r)$, is defined as the common $+1$ eigenspace of a Pauli stabilizer group $\mcS\coloneqq\langle S_X,S_Z\rangle$, with stabilizer generators given by
    \begin{align}
        S_X&\coloneqq \br{X_A\Bigmid A \text{ is an $(m-q)$-cube}},  
        \\
        S_Z&\coloneqq \br{Z_A\Bigmid A \text{ is an $(r+1)$-cube}},
    \end{align}
\end{restatable}
In the above, when we say ``$\ell$-cube'', we mean an $\ell$-dimensional subcube of the $m$-dimensional hypercube. The operator $X_A$ is the transversal operator which acts as Pauli-$X$ on the qubits contained within a specified subcube $A$, and similarly for $Z_A$. Though not immediately obvious, it turns out that this definition is equivalent to the standard construction of quantum Reed--Muller codes from a pair of classical Reed--Muller codes (see \cref{sec: QRM-codes}). Readers familiar with the 
\eczoo[hypercube codes]{hypercube_quantum} will recognize them as a particular case of this  code family given by $QRM_m(0,1)$. Stabilizer generators for $QRM_4(0,2)$ are shown in \cref{fig: m4 code stabilizers}. A \emph{symplectic basis} (see \cref{sec: CH and errors}) for the logical Pauli space of $QRM_4(0,2)$ in terms of subcube operators is given later in \cref{fig: m4 code example}.

\begin{figure}[t!]
    \centering
    \includegraphics[width=\linewidth]{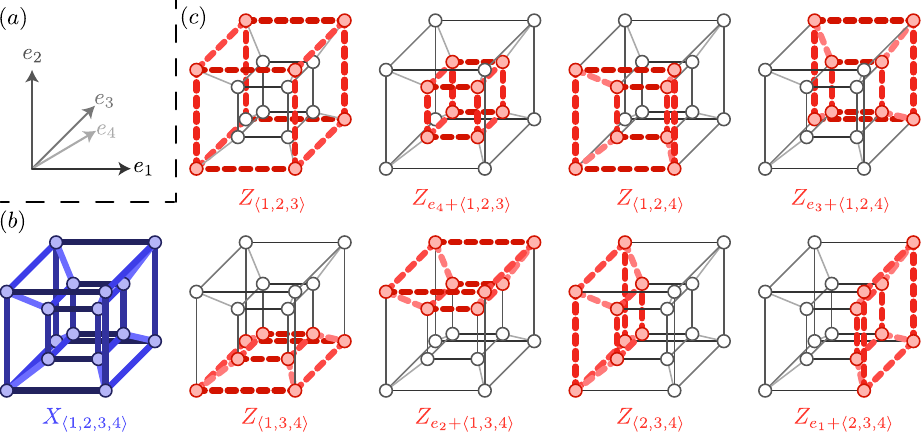}
    \caption{Stabilizer generators for the code $QRM_4(0,2)$. By definition, every $(r+1)=3$ cube defines a $Z$ stabilizer and the unique $(m-q)=4$-cube (i.e., the entire hypercube) defines the only $X$ stabilizer. (a) Orientation of the 4-cube.
    (b) Global $X$ is the only $X$ stabilizer, by definition, represented here by the (blue) solid cube. 
    (c) The (red) dashed cubes indicate the 8 subcubes in the 4-dimensional hypercube which define the $Z$ stabilizer generators of the code.}
    \label{fig: m4 code stabilizers}
\end{figure}

In the same way that transversal operators on subcubes of particular dimensions generate the Pauli stabilizers for $QRM_m(q,r)$, operators on subcubes of other dimensions generate the groups of \emph{logical} $X$ and $Z$ operators for $QRM_m(q,r)$. Building upon this idea we construct transversal logical operators comprised of diagonal $Z$ rotations acting on subcubes. We consider the single-qubit gates $\Z{k} = \ketbra{0} + e^{i \frac{\pi}{2^k}} \ketbra{1}$ where $Z(0)=\eye$, $Z(1) = S$, $Z(2) = T,$ etc. These gates also correspond with increasing levels of the \emph{Clifford Hierarchy}. 

Our first main result clarifies necessary and sufficient conditions for when the transversal application of $\Z{k}$ to a subcube will implement logic on $QRM_m(q,r)$. We prove the following:
\begin{theorem*}[Validity; Informal version of \cref{thm: subcube dimension implies logic}]
    Let $0\leq q\leq r\leq m$ be non-negative integers and consider the quantum Reed--Muller code $QRM_m(q,r)$. Suppose $A$ is a subcube of the $m$-dimensional hypercube.
    \begin{enumerate}
        \item If the dimension of $A$ is $\leq q+kr$, then applying $\Z{k}$ on the qubits in $A$ \underline{does not} preserve the code space.
        \item If the dimension of $A$ is $\geq q+kr+1$ and $\leq (k+1)r$, then applying $\Z{k}$ on the qubits in $A$ will implement a \underline{non-trivial} logical operation on the code space.
        \item If the dimension of $A$ is $\geq (k+1)r+1$, then applying $\Z{k}$ on the qubits in $A$ will implement a \underline{logical identity} on the code space.
    \end{enumerate}
\end{theorem*}
In other words, the dimension of a subcube determines---based on the values of $k$, $q$, and $r$---whether the $\Z{k}$ operator applied to the subcube implements trivial logic, non-trivial logic, or no logic. Prior works have studied necessary and sufficient conditions for when a \emph{global} application of $\Z{k}$ to every qubit of a quantum RM code will perform logic \cite{HLC21, HLC22}. In the particular case
when the subcube $A$ is the $m$-dimensional hypercube itself, the above theorem reproduces these results.

While the authors of \cite{HLC21, HLC22} determine conditions for when the global $\Z{k}$ operation performs logic on $QRM_m(q,r)$, they do not give descriptions of the implemented logical circuits. An earlier work that addresses this question \cite{rengaswamy2020optimality} does detail the logical performed by global $\Z{k}$ in the particular case of $QRM_m(r-1,r)$ codes, i.e., when $q=r-1$. Extending these results, in our work we give a complete description of the logical circuit implemented by the transversal $\Z{k}$ operation applied to \emph{any} subcube, and for an arbitrary $QRM_m(q,r)$. We will go into the details of the implemented logical circuits in the next section; for now we give some intuition for their structure. 

As the $\Z{k}$ operators are diagonal operators in the Clifford Hierarchy, they should likewise implement diagonal logical operators in the Clifford Hierarchy. Diagonal operators in the Clifford Hierarchy are fully classified as circuits composed of multi-qubit controlled versions of $\Z{k}$ operators \cite{CGK17}. Our \cref{thm: subcube dimension implies logic} is enough, already, to infer the particular type of circuit implemented by transversal $\Z{k}$ applied to a subcube. First, note that the square of $\Z{k}$ is precisely $\Z{k}^2=\Z{k-1}$. Now, suppose that $\Z{k}$ operators are applied to a subcube whose dimension is at least $q+kr+1$. Then \cref{thm: subcube dimension implies logic} implies this operation preserves $QRM_m(q,r)$. In fact, it implies something much stronger: since $q+kr+1\geq (k-1+1)r+1$, it must be that $\Z{k-1}$---the \emph{square} of $\Z{k}$---applied to the same subcube acts as \emph{logical identity} on $QRM_m(q,r)$. In other words, \cref{thm: subcube dimension implies logic} implies that such subcube operators are necessarily logically Hermitian.  

Now, the classification from \cite{CGK17} implies that the \emph{only} diagonal and Hermitian operators in the Clifford Hierarchy are circuits composed of multi-qubit controlled-$Z$ operators. Recall that the \emph{$\ell$-qubit controlled-$Z$ gate} is a diagonal gate acting on $\ell$ qubits that 
 applies a $-1$ phase to the all ones computational basis state, and leaves all other computational basis states fixed. Our second main result is to determine precisely what logical multi-controlled-$Z$ circuit is implemented by a subcube operator:
\begin{theorem*}[Logic; Informal version of \cref{thm: Zk and tildeZk equivalence conditions}]
    Consider the quantum Reed--Muller code $QRM_m(q,r)$, and suppose $A$ is a subcube of the $m$-dimensional hypercube. If the dimension of~$A$ satisfies condition 2 of the Validity Theorem, then the operator $\Z{k}_A$ implements a logical circuit composed of $\leq k$-qubit explicitly computable controlled-$Z$ gates.  
\end{theorem*}
In the next section we give a more detailed description of our main results, together with some of the main definitions of our work. Some worked-out examples are presented in \cref{sec: Examples}, including the $[[ 2^m,\binom mr]]$ code $QRM_m(r-1,r)$ studied earlier in \cite{rengaswamy2020optimality} which encompass the hypercube code family.

\section{Overview of the paper}\label{sec: overview} 

\subsection{Preliminaries}\label{sec: prelims}
Throughout this paper we use the following notation. $\NN=\ZZ_{>0}$ denotes the positive integers; non-negative integers will always be denoted by $\ZZ_{\geq 0}$. Given $m\in\NN$, $[m]\coloneqq\br{1,2,\dots,m}$. By convention, $\sum_\emptyset = 0$ and $\prod_\emptyset=1$. $\mcP_\ell$ refers to the $\ell$-qubit Pauli group. $\powerset{\cdot}$ refers to the power set of the input.

\eczoo[Stabilizer codes]{stabilizer} form one of the most promising routes toward fault-tolerant quantum computation \cite{gottesman1997stabilizer,Gottesman2024}. Let $\mcP_1=\br{\pm\eye, \pm i\eye, \pm X, \pm i X,\pm Y, \pm i Y, \pm Z, \pm iZ}$ denote the single-qubit Pauli group, and let $\mcP_n\coloneqq \mcP_1^{\otimes n}$ be the $n$-qubit Pauli group. Consider a commutative subgroup $\mcS\leq\mcP_n$ such that $-\eye\notin\mcS$. Such a subgroup is called a \emph{(Pauli) stabilizer group}. A stabilizer group necessarily has order $2^{n-k}$ for some integer $k\in\br{0,\dots,n}$, and in particular has a (non-unique) set of $\log(\abs{\mcS})=n-k$ independent generators. Given a stabilizer group $\mcS$, its joint $+1$ eigenspace defines a $2^k$-dimensional subspace $\mcC_\mcS$ of $(\CC_2)\n$ known as the \emph{stabilizer code} associated to $\mcS$:
\begin{equation}
	\mcC_\mcS\coloneqq \br{\ket\psi\in(\CC_2)\n \Bigmid S\ket\psi=\ket\psi,\;\forall S\in\mcS}.
\end{equation}
Of particular interest are the so-called \emph{CSS codes}, which have stabilizer groups generated by operators consisting of tensor products of either $X$ and $\eye$ ($X$ operators) or $Z$ and $\eye$ ($Z$ operators). Given a bit string $v=(v_1,\dots,v_n)\in\ZZ_2^n$, we define the $n$-qubit Pauli operators $X(v)\coloneqq\bigotimes_{i=1}^n X^{v_i}$ and $Z(v)\coloneqq\bigotimes_{i=1}^n Z^{v_i}$. A CSS code is then given by a stabilizer group with a generating set of the form
\begin{equation}
	\br{X(x),Z(z)\Bigmid x\in B_X, z\in B_Z},
\end{equation}
where $B_X$, $B_Z\subseteq\ZZ_2^n$ are sets of length-$n$ bit strings. 

For the remainder of this paper, we suppose that $\mcC$ is a CSS code with stabilizer group $\mcS\coloneqq\br{X(x),Z(z)\mid x\in B_X, z\in B_Z}$. 
The commutativity condition of $\mcS$ can be equivalently phrased as the requirement that every $x\in B_X$ have even overlap with every $z\in B_Z$, i.e., $|\{i:x_i=z_i=1\}|\in 2\ZZ$. This is part of a more fundamental connection between CSS codes and classical error-correcting codes; we will discuss this later in \cref{sec:Pauli group and quantum codes}.

As their name suggests, Pauli stabilizers of $\mcC$ are operators that leave every state $\ket\psi\in\mcC$ invariant, i.e., they implement the \emph{logical identity} $\overline{\eye}$ on $\mcC$. We can consider this phenomenon more generally: the \emph{unitary stabilizer group} of $\mcC$ is the set of unitaries, $U\in\unitary(2^n)$ that implement the logical identity on $\mcC$,
\begin{equation}
    \mcS^*\coloneqq \br{U\in\unitary(2^n)\Bigmid U\equiv\overline{\eye}},
\end{equation}
where by $U\equiv\overline{\eye}$ we mean that $U\ket\psi=\overline{\eye}\ket\psi=\ket\psi$ for every $\ket\psi\in\mcC$. Just as important are the \emph{undetectable unitary errors}, $\mcN^*$, defined as
\begin{equation}
    \mcN^*\coloneqq \br{U\in\unitary(2^n)\Bigmid U\mcC=\mcC}.
\end{equation}
As opposed to stabilizers, which fix $\mcC$ pointwise, undetectable errors simply have the property that $U\ket\psi\in\mcC$ for every code state. Since $\mcC$ is invariant under $\mcS^*$, we have that $\mcS^*\subseteq\mcN^*$; the set $\mcE^*\coloneqq\mcN^*\setminus\mcS^*$ is the set of \emph{logical operators} of $\mcC$. Although every $U\in\mcE^*$ preserves $\mcC$, it does not stabilize it since by definition there is some code state $\ket\psi$ such that $U\ket\psi\neq\ket\psi$.

A key goal of quantum error-correction is not only to design quantum codes that have favorable encoding rate and error-correcting capability (large distance), but that also come with a well-understood 
set of operators $U\in\mcE^*$, which can be used to implement controlled logic on encoded states. Another desirable property of
this operator set is the ease of their physical implementation. 
In many cases, this is achieved by \emph{transversal} operators, $n$-qubit operators that are tensor products of single-qubit gates, $U=\bigotimes_{i\in[n]} U_i$. For a given transversal operator $U$, the \emph{support} of $U$, $\supp(U)\subseteq[n]$, is the set of qubits where $U$ is not the identity, and the \emph{weight} of $U$ is the number of qubits in its support.


\subsection{Geometry of hypercubes and quantum RM codes}\label{sec: geometry}
We begin by introducing the simple geometry underlying quantum Reed--Muller codes. This geometric structure provides immediate intuition as to the structure and relations between logical operators. We employ the language of hypercubes and their subcubes, though it can be equivalently phrased in terms of the affine geometry $AG(m,2)$. 

To define the hypercube, consider the group $\ZZ_2^m$ for $m \in \mathbb{N}$, that is, length-$m$ bit strings under bit-wise addition modulo 2. $\ZZ_2^m$ can be generated by the set of binary vectors of Hamming weight one (the standard basis), denoted below by $S=\br{e_i}_{i=1}^m$, where $e_i=(e_{j1},\dots,e_{jm})$ with $e_{ji}=\delta_{ji}$ for all $j\in[m]$. We will frequently abuse notation by referring to $S$ and $[m]$ interchangeably. For example, when we write ``let $i\in S$'', this should be interpreted to mean ``let $e_i\in S$ for $i\in[m]$''.

The hypercube graph can be defined as the \emph{Cayley graph} of $(\ZZ_2^m,S)$ in the following way. We define a graph, $G_m=(V,E)$, whose vertex set is $V=\ZZ_2^m$, and where than is an edge $e\in E$ between two vertices $x,y\in\ZZ_2^m$ whenever $y=x+e_i$ for some $i\in [m]$. This definition implies that $G_m$ is, in fact, an edge-colored graph, where the color of an edge $(x,x+e_i)$ is defined to be $i\in S$. While the Cayley graph definition is useful for a geometric picture, it does not fully capture the incidence relations of \emph{subcubes} (or rather, sub-hypercubes) within the $m$-dimensional hypercube; a priori it only captures vertices and edges. Instead, we capture incidences between subcubes by considering the complex of standard cosets in $\ZZ_2^m$, which correspond to subgroups generated by the standard basis $S$.



\begin{definition}[subcubes of the hypercube]
A \emph{standard subcube} of the $m$-dimensional hypercube is a subgroup of the form $\standard{J}$, where $J\subseteq S$ is a subset of generators. That is, the bit strings that are contained in $\standard{J}$, are precisely those whose support lies entirely within the set $J$, viewed as a subset of $[m]$.

A \emph{subcube} is any coset of a standard subcube, i,e., subsets of $\ZZ_2^m$ of the form $A\coloneqq x+\standard{J}$ for some $x\in\ZZ_2^m$. The set $J$ is called the \emph{type} of $A$. We write $A\subcubeeq\ZZ_2^m$ to indicate that the subset $A$ is a subcube.
Note that the bits appearing outside of $J$ form an invariant of a subcube $A$ of type $J$, for any such
subcube. In other words, given two bit strings $x,y\in x+\standard{J}$, $x_i=y_i$ for every $i\in S\setminus J$.

The \emph{dimension} of a subcube $A=x+\standard{J}$ is defined to be $\dim A\coloneqq \abs{J}$. This corresponds precisely to the expected notion of dimension.
\end{definition}
Viewing subcubes of the hypercube as standard cosets within $(\ZZ_2^m,S)$ allows us to algebraically describe the intersection of two subcubes: given two subcubes, $x+\standard{J}$ and $y+\standard{K}$, their intersection is either empty, or else it is a subcube of the form $z+\standard{J\cap K}$. In other words, the intersection of subcubes of type $J$ and $K$ is either empty, or a subcube of type $J\cap K$.

The subcube structure of $\ZZ_2^m$ can be used to define a CSS code in the following way. Let elements of $\ZZ_2^m$ index a set of $2^m$ qubits. Given a subcube $A\subcubeeq\ZZ_2^m$ and a single-qubit unitary $U\in\unitary(2)$, we define a $2^m$-qubit operator $U_A$ that acts as $U$ on qubits in $A$ and as $\eye$ elsewhere,
\begin{equation*}
    (U_A)_x = \begin{cases}
        U, &\text{if } x\in A,\\
        \eye, &\text{otherwise}.
    \end{cases}
\end{equation*}

We now define the $X$ and $Z$ generators of a CSS code using the subcubes within $\ZZ_2^m$. Let $0\leq q\leq r\leq m$ be non-negative integers, and define the following sets:
\begin{equation}
    \begin{aligned}
        S_X&\coloneqq \br{X_A\Bigmid A \text{ is an $(m-q)$-cube}}, 
        \\
        S_Z&\coloneqq \br{Z_A\Bigmid A \text{ is an $(r+1)$-cube}}. 
    \end{aligned}
\end{equation}
In other words, subcubes with dimension $r+1$ give a set of $Z$ operators, and subcubes with \emph{codimension} $q$ give a set of $X$ operators. 

In fact, the geometric picture immediately demonstrates the commutativity of $X$ and $Z$ stabilizers. Suppose we have an $Z_A\in S_Z$ and an $X_B\in S_X$. These operators overlap on the qubits in the set $A\cap B$. If $A\cap B$ is empty then clearly the operators commute. Otherwise, they overlap on a subcube with dimension $\abs{J\cap K}$, where $J$ and $K$ are the types of $A$ and $B$, respectively. Now by construction, $\abs{J}+\abs{K}=m+(r-q)+1\geq m+1$ as $r\geq q$. Since $J$ and $K$ are both subsets of $S$ which has $m$ elements, clearly $\abs{J\cap K}\geq 1$. Since the number of vertices in a subcube with dimension greater than $0$ is always even, $\abs{A\cap B}$ must be even, and so $Z_A$ and $X_B$ commute. As the phase of every operator in $S_X$ and $S_Z$ is $+1$, we have proven the following:

\begin{fact}
$\langle S_X, S_Z\rangle$ defines a valid stabilizer group.
\end{fact}
In \cref{def: QRM}, $S_X$ and $S_Z$ were chosen as the generators for $QRM_m(q,r)$; here we have verified that they do form a valid CSS code.
We note that the subfamily of quantum Reed--Muller codes studied in \cite{rengaswamy2020optimality} in our notation corresponds to the codes $QRM_m(r-1,r)$.

A basis for the space of logical Pauli operators for $QRM_m(q,r)$ can also be stated using the language of subcubes. Consider the following sets:
\begin{equation}
    \begin{aligned}
        L_Z&\coloneqq \Big\{Z_{\standard{J}}&\Bigmid J\subseteq S,\; q+1\leq\abs{J}\leq r\Big\},\\
        L_X'&\coloneqq \Big\{X_{\standard{S\setminus J}}&\Bigmid J\subseteq S,\; q+1\leq\abs{J}\leq r\Big\}.
    \end{aligned}
\end{equation}
We note that, whereas the stabilizers for $QRM_m(q,r)$ were defined using arbitrary subcubes of a particular dimension, the operators in these sets are implemented on \emph{standard} subcubes with a varying dimension.

We will see in \cref{sec: QRM-codes} that $L_Z$ and $L_X'$ generate the set of \emph{logical Pauli operators} for $QRM_m(q,r)$. In fact, $L_Z$ and $L_X'$ form a minimal generating set for this space, so the dimension of $QRM_m(q,r)$ is precisely equal to $\abs{L_Z}=\abs{L_X'}=\sum_{i=q+1}^r \binom{m}{i}$ and the distance of the code is $2^{\min \{q+1,m-r\}}$.

\subsection{Transversal logic on \texorpdfstring{$QRM_m(q,r)$}{QRMm(q,r)}}

While $L_Z$ and $L_X'$ give a basis for logical operators comprised of physical Pauli gates, they do not capture the behavior of non-Pauli operators applied transversally to the code space. We will show that in some sense, these $L_Z, L_X'$ sets can be generalized to 
higher-level Clifford operators. In particular, we examine the question of whether or not transversal operators on subcubes using non-Pauli gates can also implement logical operations on $QRM_m(q,r)$.

To examine the realizable logical operators, we define logical qubits in terms of their $X$ and $Z$ logical operators. For example, 
unencoded qubits are determined by the $n$-qubit Pauli group, which is generated by the weight-1 operators $\br{X_i,Z_i}$ together with phases, where $X_i$ and $Z_i$ act as $X$ and $Z$, respectively, on only the $i$-th qubit. Importantly, the set $\br{X_i,Z_i}_{i\in[n]}$ has the property that it is \emph{symplectic}, i.e., $X_i$ and $Z_j$ anti-commute if and only if $i= j$. Likewise, to determine the logic performed on a quantum code space one must first detail a \emph{symplectic basis} of logical Pauli errors. In short, a symplectic basis guarantees that the logical qubits of a code can be controlled independently of each other. We will go into the details of symplectic bases and, in particular, the symplectic basis for $QRM_m(q,r)$ further in \cref{sec: CH and errors} and \cref{sec: QRM-codes}. For now, we will simply define the logical Pauli operators.

In the previous section, we introduced the sets $L_Z$ and $L_X'$ and stated that they generate the group of logical Pauli operators on $QRM_m(q,r)$. Unfortunately, the set $\br{L_Z, L_X'}$ \emph{does not} satisfy the symplectic condition needed to define logical qubits. Instead, we will slightly modify the set $L_X'$ by shifting these operators to act on non-standard subcubes instead of standard ones. Note that the sets $L_Z$ and $L_X'$ were indexed by subsets of generators $J\subseteq S$, obeying the condition $q+1\leq\abs{J} \leq r$. We will use such subsets to index the logical qubits of $QRM_m(q,r)$ in the following way: 

\begin{restatable}[Index set for the logical qubits]{definition}{QRMdef}\label{def: QRM logical labels}
Consider the quantum code $QRM_m(q,r)$. The collection of subsets  $\mcQ\coloneqq\br{J\subseteq S\mid q+1\leq \abs{J}\leq r}$ is called the \emph{index set for logical qubits of $QRM_m(q,r)$}. For a subset $J\subseteq S$, we use the shorthand $e_J\coloneqq \sum_{i\in J} e_i\in\ZZ_2^m$ to denote the incidence bit string of length $m$ corresponding to $J$. For $J\in\mcQ$, the $J$\emph{-th qubit\footnote{We emphasize a possible point of confusion: the index set for logical qubits $\mcQ$ is, itself, a collection of subsets $J\subseteq S$.} of $QRM_m(q,r)$} is defined via the logical Pauli operators $\overline{Z}_J\coloneqq Z_{\standard{J}}$ and $\overline{X}_J\coloneqq X_{e_J+\standard{S\setminus J}}$. 
\end{restatable}
In other words, the logical Pauli operators for $QRM_m(q,r)$ are generated by the sets
\begin{equation}
    \begin{aligned}
        L_Z&\coloneqq \big\{Z_{\standard{J}}\Bigmid J\in\mcQ \big\},\\
        L_X&\coloneqq \big\{X_{e_J+\standard{S\setminus J}}\Bigmid J\in \mcQ\big\}.
    \end{aligned}
\end{equation}
We prove in \cref{lem: standard basis for QRM} that $\br{L_Z,L_X}$ is a symplectic basis for the space of logical Pauli operators.
See \cref{fig: m4 code example} in the end of Part I for a visualization of the logical Pauli basis in the case of the quantum RM code $QRM_4(0,2)$.

Our goal now is to construct transversal operators for $QRM_m(q,r)$ that lie outside of the Pauli group, which have a geometric structure similar to $L_Z$ and $L_X$. We will consider the following single-qubit $Z$ and $X$ rotation gates:
\begin{align}
    \Z{k}&\coloneqq \ketbra{0}+e^{i\frac{\pi}{2^{k}}}\ketbra{1},
    & \X{k}&\coloneqq \ketbra{+}+e^{i\frac{\pi}{2^{k}}}\ketbra{-},\\
    & = \begin{bmatrix}
        1 & 0 \\
        0 & e^{i\frac{\pi}{2^{k}}}
    \end{bmatrix},
    &&=\frac{1}{2}\begin{bmatrix}
        1+e^{i\frac{\pi}{2^{k}}} & 1-e^{i\frac{\pi}{2^{k}}} \\
        1-e^{i\frac{\pi}{2^{k}}} & 1+e^{i\frac{\pi}{2^{k}}}
    \end{bmatrix}.
\end{align}
The $\Z{k}$ operators are defined so that they reproduce the natural $k$-th level Clifford Hierarchy single-qubit $Z$ basis gates: $Z(-1)=\eye$, the identity, $Z(0)=Z$, the Pauli $Z$ operator, $Z(1)=\phasegate=\sqrt{Z}$, the phase gate, $Z(2)=T=\sqrt{\phasegate}$, the $T$ gate, etc.  (see \cref{sec: CH and errors}). Operators
$\Z{k}$ and $\X{k}$ are related to each other via the Hadamard matrix, $\had \Z{k}\had = \X{k}$. Note that for $\ell\in\br{0,\dots,k+1}$, $\Z{k}^{2^\ell}=\Z{k-\ell}$ and $\X{k}^{2^\ell}=\X{k-\ell}$, implying that $\Z{k}^{2^{k+1}}=\X{k}^{2^{k+1}}=\eye$, i.e., 
they have order $2^{k+1}$.

As before, suppose that $2^m$ physical qubits are indexed by the elements of $\ZZ_2^m$, and for a subcube $A\subcubeeq\ZZ_2^m$, define the transversal operator $\Z{k}_A$ via
\begin{align}
    \left(\Z{k}_A\right)_x&\coloneqq \begin{cases}
        \Z{k}, &\text{ if } x\in A, \\
        \eye, &\text{ otherwise}.
        \end{cases}
\end{align}

The {\em main results} of the present work are (1) necessary and sufficient conditions on $k$, $A$ that determine when $\Z{k}_A$ is an undetectable error for $QRM_m(q,r)$ and (2) explicit descriptions of the logic implemented by $\Z{k}_A$\footnote{We do not explicitly consider the case of $\X{k}_A$ operators in this work, as applying a global Hadamard transform to the code swaps the role of the $X$ and $Z$ operators, $\had^{\otimes 2^m}QRM_m(q,r)\had^{\otimes 2^m}=QRM_m(m-r-1,m-q-1)$. This implies that every result we prove for the $\Z{k}_A$ operators will be true for the $\X{k}_A$ operators by replacing $r\mapsto m-q-1$ and $q\mapsto m-r-1$.}. We therefore can implement higher level Clifford operators via higher level physical gates, and provide a geometric characterization of the realized logic.

In \cref{sec: sufficiency conditions} we prove the following:
\begin{theorem*}[\cref{thm: subcube dimension implies logic} for $\Z{k}_A$]
    Let $0\leq q\leq r\leq m$ be non-negative integers and consider the quantum Reed--Muller code $QRM_m(q,r)$. Suppose $A$ is a subcube of the $m$-dimensional hypercube. In order for $\Z{k}_A$ to preserve the code space it must be true that $m\geq q +kr+1$. Further,
    \begin{enumerate}
        \item $\Z{k}_A\in\mcS^*$ if and only if $\dim A\geq (k+1)r+1$.
        \item $\Z{k}_A\in\mcE^*$ if and only if $q+kr+1\leq \dim A\leq (k+1)r$.
    \end{enumerate}
\end{theorem*}
In other words, whenever the dimension of $A$ is large enough ($\geq q+kr+1$) the operator $\Z{k}_A$ necessarily preserves the code space of $QRM_m(q,r)$. Furthermore, when the dimension of $A$ is even larger ($\geq (k+1)r+1$), $\Z{k}_A$ necessarily leaves each state of $QRM_m(q,r)$ invariant. The proof of \cref{thm: subcube dimension implies logic} is rather straightforward: by utilizing the recursive structure of the Clifford Hierarchy along with the geometric properties of the hypercube, we inductively prove the result for increasing values of $k$.

To study the logic realized by $Z(k)_{A}$, we need to generalize our notion of the index set for logical qubits. In particular, higher-level Clifford operators can be entangling, e.g. $CZ,\; CCZ,$ etc. Thus, our index set must now consider subcubes that span multiple logical operators.  

We will begin by considering the subsets of generators $ K\subseteq S$ that satisfy $q+kr+1\leq \abs{K}\leq (k+1)r$. When $k = 0$, this is precisely the index set of single qubit logical operators. We now define index sets for larger $k$, which will give a basis for the logical operators that we consider in this paper.

\begin{restatable}[Index set for the $k$-th level logicals; generalizes \cref{def: QRM logical labels}]{definition}{kgatesets}\label{def: k gate sets}
    For $k\in\ZZ_{\geq 0}$,  the \emph{index set for the $k$-th level logical operators of} $QRM_m(q,r)$, denoted by $\mcQ_k\subseteq \powerset{S}$, is given by the following collection of subsets of generators (which implicitly depends on the choices of $q$ and $r$:
    \begin{equation}
        \mcQ_k\coloneqq \br{K\subseteq S\Bigmid q+kr+1\leq \abs{K}\leq (k+1)r}.
    \end{equation}
\end{restatable}
That is, $K\in\mcQ_k$ implies that $\Z{k}_{\standard{K}}$ acts on a subcube with dimension large enough to preserve stabilizers but not so large as to realize trivial logic.

Consider the subsets of integers $\{q+kr+1,\dots,(k+1)r\}\subset \NN$ for $k\in\ZZ_{\geq 0}$, which are used to define the $\mcQ_k$ collections. These subsets are pairwise disjoint, each contains $r-q$ numbers, and they are separated from each other by $q$ numbers. They partition $\NN$ only in the case that $q=0$, otherwise there are integers that do not fall in $\{q+kr+1,\dots,(k+1)r\}$ for any $k\in\ZZ_{\geq 0}$. These facts have a number of simple implications on the collections $\mcQ_k$:
\begin{enumerate}
    \item The $\mcQ_k$ are all disjoint: If $k\neq \ell$ then $\mcQ_k\cap\mcQ_\ell=\emptyset$. Thus, if $K\in\mcQ_k$ then value of $k$ is unique for $K$.\\[-1.5em]
    \item If $k<\ell$, $K\in\mcQ_k$, and $L\in\mcQ_\ell$, then $\abs{K}+q<\abs{L}$.\\[-1.5em]
    \item When $q\geq 1$, there are $K\subseteq S$ for which $K\notin\mcQ_k$ for any $k\in\ZZ_{\geq 0}$.
\end{enumerate}
\cref{fig: q=0 r=2 admissible dimensions,fig: q=1 r=2 admissible dimensions} give visualizations for dimensions and subcube operators which support trivial, non-trivial, or no logic for two different quantum RM codes.

\begin{figure}[bh!]
    \centering
    \includegraphics[width=\linewidth]{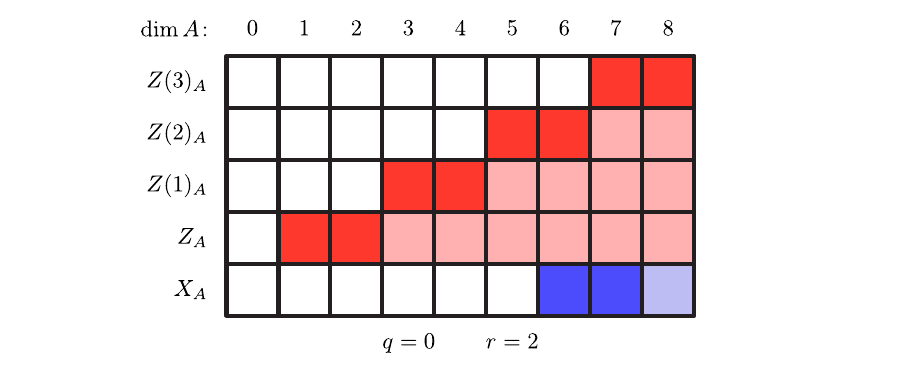}
    \caption{For $QRM_m(q,r)$ the dimension of a subcube $A\subcubeeq\ZZ_2^m$
     determines whether $\Z{k}_A$: (1) acts trivially on the code, (2) acts as a logical operator on the code, or (3) does not preserve the code space.\\ 
In the table above we consider the case of $QRM_8(0,2)$. 
 The columns are indexed by the dimension of the subcubes and the rows correspond to various 
    operators including $X_A$, $Z_A$ and a higher-level diagonal Cliffords. A light-colored box indicates that the given operator on a subcube of the given dimension acts trivially on the code. A dark box indicates a dimension where the subcube operator performs logic on the code. A white box indicates that the operator does not preserve the code space.\\
    We note that, as $q=0$, every dimension greater than zero admits a logical operator in some level of the Clifford Hierarchy, \emph{cf.} \cref{fig: q=1 r=2 admissible dimensions}.}
    \label{fig: q=0 r=2 admissible dimensions}
\end{figure}
\begin{figure}[ht!]
    \centering
    \includegraphics[width=\linewidth]{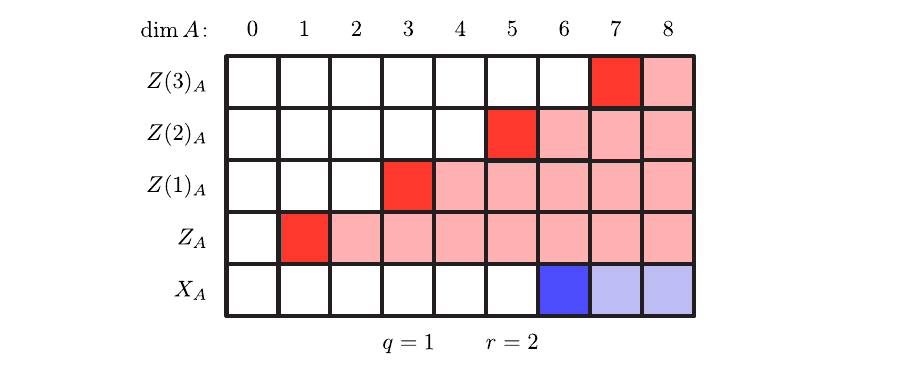}
    \caption{Admissible dimensions for subcube operators on $QRM_8(1,2)$. Note that, as $q\geq 1$, there are dimensions that do not support logical subcube operators in any dimension of the Clifford Hierarchy, e.g., if $A$ is a square, \emph{cf.} \cref{fig: q=0 r=2 admissible dimensions}.}
    \label{fig: q=1 r=2 admissible dimensions}
\end{figure}

For each $k\in\ZZ_{\geq 0}$ consider the following set of diagonal and transversal operators acting on the $m$-dimensional hypercube:
\begin{equation}
    \br{\Z{k}_{\standard{K}} \Bigmid K\in\mcQ_k}.
\end{equation}
Operators in this set are all \emph{standard subcube operators}, i.e., they act on standard subcubes of the hypercube.
For $k=0$, this set is precisely the basis of $Z$ logical Pauli operators for $QRM_m(q,r)$. Further, by \cref{thm: subcube dimension implies logic} each operator in this set is a non-trivial logical operator for $QRM_m(q,r)$.

The majority of the our paper is spent on determining precisely \emph{what} logic is performed by the $\Z{k}_{A}$ operators for arbitrary $k\geq 0$. Instead of considering arbitrary subcubes $A\subcubeeq\ZZ_2^m$, it turns out that we only need to consider the standard subcube operators; we will prove in \cref{sec: basis for k-th level} that every $\Z{k}_{A}$ operator can be decomposed as a product of operators $\Z{k'}_{\standard{K}}$ for various $k'\leq k$ and $K\subseteq S$. For the remainder of this section we will restrict ourselves to the case of standard subcube operators.

In the same way that $\Z{0}_{\standard{J}}$ operators (and their products) perform circuits of logical $Z$ operators, an operator $\Z{k}_{\standard{K}}\in\mcE^*$ acts as a circuit of logical \emph{multi-controlled-$Z$} operations, where the number of logical qubits each gate in the circuit can interact with is at most $k+1$. Throughout the text we include bars over the following operators to emphasize that a \emph{physical} implementation of $\Z{k}$ operators will ultimately realize a \emph{logical} implementation of a multi-controlled-$Z$ circuit.

\begin{definition}[$C^{(\ell)} Z$ gates]
    For $\ell\in\NN$, the logical \emph{multi-controlled-$Z$} gate is defined recursively as the $(\ell+1)$-qubit unitary operator $\overline{C^{(\ell)}Z}\coloneqq \overline{\ketbra{0}}\otimes \overline{\eye}+\overline{\ketbra{1}}\otimes \overline{C^{(\ell-1)}Z}$, where $\overline{C^{(0)}Z}\coloneqq \overline{Z}$. 
\end{definition}
The operator $\overline{C^{(\ell)}Z}$ is symmetric in the $\ell$ qubits; in particular, $\overline{C^{(\ell)}Z}$ is a diagonal gate that introduces a $-1$ phase to the all-ones computational basis state, $\overline{\ket{1^\ell}}$, and acts as identity on all other computational basis states. 

The next group of definitions introduce an important set of concepts for our results on transversal logic.

\begin{definition}[Set-controlled $Z$ gates]
    Given a collection of logical qubits $\mcJ\subseteq\mcQ$, $\abs{\mcJ}=\ell$, define the logical \emph{$\mcJ$-controlled-$Z$}, $\overline{C^\mcJ Z}$, as the $\abs{\mcQ}$-qubit gate that acts as $\overline{C^{(\ell)}Z}$ on the qubits in $\mcJ$ and identity elsewhere, $\overline{C^\mcJ Z}\coloneqq \overline{C^{(\ell)}Z}\vert_\mcJ\otimes \overline{\eye}\vert_{\mcQ\setminus \mcJ}$. By convention, $\overline{C^\emptyset Z}\coloneqq \overline{\eye}$ is the logical identity.
\end{definition}

\begin{definition}[Controlled $Z$ gates over collections of sets of logical qubits]
    Given a collection of sets of logical qubits $\mcF\subseteq \powerset{\mcQ}$, define the logical \emph{$\mcF$-controlled-$Z$} operator as the circuit consisting of $\overline{C^\mcJ Z}$ operators for each $\mcJ\in\mcF$, $\overline{C^\mcF Z}\coloneqq\prod_{\mcJ\in\mcF}\overline{C^\mcJ Z}$.
\end{definition}


As mentioned in \cref{sec:intro}, \cref{thm: subcube dimension implies logic} implies that the logical circuit implemented by $\Z{k}_{\standard{K}}$, $K\in\mcQ_k$, will be a circuit composed of multi-controlled-$Z$ operators. Given a $K\in\mcQ_k$, we will now define such a collection that will, in many cases, correctly determine the corresponding logical circuit for $\Z{k}_{\standard{K}}$. Recall that a set of logical qubits $\mcJ$ is represented by a 
collection of subsets of generators $J\subseteq S$. 

\begin{restatable}[Minimal covers for logical index sets] {definition}{minimalcovers}\label{def: minimal covers}
    Suppose that $K\in\mcQ_k$.
    A set of logical qubits $\mcJ\subseteq\mcQ$ is said to form a \emph{$\mcQ$-minimal cover for $K$}, or simply a \emph{minimal cover for $K$}, if (1) $\mcJ$ is a cover of $K$, i.e., $\bigcup_{J\in\mcJ}J=K$, and (2) the number of qubits in $\mcJ$ is $\abs{\mcJ}=k+1$. That is, $\mathcal{J}$ is a $\mathcal{Q}$-minimal cover for $K$ if all of its unique generators are exactly the generators of $K$.  
    
    Since $\abs{J}\leq r$ for each $J\in\mcQ$ and $\abs{K}\geq q+kr+1$ by \cref{def: k gate sets}, $k+1$ is the smallest possible number of sets from $\mcQ$ that cover $K$, hence the ``minimal'' designation. 
    
    Let $\mcF(K)\subseteq\powerset{\mcQ}$ denote the collection of all minimal covers for $K$,
    \begin{equation*}
        \mcF(K) \coloneqq \br{\mcJ\subseteq \mcQ \Bigmid \mcJ \text{ is a minimal cover for $K$}}.
    \end{equation*}
\end{restatable}

\begin{remark}
    The ``minimality'' condition in the above definition is necessary to ensure that the multi-controlled-$Z$ gates of a circuit corresponding to $\mcF(K)$ will act on precisely $k+1$ qubits, which is, itself, necessary to ensure the implemented logical circuit lies in the $k$-th level of the Clifford Hierarchy (see \cref{sec: CH and errors} and \cref{sec:CZ circuits}). The motivation for the ``cover'' property is more subtle: $\mcJ$ is a cover of $K$ if and only if $\Z{k}_{\standard{K}}$ jointly overlaps with all of the logical $\overline{X}_J$ operators, $J\in\mcJ$, on an \emph{odd} number of qubits (e.g., see \cref{fig: intersecting cubes}). In the special case that $k=0$ this reflects the anti-commutativity of the $\overline{Z}_J=\Z{0}_{\standard{J}}$ and $\overline{X}_J$ operators. Informally, the cover property of $\mcJ$ is indicative of a fundamental overlap requirement for a transversal $\Z{k}$ operator to implement a logical multi-controlled-$Z$ gate on a set of qubits. Perhaps an odd overlap on the joint support of $X$ logicals is part of a more general phenomenon related to the classification of the circuits implemented by transversal $\Z{k}$ operators; we leave this question, which may be related to the triorthogonality property \cite{Bravyj2012magic}, for future work.
\end{remark}

\begin{figure}[t!]
    \centering
    \includegraphics{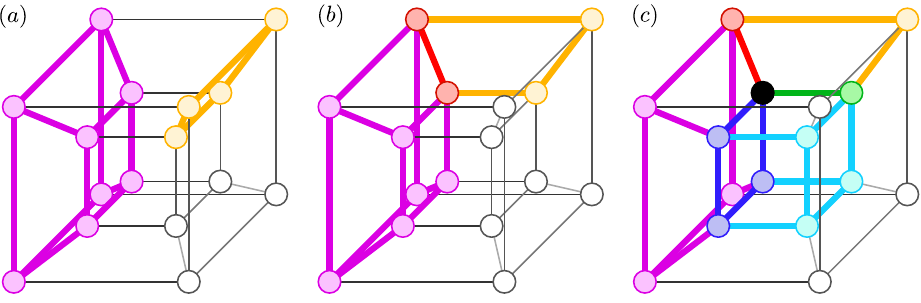}
    \caption{Consider the code $QRM_4(0,2)$ and the standard subcube operator $\phasegate_{\standard{2,3,4}}$. To motivate the utility of the ``cover'' property for a collection of logical qubits $\mcJ\subseteq\mcQ$, we consider the intersection of the standard subcube $\standard{2,3,4}$ (represented as a magenta cube in each subfigure) with various collections. \\
    (a) In this case $\mcJ=\br{\br{1,2}}$ is not a cover for $\br{2,3,4}$. The logical $X$ operator it corresponds to acts on the \emph{non-standard} subcube $1100+\standard{3,4}$ (represented by an orange square). Clearly this subcube does not intersect $\standard{2,3,4}$, so subcube operators acting on them commute.\\
    (b) In this case $\mcJ=\br{\br{2,3}}$ is not a cover for $\br{2,3,4}$. The logical $X$ operator it corresponds to acts on the \emph{non-standard} subcube $0110+\standard{1,4}$ (represented by an orange square). By construction this subcube intersects $\standard{2,3,4}$ on an even number of qubits (represented as red dots).\\
    (c) In this case $\mcJ=\br{\br{2,3},\br{4}}$ \emph{is} a cover for $\br{2,3,4}$. It corresponds to two logical $X$ operators acting on \emph{non-standard} subcubes: $0110+\standard{1,4}$ (orange square) and $0001+\standard{1,2,3}$ (cyan cube). The joint intersection of these subcubes with the standard subcube $\br{2,3,4}$ is a \emph{single} qubit (black vertex), $0111$.}
    \label{fig: intersecting cubes}
\end{figure}

When $k=0$, any $\mcQ$-minimal cover for $J\in\mcQ$ necessarily contains a single element from $\mcQ$, which by definition must be $J$ itself. So, the set of all minimal covers for $J\in\mcJ$ is simply $\mcF(J)=\br{\br{J}}$. Any theorem describing the logical circuit for $\Z{k}_{\standard{K}}$, $K\in\mcQ_k$ must necessarily reduce to $\Z{0}_{\standard{J}}=\overline{Z}_J$ in the case that $J\in\mcQ$. We note that, at least in this simple case of $k=0$, it is trivial that $\Z{0}_{\standard{J}}=\overline{C^{\mcF(J)}Z}$ for $J\in\mcQ$; this fact will hold for more general $k\geq 0$, at least in the case that $q\geq 1$.

We are now ready to state our main theorem, proven in \cref{sec: unsigned logic}, on the transversal application of $\Z{k}$ to a subcube.
We suppose that $q<r$, as otherwise $QRM_m(q,r)$ encodes no logical qubits.
\begin{theorem*}[Description of $Z(k)_{\langle K \rangle}$ logic; informal version of \cref{thm: Zk and tildeZk equivalence conditions}]
    Consider the code $QRM_m(q,r)$ and suppose $q\geq 1$. For every $K\in\mcQ_k$, the operator  $\Z{k}_{\standard{K}}$ implements the logical multi-controlled-$Z$ circuit corresponding to the collection of minimal covers of $K$:
    \begin{equation}
        \Z{k}_{\standard{K}} \equiv \overline{C^{\mcF(K)}Z}.
    \end{equation}
\end{theorem*}
When $q=0$, the transversal $\Z{k}$ operators acting on standard subcubes can correspond to logical multi-controlled-$Z$ circuits that are defined by collections more general than a set of minimal covers. As an example, for the hypercube code family we will prove the following:
\newpage
\begin{lemma*}[\cref{lem: hypercube codes unsigned logic}]
    Consider the $[[2^m,m,2]]$ codes $QRM_m(0,1)$. For every $K\subseteq S$,
    \begin{equation}
        \Z{\abs{K}-1}_{\standard{K}} \equiv \overline{C^{\powerset{K}}Z},
    \end{equation}
    where $\powerset{K}$ is the power set of $K$.
\end{lemma*}
In this case a $\Z{k}_{\standard{K}}$ operator still acts as a multi-controlled-$Z$ circuit, but the logical circuit it implements no longer has the property that it consists solely of $k$-qubit multi-controlled-$Z$ gates as in \cref{thm: Zk and tildeZk equivalence conditions}. In \cref{sec: unsigned logic} we will fully detail the logical circuits implemented by $\Z{k}_{\standard{K}}$ operators when $q=0$ and $r>q$ is arbitrary.

To characterize logical circuits implemented by $\Z{k}_{\standard{K}}$ operators,
we proceed in a somewhat indirect way, by instead considering so-called ``signed'' versions of these operators. We detail this through the next section.

\subsection{The signed operator case}
Readers familiar with the hypercube code family may find the choice to consider transversal $\Z{k}$ operators 
with no adjoints unusual. Consider, for example, the well-known $[[8,3,2]]$ \eczoo[code]{stab_8_3_2}, which in our notation is the code $QRM_3(0,1)$. The physical qubits of this code live on the vertices of a 3-dimensional cube, $Z$ stabilizers are given by faces of the cube, and there is a single $X$ stabilizer that acts on every qubit. It is known that applying the operator $\widetilde{T}$, defined as a $T$ gate on vertices with an even-weight index and $T^\dagger$ on vertices with an odd-weight index, will apply a logical $\overline{CCZ}$ to the 3 encoded qubits \cite{Campbellcolor}.

Now, in our notation the physical qubits are indexed by $\ZZ_2^3$, which is generated by $S\coloneqq\br{e_1,e_2,e_3}$, and logical qubits are indexed by single-element sets of generators, $\overline{1}\coloneqq\br{e_1}$, $ \overline{2}\coloneqq\br{e_2}$, and $ \overline{3}\coloneqq\br{e_3}$. The operator $\widetilde{T}$ is applied to the standard subcube $\standard{S}=\ZZ_2^3$--- the entire cube--- and it implements logical $\overline{C^{\br{\overline{1},\overline{2},\overline{3}}}Z}$. We observe the following:
\begin{itemize}
    \item $S\in\mcQ_2$: $0+2\cdot 1+1\leq \abs{K}=3\leq (2+1)\cdot 1$.
    \item $\mcF(S)=\br{\br{\overline{1},\overline{2},\overline{3}}}$: As $S\in\mcQ_2$, a minimal cover for $S$ must contain the index sets of 3 qubits, and the only choice of three logical qubits is all of $S$.
\end{itemize}
Thus, the 8-qubit operator $\widetilde{T}$ implements a logical $\overline{C^{\mcF(S)}}$.

We see in this example that, at least for $QRM_3(0,1)$, the ``signed'' operator, $\widetilde{T}$, correctly implements the logical circuit defined by the set of minimal covers for $S$. In an analogous way as $\widetilde{T}$, one can define signed phase operators, $\tildephase_{\standard{e_i,e_j}}$, acting on squares of the cube, and these operators are known to implement logical controlled-$Z$ between the qubits defined by the edges of the square. This same idea generalizes to the entire $[[2^m,m,2]]$ hypercube code family, and to signed versions of global $\Z{k}$ operators acting on subcubes. In summary, the operators that act as $\Z{k}$ on even-index qubits of the hypercube and $\Z{k}^\dagger$ on the odd-index qubits appear to form a 
fundamental operator set for $QRM_m(0,1)$. It turns out that this is true more generally for \emph{all} $QRM_m(q,r)$ codes.

Elaborating on this line of thought, consider the natural 2-coloring on the vertices of the $m$-dimensional hypercube, $\ZZ_2^m$, given by the parity of the Hamming weight of a vertex. Using this coloring, the \emph{signed $\Z{k}$ operator} acting on a subcube $A\subcubeeq\ZZ_2^m$, $\tildeZ{k}_A$, is defined by
\begin{equation}
    \left(\tildeZ{k}_A\right)_x\coloneqq \begin{cases}
        \Z{k}, &\text{ if $x\in A$ and $\abs{x}$ is even} \\
        \Z{k}^\dagger, &\text{ if $x\in A$ and $\abs{x}$ is odd} \\
        \eye, &\text{ otherwise}.
    \end{cases}
\end{equation}
That is, $\tildeZ{k}_A$ acts on $\Z{k}$ on the qubits in $A$ that have even Hamming weight and as the inverse of $\Z{k}$ on the odd Hamming weight qubits.

In \cref{sec: unsigned dimension conditions} we prove the following in precisely the same way as the case of $\Z{k}_A$:
\begin{theorem*}[\cref{thm: subcube dimension implies logic} for $\tildeZ{k}_A$]
    Let $0\leq q\leq r\leq m$ be non-negative integers and consider the quantum Reed--Muller code $QRM_m(q,r)$. Suppose $A$ is a subcube of the $m$-dimensional hypercube. In order for $\tildeZ{k}_A$ to preserve the code space it must be true that $m\geq q +kr+1$. Further,
    \begin{enumerate}
        \item $\tildeZ{k}_A\in\mcS^*$ if and only if $\dim A\geq (k+1)r+1$.
        \item $\tildeZ{k}_A\in\mcE^*$ if and only if $q+kr+1\leq \dim A\leq (k+1)r$.
    \end{enumerate}
\end{theorem*}

Our main result, concerning the logic implement by signed subcube operators, is the following:

\begin{theorem*}[\cref{thm: logical multi-controlled-Z circuit}]
    For every $K\in\mcQ_k$, $\tildeZ{k}_{\standard{K}}$ implements the logical multi-controlled-$Z$ circuit corresponding to the collection of minimal covers of $K$:
    \begin{equation}
        \tildeZ{k}_{\standard{K}} \equiv \overline{C^{\mcF(K)}Z}.
    \end{equation}
\end{theorem*}
\cref{thm: logical multi-controlled-Z circuit} is proven by showing that the physical operator, $\tildeZ{k}_{\standard{K}}$, conjugates the physical implementations of the logical Pauli operators of $QRM_m(q,r)$ in precisely the same way as the logical $\overline{C^{\mcF(K)}Z}$ circuit. This proof is given in \cref{sec: signed logic}.

We give the logical circuits for \emph{unsigned} operators through the use of an operator decomposition theorem. In particular, we first show how to decompose a standard {unsigned} operator, $\Z{k}_{\standard{K}}$ into a product of standard \emph{signed} operators, $\tildeZ{k'}_{\standard{K'}}$.
The fact that unsigned operators no longer implement logical circuits corresponding to collection of minimal covers when $q=0$ is ultimately a byproduct of the structure of the logical index sets $\mcQ_k$, specifically, that the sets $\br{kr+1,\dots,(k+1)r}$ partition $\NN$. Our main result on the logic of unsigned operators is more formally stated as the following:

\begin{theorem*}[Description of $Z(k)_{\langle K \rangle}$ logic; \cref{thm: Zk and tildeZk equivalence conditions}]
    Let $K\in\mcQ_k$. If $q\geq 1$ then the unsigned operator $\Z{k}_{\standard{K}}$ implements the same logical circuit on $QRM_m(q,r)$ as its corresponding \emph{signed} version, $\tildeZ{k}_{\standard{K}}$
\end{theorem*}

In particular, when $q\geq 1$, it turns out that every signed operator in the decomposition of $\Z{k}_{\standard{K}}$ acts trivially on $QRM_m(q,r)$ by \cref{thm: subcube dimension implies logic}, except for a single $\tildeZ{k}_{\standard{K}}$ term. Hence, when $q\geq 1$
we conclude that $\Z{k}_{\standard{K}}$ and $\tildeZ{k}_{\standard{K}}$ perform the same logical operation on $QRM_m(q,r)$, yielding \cref{thm: Zk and tildeZk equivalence conditions}. In \cref{sec: unsigned logic} we prove \cref{thm: Zk and tildeZk equivalence conditions} and \cref{lem: hypercube codes unsigned logic}, along with proving which logical circuits are implemented by unsigned subcube operators for $QRM_m(0,r)$ when $r\geq 2$.


\subsection{Discussion and future directions}
Here we summarize some remarks and avenues for possible extensions prompted by our results.


\paragraph{Subcube operators in the $X$ basis.}

As mentioned previously, our results for the $\Z{k}$ subcube operators translate directly to the case of $\X{k}$ operators. For instance, we have the following analogous version of \cref{thm: subcube dimension implies logic} in the $X$ basis:
\begin{restatable}{theorem}{Xbasis}\label{clm: subcube dimension implies logic (X basis)}
    Suppose $A\subcubeeq\ZZ_2^m$ is a subcube of $\ZZ_2^m$.
    \begin{enumerate}
        \item $\X{k}_{A}\in\mcS^*$ if and only if $\dim A\geq (k+1)(m-q-1)+1$.
        \item $\X{k}_{A}\in\mcE^*$ if and only if $m-r+k(m-q-1)\leq \dim A\leq (k+1)(m-q-1)$.
    \end{enumerate}
\end{restatable}
As detailed in \cref{app: X basis operators}, the bounds for non-trivial logic given in \cref{thm: subcube dimension implies logic} and \cref{clm: subcube dimension implies logic (X basis)} are often incompatible with each other. In particular:
\begin{enumerate}
    \item For any $A,B\subcubeeq \ZZ_2^m$, if a quantum Reed--Muller codes supports transversal $\Z{k}_A\in\mcN^*$ for some $k\geq 2$, then it cannot support transversal $\X{k}_B\in\mcN^*$ for \emph{any} $k\geq 0$, and vice versa.
    \item Only when $m=q+r+1$, can the codes $QRM_{q+r+1}(q,r)$ simultaneously support \emph{global} transversal $\Z{1}$ and $\X{1}$ operators. They \emph{cannot} support $\Z{k}_A\in\mcN^*$ or $\X{k}_A\in\mcN^*$ for any value of $k\geq 2$. Interestingly, however, a global Hadamard operator implements non-trivial logic in $QRM_{q+r+1}(q,r)$.
\end{enumerate}

\paragraph{Diagonal and transversal operators in the Clifford Hierarchy.}
The physical operators we consider here all share a common structure: they are (1) diagonal, (2) transversal, and (3) lie in the Clifford Hierarchy. Denoting the group of all unitary operators satisfying these three conditions by $\dtc\leq \unitary(2)^{\otimes 2^m}$, one implication of the Validity Theorem is the following:
\begin{equation}\label{eq: validity result}
    \left\langle e^{i\theta} \tildeZ{k}_{\standard{K}} \Bigmid k\in\ZZ_{\geq 0},\; K\in\mcQ_k,\; \theta\in[0,2\pi)\right\rangle \subseteq \mcN^*\cap \dtc.
\end{equation}
That is, the group generated by the standard subcube operators for $k\geq0$ and $K\in\mcQ_k$ is a group of undetectable errors for $QRM_m(q,r)$, all with the property that they lie in the group $\dtc$. It is natural to wonder whether or not the converse is true: are there $\dtc$ operators that preserve the code space of $QRM_m(q,r)$, but that \emph{cannot} be produced via products of the basis subcube operators indexed by the $\mcQ_k$ collections?

We have proven through the Validity Theorem that the converse \emph{does} hold for subcube operators. In particular, given a subcube $A\subcubeeq\ZZ_2^m$, an operator $\tildeZ{k}_A\in\mcN^*$ must necessarily have a decomposition into standard subcube operators indexed by the $\mcQ_k$. Additionally, the CSS construction provides a converse statement when $k=0$: if $Z_M\in\mcN^*$ is a $Z$ operator acting on an arbitrary \emph{subset} $M\subseteq \ZZ_2^m$ then it necessarily can be decomposed as a product of $Z_{\standard{J}}$ operators for $J\in\mcQ$. This is, in fact, the statement that the group of undetecjktable Pauli $Z$ errors for $QRM_m(q,r)$ is isomorphic to the \emph{classical} Reed--Muller code of order $m-q-1$, $RM(m-q-1,m)$. 

In \cref{app: dtc} we formulate a possible converse to \cref{eq: validity result} in the language of linear codes over \emph{rings} instead of fields. For the ring $R_k\coloneqq \ZZ_{2^{k+1}}$, we construct a family of \emph{generalized Reed--Muller codes}\footnote{Generalizations of Reed--Muller codes to ring alphabets
have been previously studied within the framework of finite ring extensions and Galois rings \cite{bhaintwal2010generalized}. We believe that the codes we define here are different from the code families considered in the literature.}  as submodules of $R_k^{2^m}$ by drawing inspiration from the geometric construction of RM codes used throughout our paper. We detail a possible characterization of $\mcN^*\cap\dtc$ in terms of these generalized RM codes over $R_k$, though we leave the study of this characterization for future work.

Beyond $\dtc$ operators, one can also consider the space of diagonal operators in the Clifford Hierarchy, fully classified in \cite{CGK17}. This prompts us to pose the following question: Can the geometric structure of quantum RM codes be used to give necessary and sufficient conditions for when \emph{constant-depth circuits} from the diagonal Clifford Hierarchy perform logic? We have not attempted to answer it in this paper.

\paragraph{Puncturing quantum RM codes.}
Many distillation protocols employ punctured or shortened quantum RM codes. We hope to lift our geometric results to regimes in which codes have been deformed, providing new intuition as to their logical operators.  

Our results demonstrate that quantum RM codes support logical circuits of multi-controlled-$Z$ gates by applying transversal subcube operators. While unitary synthesis via phase polynomials has been explored \cite{campbell2017unified}, compilation to magic $T$ states is prevalent and rigorously studied \cite{beverland2022assessing, kliuchnikov2023shorter, haah2018codes}. 
$T$ (and related $\overline{\Z{k}}$ logicals) can be achieved by painstakingly puncturing coordinates from quantum RM codes. Take, for instance, the $[[2^m-1,1,3]]$ family of \eczoo[simplex codes]{diagonal_clifford}. To construct these codes one first considers the quantum RM code $QRM_m(1,1)$, which encodes no logical qubits, but is nonetheless a CSS code whose $X$ stabilizers are given by the $(m-1)$-cubes in the $m$-dimensional hypercube and whose $Z$ stabilizers are given by the $2$-cubes (squares). Consider however, what happens when we remove a vertex from the hypercube but still define $X$ and $Z$ operators using $(m-1)$-cubes and $2$-cubes (there are now less of each as any subcube that contained the expunged vertex itself was removed). This process has the effect of \emph{shortening} and \emph{puncturing} the chosen RM codes
\begin{equation}\label{eq: punctured QRM code}
    RM(1,m)^{\circ}\subset RM(1,m)^*.
\end{equation}
In the language of coding theory, the \emph{shortened} code $RM(1,m)^\circ$ is obtained by removing the first bit (or any bit) from all codewords of $RM(1,m)$ that have a 0 in that position. The \emph{punctured} code $RM(1,m)^*$ is obtained from $RM(1,m)$ by simply removing the first bit from every codeword.
The shortened code is contained within the punctured code, the quotient space $RM(1,m)^*/RM(1,m)^\circ$ has dimension 1, and the dual picture \emph{also} has the structure of a punctured/shortened RM code:
\begin{equation}
     (RM(1,m)^*)^\perp = RM(m-2,m)^\circ \subset RM(m-2,m)^* = (RM(1,m)^{\circ})^\perp.
\end{equation}
The family $\CSS( RM(m-2,m)^*,  RM(1,m)^* )$ is the $[[2^m-1,1,3]]$ simplex code family, and it is known that the {global} transversal $\Z{k}$ operator \cite{koutsioumpas2022smallestcodetransversalt} implements a \emph{logical} $\overline{\Z{k}}^\dagger$ operator.

To distill solely $T$ gates, prior work has focused on puncturing triorthogonal codes, codes with special symmetries that are closely related to quantum RM codes \cite{campbell2012magic, Bravyj2012magic}. In the case of triorthogonal codes, the parity check generators can be described via characteristic polynomials corresponding with specific RM codes. In this way, all triorthogonal codes and puncturings with $n + k \leq 38$ have been numerically studied \cite{nezami2022classification}.  

In summary, puncturing is not well understood from a theoretical nor practical standpoint. The theoretical intuition for how deformed operators perform is nascent, with early work studying specific deformations \cite{vasmer2022morphing} or exhaustively enumerating valid code instances \cite{rengaswamy2020optimality}. 
Practically, exhaustive enumeration strategies are intractable for larger $n$; this prevents the design of distilleries with optimal rate given specified $n, d$ parameters. We hope that an extension of our formalisms could elucidate the effects of puncturing, enabling the dynamic design of distilleries.





\paragraph{Reducing physical qubit overhead.}
While we have shown that quantum RM codes support non-trivial logical circuits through the physical implementation of subcube operators, a priori the parameters of quantum RM codes are largely impractical for, say, magic-state distillation. For instance, the code length grows exponentially in the dimension $m$ of the hypercube. The maximal level of the Clifford Hierarchy attainable with $QRM_m(q,r)$, $k_{\max}$, must satisfy $q+k_{\max}r+1\leq m$, so the number of physical qubits needed will likewise grow exponentially in $k_{\max}$. A crucial next step is to reduce the physical qubit overhead of codes that support transversal logic in higher levels of the Clifford Hierarchy.

As mentioned previously, puncturing quantum RM codes is one way to do this: each time the code is punctured a single physical qubit is removed. Our geometric construction of quantum RM codes hints at another way to reduce the number of physical qubits. Drawing inspiration from several recent works on asymptotically-good qLPDC codes \cite{panteleev2022asymptotically,leverrier2022quantum} and constructions of some quantum locally-testable codes \cite{leverrier2022towards}, one could consider the \emph{quotient} of the hypercube by the action of a group. As an example, consider the {\em folded cube graph} obtained by identifying every vertex $x\in\ZZ_2^m$ with its opposite vertex, $\bar{x}=x+1^m$ and likewise identifying a subcube $A\subcubeeq\ZZ_2^m$ with its opposite, $\bar A=1^m+A$. This action preserves the commutativity of $X$ and $Z$ operators defined using $(m-q)$-cubes and $(r+1)$-cubes, respectively, and therefore produces a quantum code with $2^m/2$ qubits instead of $2^m$ qubits. A natural question is whether or not these codes support transversal logic in the same way as quantum RM codes. More general group actions can also preserve commutativity of appropriate choices of $X$ and $Z$ subcube operators, while reducing physical qubit overhead by increasing multiplicative factors. 
Can the transversal logic implemented on such codes also be understood using techniques from our work?

\paragraph{The dual view.} 
We have described our results in terms of the $m$-dimensional hypercube and its complex of subcubes, but there is an equivalent description of quantum RM codes in terms of the $m$-dimensional \emph{hyperoctahedral complex}, which is dual to the hypercube construction. In particular, vertices of the hypercube correspond to the facets of the hyperoctahedron. More generally, the hyperoctahedron can be viewed as a \emph{simplicial complex} whose $\ell$-dimensional simplices are in one-to-one correspondence with $(m-\ell-1)$-cubes, as detailed in \cref{app: hyperoctahedral}.

The hypercube codes $QRM_m(0,1)$, when viewed in the hyperoctahedral picture, correspond precisely to a family of ball codes discussed in \cite{vasmer2022morphing}. 
The authors use a geometric color-code perspective on these codes to derive certain properties including the logic of certain signed transversal rotation operators. 
In particular, using a procedure they call `morphing', they are able to construct code families with increasing distance from ball codes; see also \cite{hangleiter_fault-tolerant_2024}.
It is an interesting open question whether or not such gluing procedures exist for general quantum RM codes, for which our geometric perspective may serve as a helpful starting point.

\section{Examples}\label{sec: Examples}
We now explore some of the structure in the logical circuits we have defined for quantum RM codes.
While the language of minimal covers is essentially the simplest general description of the logic implemented by subcube operators, the logic for the $[[2^m,m,2]]$ hypercube codes and the more general $[[2^m,\binom mr, 2^{\min(m-r,r)}]]$ $QRM_m(r-1,r)$ codes can be phrased in a simpler way; we detail this in \cref{example: hypercube,example: r-1 r}, respectively. In \cref{example: general} we will briefly look at the case of general quantum RM codes, including a table of specific codes and their properties, a simple algorithm to compute the collection of minimal covers, and example logical circuits that can be implemented.

To recall some notation, $\ZZ_2^m$ is the Abelian group generated by the set $S\coloneqq\br{e_i}_{i\in[m]}$, where $e_i$ is the length-$m$ bit string with a single $1$ in the $i$-th position. By abuse of notation, we will equate $S$ simply with the integers from 1 up to $m$, $S=[m]$. In other words, for every $i\in[m]$, the statement ``$i\in S$'', should be taken to mean $e_i\in S$. As the logical qubits of a quantum RM code are indexed by subsets of $S$, we find the notation ``$J=\br{1,3,5}$'' more intuitive that the proper notation $J=\br{e_1,e_3,e_5}$. In this way, the standard cube $\standard{1,3,5}$ is the set of length-$m$ bit strings that are supported on the subset $\{1,3,5\}$ of the coordinates.

\subsection{\texorpdfstring{$QRM_m(0,1)$}{QRMm(0,1)}}\label{example: hypercube}
As a simple example, we first consider the hypercube code family, $QRM_m(0,1)$. In this case, the logical qubit set $\mcQ$ is defined as
\begin{equation}
    \mcQ\coloneqq \br{\br{i}\Bigmid i\in S},
\end{equation}
i.e., all single-element subsets of $S$. To simplify notation, we will denote these sets as $\overline{i}\coloneqq\br{i}$, so that the $i$-th logical qubit of $QRM_m(0,1)$ is given by the index $\overline{i}$.

The $k$-th level logical index sets are defined by
\begin{equation}
    \mcQ_k\coloneqq \br{ K\subseteq S \Bigmid \abs{K}=k+1},
\end{equation}
i.e., all $(k+1)$-element subsets of $S$.
Thus, by definition \emph{every} subset of $S$ is in $\mcQ_k$ for some $k\in\ZZ_{\geq 0}$, and so by \cref{thm: subcube dimension implies logic} \emph{every} signed and unsigned standard subcube $\standard{K}$, $K\subseteq S$, gives rise to a logical operator in the $(\abs{K}-1)$-th level of the Clifford Hierarchy. See \cref{fig: hypercube admissible} for a visual representation of this fact.

\begin{figure}[ht]
    \centering
    \includegraphics{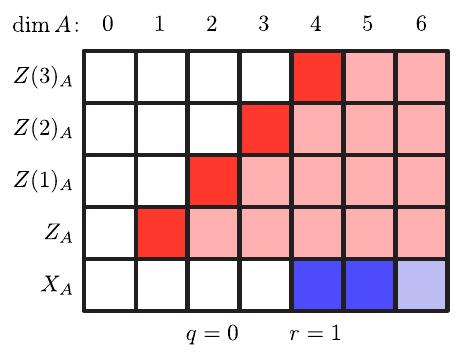}
    \caption{Consider the hypercube code family $QRM_m(0,1)$. In the above figure, the shade of a given box indicates how the given operator for the given dimension will act on the code space:\\
    (1) A dark box indicates logic,\\
    (2) A light box indicates a logical identity, and\\
    (3) A white box indicates the code space is not preserved.}
    \label{fig: hypercube admissible}
\end{figure}

Recall that by \cref{thm: logical multi-controlled-Z circuit}, for $K\in\mcQ_k$ the operator $\tildeZ{k}_{\standard{K}}$ will implement the logical multi-controlled-$Z$ circuit corresponding to the collection of minimal covers of $K$. Given $K\in\mcQ_k$, we proceed to compute this set $\mcF(K)$.

By definition, a subset of qubits $\mcJ\subseteq\mcQ$ is a minimal cover for $K$ if: (1) $\bigcup_{\overline{i}\in\mcJ} \overline{i} = K$, and (2) $\abs{\mcJ}=k+1$. As each logical qubit $\overline{i}$ is a single-element subset, the only collection of logical qubits whose union is all of $K$ is precisely the set $\mcK\coloneqq\br{\overline{i}\bigmid i\in K}$. Thus, we see there is only a single minimal cover for $K$ and that $\mcF(K)=\br{\mcK}$.

Using \cref{thm: logical multi-controlled-Z circuit} we can succinctly detail the logical circuit given by $\tildeZ{\abs{K}-1}_{\standard{K}}$:
\begin{circuit}
    Let $K\subseteq S$ be any subset of $k+1$ generators. Denote the $(k+1)$-element collection of logical qubits $\mcK\coloneqq\br{\overline{i}\bigmid i\in K}$. For a hypercube code, the signed $\Z{k}$ operator applied to the physical qubits in $\standard{K}$ implements a single $(k+1)$-qubit multi-controlled-$Z$ operator on the logical qubits in $\mcK$:
    \begin{equation}
        \tildeZ{k}_{\standard{K}} \equiv \overline{C^\mcK Z}.
    \end{equation}
\end{circuit}
\begin{remark}
    One may be tempted to use the sets $K$ and $\mcK$ interchangeably as they have the same size and each element of $\mcK$ is defined using an element of $K$. We reiterate that logical qubits are necessarily \emph{subsets} of generators, themselves. There is essentially no difference in the case of hypercube codes as the logical qubits are single-element subsets, but for other choices of $q$ and $r$ this distinction is important.
\end{remark}

\cref{lem: hypercube codes unsigned logic} gives us the logical circuit of any unsigned operator on a standard subcube.
\begin{circuit}
    Let $K\subseteq S$ be any subset of $k+1$ generators and let $\mcK\coloneqq\br{\overline{i}\mid i\in K}$. For a hypercube code, the transversal $\Z{k}$ operator applied to the physical qubits in $\standard{K}$ implements a multi-controlled-$Z$ gate to every possible subset of qubits in $\mcK$:
    \begin{equation}
        \Z{k}_{\standard{K}}\equiv \prod_{\mcJ\subseteq\mcK} \overline{C^\mcJ Z}.
    \end{equation}
\end{circuit}

Using \cref{thm: subcube dimension implies logic} (or alternatively, \cref{cor: tilde operators are logically Hermitian}), subcube operators in lower levels of the Clifford Hierarchy than those above are necessarily trivial:
\begin{fact}
     Let $K\subseteq S$ be any subset of $k+1$ generators. The signed and unsigned $\Z{j}$ operators applied to the physical qubits in $\standard{K}$ are both stabilizers of the hypercube code for every $j<k$:
     \begin{equation}
         \Z{j}_{\standard{K}}\equiv \tildeZ{j}_{\standard{K}}\equiv\overline{\eye} \hspace{0.5em}\text{ for all } j<k.
     \end{equation}
\end{fact}

Lastly, \cref{thm: subcube dimension implies logic} combined with \cref{thm: Zk logical decomposition} and \cref{fact: transversal Zk subcube decomposition} imply the following result for \emph{arbitrary} subcube operators:
\begin{circuit}
    Signed and unsigned subcube operators of the same type necessarily implement the same logical circuits. That is, for any subset $K\subseteq S$ of $(k+1)$-generators, any $x\in\ZZ_2^m$, and $\mcK$ as defined above, the following hold for a hypercube code,
    \begin{align}
        \tildeZ{k}_{x+\standard{K}}&\equiv \tildeZ{k}_{\standard{K}} \equiv\overline{C^\mcK Z},\\
        \Z{k}_{x+\standard{K}}&\equiv \Z{k}_{\standard{K}} \equiv\prod_{\mcJ\subseteq\mcK} \overline{C^\mcJ Z}.
    \end{align}
\end{circuit}





\subsection{\texorpdfstring{$QRM_m(r-1,r)$}{QRMm(r-1,r)}}\label{example: r-1 r}
Perhaps the most natural generalization of the hypercube code is the family of quantum RM codes given by $QRM_m(r-1,r)$. Some subcube operators for these codes have been considered in past works \cite{rengaswamy2020optimality}. We are able to provide a complete classification of the logical circuits implemented by signed and unsigned operators on these codes, whereas the authors \cite{rengaswamy2020optimality,hu2022divisible,Hu2022designingquantum} only gave descriptions of \emph{global} transversal operators on $QRM_m(r-1,r)$.

The logical qubits of $QRM_m(r-1,r)$ are indexed by the set
\begin{equation}
    \mcQ\coloneqq \br{J\subset S\Bigmid \abs{J}=r},
\end{equation}
i.e., all $r$-sized subsets of generators. The first convenient fact about $QRM_m(r-1,r)$ codes is that the $X$ logical operators can be given by \emph{standard} subcube operators, rather that subcube operators that have been shifted away from the $0^m$ vertex. 
\begin{fact}
    For $QRM_m(r-1,r)$, $X_{\standard{S\setminus J}}\equiv X_{x+\standard{S\setminus J}}$ \emph{for every} $x\in\ZZ_2^m$.
\end{fact}
Thus, when $QRM_m(r-1,r)$ the sets
\begin{equation}
    \begin{aligned}
        L_Z&\coloneqq \Big\{Z_{\standard{J}}\Bigmid \abs{J}=r\Big\},\\
        L_X&\coloneqq \Big\{X_{\standard{J}}\Bigmid \abs{J}=m-r\Big\},
    \end{aligned}
\end{equation}
form a symplectic basis for the space of logical Pauli operators of $QRM_m(r-1,r)$.

Similarly to how the logical qubit index sets are defined by subsets of $S$ with a particular size, the $k$-th level logical index sets are also highly restricted:
\begin{equation}
    \mcQ_k\coloneqq \br{ K\subseteq S \Bigmid \abs{K}=(k+1)r}.
\end{equation}
Thus, for increasing values of $r$ there are more and more dimensions that do not support logical subcube operators.

Our \cref{thm: subcube dimension implies logic} for $QRM_m(r-1,r)$ can now be stated as:
\begin{fact}
    Consider the quantum RM code $QRM_m(r-1,r)$ and let $A\subcubeeq \ZZ_2^m$ be any subcube. The operators $\Z{k}_A$ and $\tildeZ{k}_A$ are logical operators if and only if $\dim A=(k+1)r$.
\end{fact}
See \cref{fig: r-1 r admissible} for a visual representation of this fact.
\begin{figure}[hb!]
    \centering
    \includegraphics{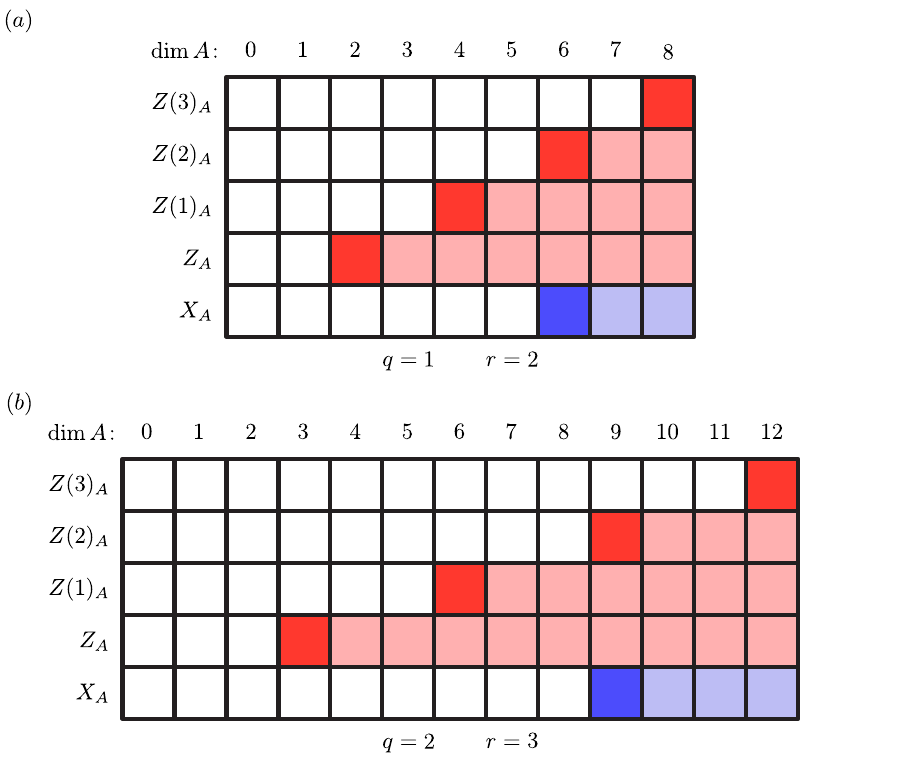}
    \caption{Consider the quantum RM code family $QRM_m(r-1,r)$. In the above figures, the shade of a given box indicates how the given operator for the given dimension will act on the code space:\\
    (1) A dark box indicates logic,\\
    (2) A light box indicates a logical identity, and\\
    (3) A white box indicates the code space is not preserved.}
    \label{fig: r-1 r admissible}
\end{figure}

The authors of \cite{rengaswamy2020optimality} give sufficiency for the $\Z{k}_A$ operator when $A=\ZZ_2^m$ is the \emph{entire} hypercube, and they prove what logical circuit is implemented.\footnote{We denote the Pauli operators as the 0-th level of the Clifford Hierarchy, whereas \cite{rengaswamy2020optimality} uses the 1-st level of the Clifford Hierarchy to represent the Paulis. Thus, Theorem 19 in \cite{rengaswamy2020optimality} states the condition $r\mid m$ implies that $\Z{\frac{m}{r}-1}\in\mcN^*$.} Necessity in the global case is proven in \cite{Hu2022designingquantum}.  We now detail the implemented logical circuits for arbitrary $\Z{k}_{\standard{K}}$ operators when $K\in\mcQ_k$.

As in the case of hypercube codes, we will determine the collection of minimal covers for a set $K\in\mcQ_k$. By definition, a subset of qubits $\mcJ\subseteq\mcQ$ is a minimal cover for $K$ if: (1) $\bigcup_{J\in\mcJ} J = K$, and (2) $\abs{\mcJ}=k+1$. As each logical qubit index $J\in\mcQ$ contains precisely $r$ elements, and we seek a collection of $k+1$ logical qubit indices whose union contains precisely $(k+1)r$ elements, we see that $\mcJ$ is necessarily a (pairwise disjoint) partition of $K$ into subsets of size $r$. Given a set $K$ of $(k+1)r$ elements, a collection of subsets of $K$, $\mcJ\subset\powerset{K}$, is said to be an \emph{$r$-partition} of $K$, denoted by $\mcJ\vdash_r K$, if (1) every $J\in\mcJ$ has size $\abs{J}=r$, and (2) $\mcJ$ is a cover of $K$. Note that for $\abs{K}=(k+1)r$ these two conditions are enough to guarantee that the sets in $\mcJ\vdash_r K$ are disjoint.
Ultimately, the minimal covers for $K\in\mcQ_k$ are given by
\begin{equation}
    \mcF(K) = \br{\mcJ \Bigmid \mcJ\vdash_r K},
\end{equation}
which is equal the previous definition for hypercube codes when $r=1$.

As we already considered the case of $r=1$ in the previous section. For the remainder of this section we suppose that $r\geq 2$. This assumption affords us the following via \cref{thm: Zk and tildeZk equivalence conditions}:
\vspace{-1em}
\begin{fact}
    Consider $QRM_m(r-1,r)$ where $r\geq 2$. For every $K\in\mcQ_k$, the signed and unsigned $\Z{k}$ operators applied to $\standard{K}$ perform the same logical circuit, $\Z{k}_{\standard{K}}\equiv \tildeZ{k}_{\standard{K}}$.
\end{fact}

Given we have already determined $\mcF(K)$, we now have the following:
\begin{circuit}
    For the code $QRM_m(r-1,r)$, $r\geq 2$, and $K\subseteq S$ with $\abs{K}=(k+1)r$, the transversal $\Z{k}$ operator applied to the physical qubits in $\standard{K}$ implements $(k+1)$-qubit multi-controlled-$Z$ gates to every subset of logical qubits whose index sets partition $K$:
    \begin{equation}
        \Z{k}_{\standard{K}} \equiv \prod_{\mcJ\vdash_r K}\overline{C^\mcJ Z}.
    \end{equation}
\end{circuit}
The authors of \cite{rengaswamy2020optimality} prove this for the $\Z{k}_A$ operator when $A=\ZZ_2^m$ is the \emph{entire} hypercube, though it is phrased in the language of phase polynomials.

Lastly, we note that in $QRM_m(r-1,r)$ codes we can use \emph{any} subcube of a particular type to implement the desired logical circuit. That is, \cref{thm: subcube dimension implies logic} combined with \cref{thm: Zk logical decomposition} and \cref{fact: transversal Zk subcube decomposition} imply the following result for \emph{arbitrary} subcube operators:
\begin{circuit}
    Subcube operators of the same type necessarily implement the same logical circuits. That is, for any subset $K\in\mcQ_k$ and any $x\in\ZZ_2^m$, the following holds for $QRM_m(r-1,r)$,
    \begin{equation}
        \Z{k}_{x+\standard{K}}\equiv \Z{k}_{\standard{K}}.
    \end{equation}
\end{circuit}










\newpage
\subsection{\texorpdfstring{$QRM_m(q,r)$}{QRMm(q,r)}, in general}\label{example: general}
Now consider the case of general quantum RM codes. In this case we have a symplectic basis given by the sets
\begin{equation}
    \begin{aligned}
        L_Z&\coloneqq \big\{Z_{\standard{J}}\Bigmid J\in\mcQ \big\},\\
        L_X&\coloneqq \big\{X_{e_J+\standard{S\setminus J}}\Bigmid J\in \mcQ\big\},
    \end{aligned}
\end{equation}
where we recall that $e_J\coloneqq \sum_{j\in J} e_j$ is the indicator bit string for the set $J\subseteq S$. When $q=r-1$ we did not need to shift the $X$ logicals by $e_J$, since every subcube operator of the same type implemented the same logic. In the general case where $r-q\geq 2$ the $e_J$'s must be present to have a symplectic basis. \cref{fig: m4 code example} gives a visual representation for the symplectic basis of the code $QRM_4(0,2)$. We note that when considering a particular dimension of logical operators, $L_i\coloneqq \br{\overline{Z}_J,\overline{X}_J\mid \abs{J}=i}$ for some $i\in\br{q+1,\dots,r}$, the set $L_i$ \emph{does} form a symplectic set. It is only when considering basis logical \emph{across} dimensions, e.g., $L_i\cup L_j$ for $i\neq j$, that the symplectic condition fails.

\begin{figure}[p!]
    \centering
    \includegraphics{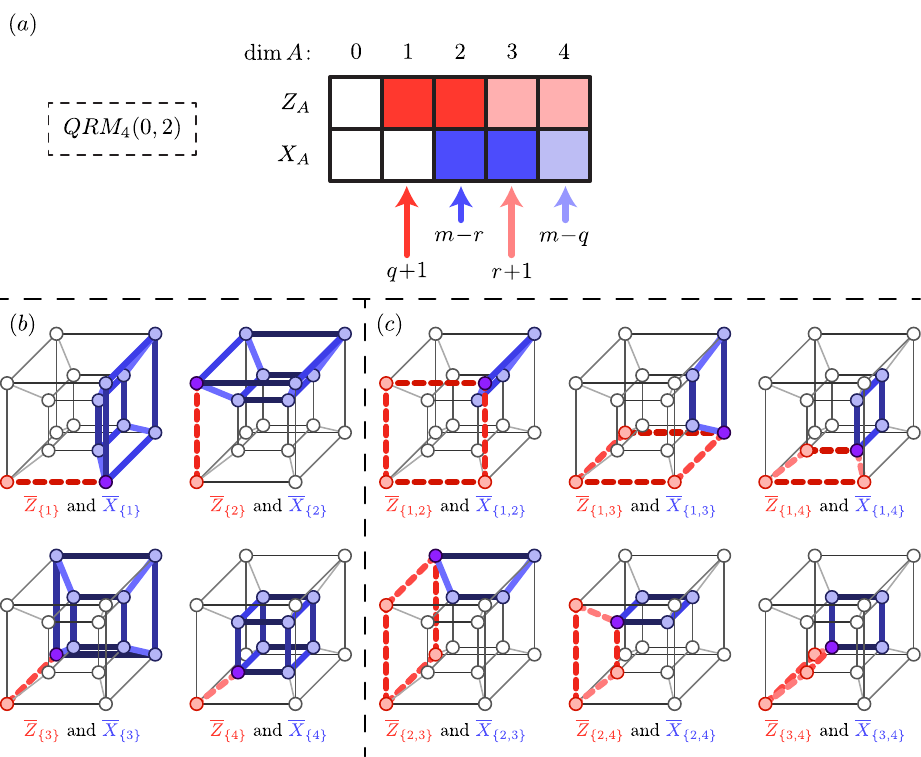}
    \caption{Consider the quantum RM code $QRM_4(0,2)$, whose physical qubits are indexed by the vertices of the 4-dimensional hypercube.  \\
    (a) By construction, a transversal $Z$ operator applied to a subcube with dimension equal to either 1 or 2 (edges/squares) is necessarily a $Z$ logical operator (represented by dark red boxes). Similarly, a transversal $X$ operator applied to a subcube with dimension equal to either 2 or 3 (squares/cubes) is necessarily an $X$ logical operator (represented by dark blue boxes). The light red and blue boxes indicate dimensions where $Z_A$ and $X_A$ acts as stabilizers of the code, respectively. The white boxes indicate dimensions where neither $Z_A$ nor $X_A$ preserve the code space.\\
    (b)--(c) The code has $\binom41+\binom42=10$ logical qubits that are indexed by subsets $J\subseteq[m]$ with $\abs{J}=1$ or $2$, and the distance of the code is 2. Thus, there are two classes of logical operators, those whose index set $J$ has size $\abs{J}=1$, shown in (b), and those whose index set has size $\abs{J}=2$, shown in (c). Each of the 10 4-cubes shown in (b)--(c) represent a symplectic pair of logical Pauli operators. The (red) dashed edges and squares indicate $Z$ subcube operators and the (blue) solid squares and cubes indicate $X$ subcube operators. A symplectic pair of operators overlap on a single qubit, namely, the qubit with index $e_J\coloneqq\sum_{i\in J}e_i$. These qubits are represented by (purple) vertices which lie at the intersection of the corresponding dashed and solid subcubes.\\
    (b) The first class of logical operators are the $Z$ operators that act on subcubes of dimension 1 (dashed red edges), together with the $X$ operators that act on subcubes of \emph{codimension} 1 (solid blue cubes). \\
    (c) The second class of logical operators are the $Z$ operators that act on subcubes of dimension 2 (dashed red squares), together with the $X$ operators that act on subcubes of \emph{codimension} 2 (solid blue squares).}
    \label{fig: m4 code example}
\end{figure}

\paragraph{Code properties.} 
The parameters and derived properties of quantum RM codes are summarized in \cref{tab:code parameters}. The property $k_{\max}$ denotes the highest level of the Clifford Hierarchy that can be transversally implemented on the given code using $\Z{k}$ subcube operators; as every subcube has dimension has at most $m$, \cref{thm: subcube dimension implies logic} implies that $k_{\max}$ is the largest integer such that $q+k_{\max}r+1\leq m$ for a given $QRM_m(q,r)$. We give the properties of all quantum RM codes with $m \leq 10$ and $k_{\max}\geq 2$ (i.e., can support transversal $T$) in \cref{tab:codeimplementations}.

\begin{table}[h]
    \centering
    \begin{tabular}{|c|l|c|}
        \hline\textbf{Variable} & \textbf{Description} & \textbf{Formula} \\
        \hline\hline
        $m$ & $2^m$ physical qubits indexed by the $m$-dimensional hypercube & \\
        $r$ & $Z$ stabilizers given by $(r + 1)$-cubes & \\
        $q$ & $X$ stabilizers given by $(m-q)$-cubes & \\
        $\kappa$ & number of logical qubits & $\sum_{i =q+1}^r \binom{m}{i}$  \\
        $d$ & code distance & $2^{\min\{q + 1, m - r\} }$ \\
        $k_{\max}$ & highest achievable level in the Clifford Hierarchy& $\left\lfloor\frac{m - (q + 1)}{r}\right\rfloor$\\
        \hline
    \end{tabular}
    \caption{Parameters and properties of quantum RM codes.}
    \label{tab:code parameters}
\end{table}

\input{tables/code_table}

\newpage
\paragraph{Computing minimal covers.}
The following simple algorithm can be used to compute the collection of $\mcQ$-minimal covers given a $k$-th level logical index set $K\in\mcQ_k$ (implemented \href{https://github.com/christopherkang/geometric-qrm}{here}).

\vspace{-0.3em}

\begin{algorithm}[H]
    \caption{Constructs logical circuit implemented by a $K$ type and $k$ level signed standard subcube operator. Acts upon an $m, r, q$ code.}
    \begin{algorithmic}[1]\label{alg:qlss-qsvt-hamiltonian}
        \Require $k$ is the level of the physical operators applied to the subcube, $m, r, q$ are valid parameters, $K \subseteq [m]$, 
        \Procedure{logical-operator-enumeration}{$k, K, (m, r, q)$}
            \State $\mathcal{Q}_{0 | K} = \{ J \subseteq K |  J\in\mathcal{Q}_0 \}$ \Comment{Enumerate index set of logical operators within $K$; our operators will not affect logic outside. $\mathcal{Q}_{0 | K}$ are the sets of logicals to consider }
            \State $ \mcF \gets \{ \}$
            \For{$\mcJ \subseteq \mathcal{Q}_{0 | K}$ with $\abs{\mcJ} = k + 1$} \Comment{Get minimal permutations of logicals}
                \If { $\cup_{J \in \mcJ } J = K$}
                    \State $\mcF \gets \mcF \cup \{ \mcJ \}$
                \EndIf
            \EndFor
            \State\Return $\mcF$ \Comment{This is $\mcF(K)$, precisely where logic happens}
        \EndProcedure
    \end{algorithmic}
\end{algorithm}

\vspace{-1.85em}

\paragraph{Example circuits.}
\cref{fig:25502,fig:25602,fig:26602,fig:26612} demonstrate example logical circuits produced via signed subcube operators, which are non-trivial to describe without the ``minimal cover'' terminology. \cref{fig:25602} and \cref{fig:26602} demonstrate that applying a signed subcube of different dimensions effects completely different logic. \cref{fig:26612} shows yet another implementation, albeit now with a slightly altered $r$ parameter. Finally, observe that \cref{fig:25502} and \cref{fig:25602} effect the same logical circuit, despite the fact that the QRM codes have a different $m$ parameter.


\vspace{-0.35em}

\subsection*{Acknowledgements}
We are grateful to Aleksander Kubica for helpful discussions. 
N.C. acknowledges Luke Schaeffer for many helpful discussions on the Clifford Hierarchy.
D.H. is grateful for many discussions with Michael Gullans and Dolev Bluvstein which are always illuminating and inspiring. 

This work was done in part while we were visiting the Simons Institute for the Theory of Computing, supported by DOE QSA grant FP00010905.
Research of A.B. and N.C. was partially supported by NSF grant CCF2330909. Research of N.C. was also partially supported by NSF grant DMS2231533. D.H. acknowledges financial support by the U.S. DoD through a QuICS Hartree fellowship. C.K. is funded in part by EPiQC, an NSF Expedition in Computing, under award CCF-1730449; in part by STAQ under award NSF Phy-1818914/232580; in part by the US Department of Energy Office of Advanced Scientific Computing Research, Accelerated Research for Quantum Computing Program; and in part by the NSF Quantum Leap Challenge Institute for Hybrid Quantum Architectures and Networks (NSF Award 2016136), in part based upon work supported by the U.S. Department of Energy, Office of Science, National Quantum Information Science Research Centers, and in part by the Army Research Office under Grant Number W911NF-23-1-0077. The views and conclusions contained in this document are those of the authors and should not be interpreted as representing the official policies, either expressed or implied, of the U.S. Government. The U.S. Government is authorized to reproduce and distribute reprints for Government purposes notwithstanding any copyright notation herein.

\clearpage
\newpage

\begin{figure}[t!]
    \centering
    \includegraphics[width=0.9\linewidth]{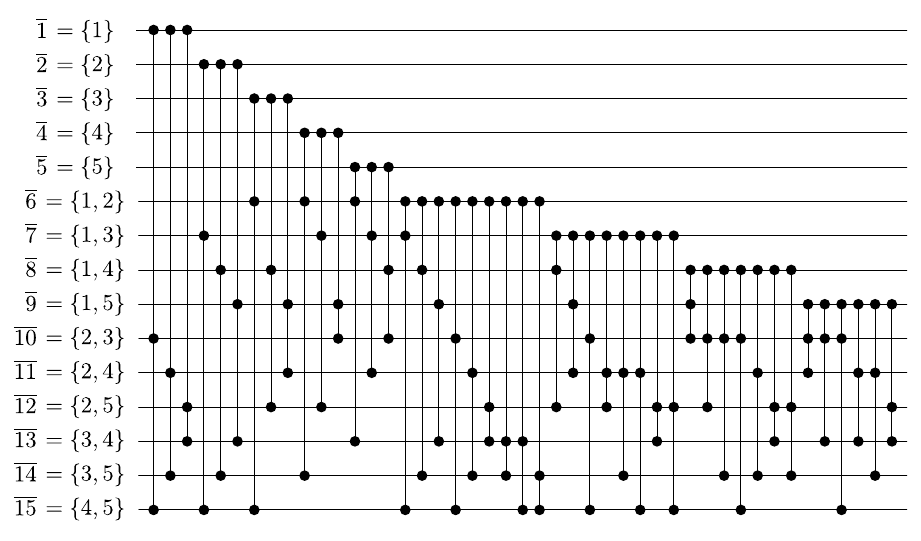}
    \caption{Global signed $T$ applied to $QRM_5(0,2)$.}
    \label{fig:25502}
\end{figure}

\begin{figure}[b!]
    \centering
    \includegraphics[width=0.9\linewidth]{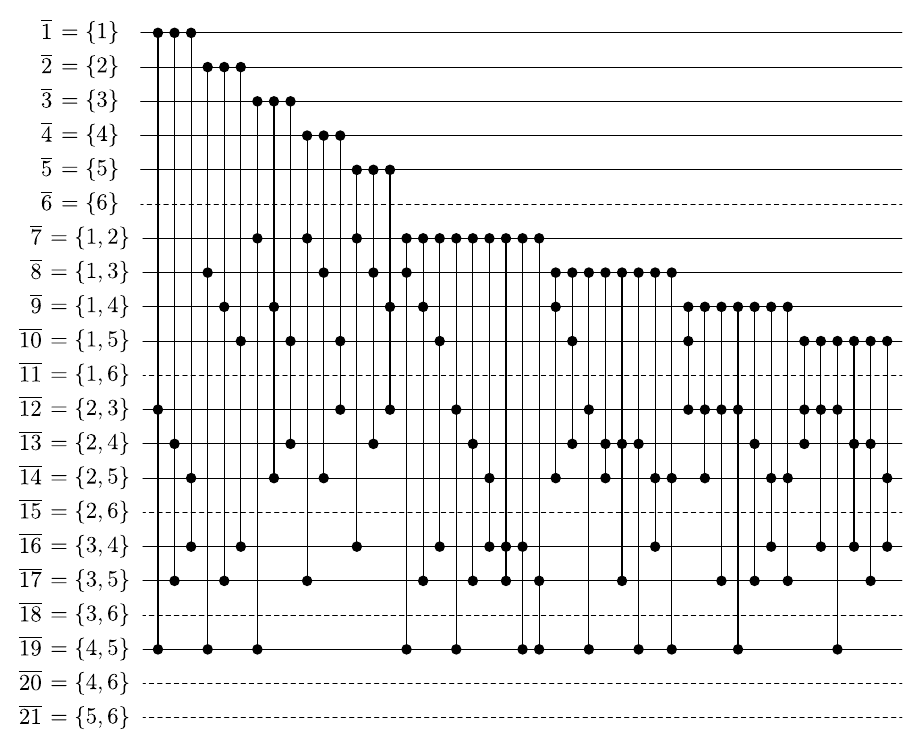}
    \caption{Signed $T$ operator applied to the $\standard{1,2,3,4,5}$ subcube of $QRM_6(0,2)$. Note that 6 qubits are unaffected (dotted lines) and that the circuit is identical to the one given in \cref{fig:25502} on the remaining qubits.}
    \label{fig:25602}
\end{figure}

\clearpage
\newpage

\begin{figure}[t!]
    \centering
    \includegraphics[width=0.9\linewidth]{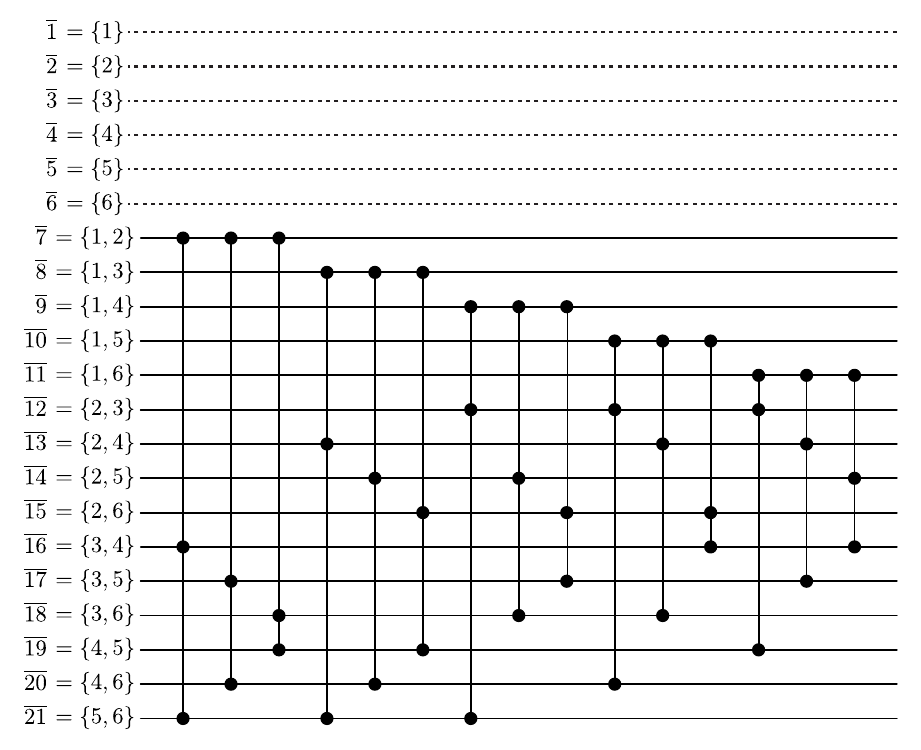}
    \caption{Global signed $T$ applied to $QRM_6(0,2)$. Note that 6 qubits are unaffected (dotted lines) and that the circuit is identical to the one given in \cref{fig:26612} on the remaining qubits.}
    \label{fig:26602}
\end{figure}

\begin{figure}[b!]
    \centering
    \includegraphics[width=0.9\linewidth]{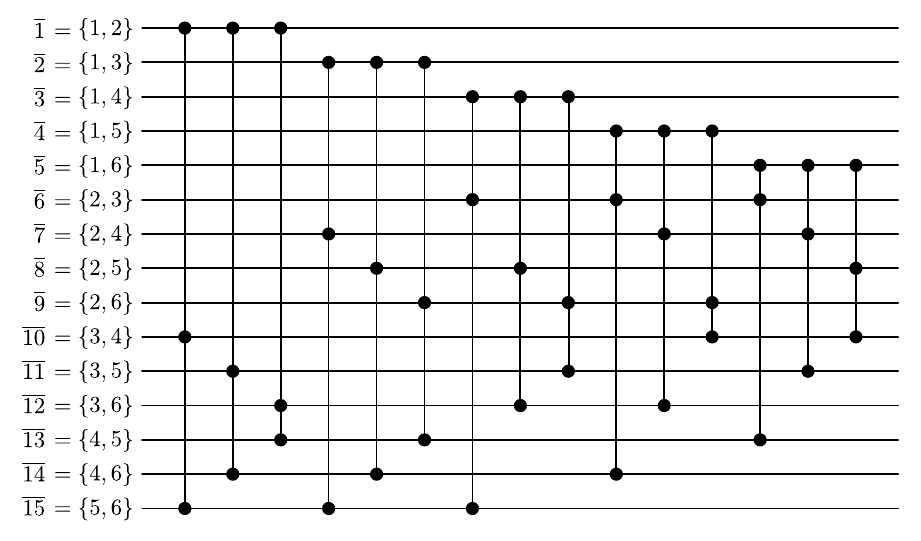}
    \caption{Global signed $T$ applied to $QRM_6(1,2)$.}
    \label{fig:26612}
\end{figure}

\clearpage
\newpage
\part{Formal results}

\section{The Clifford Hierarchy and CSS code preliminaries}

\subsection{The Pauli group and CSS codes}\label{sec:Pauli group and quantum codes}
Let $\mcP_1=\br{\pm\eye, \pm i\eye, \pm X, \pm i X,\pm Y, \pm i Y, \pm Z, \pm iZ}$ denote the single-qubit Pauli group, and $\mcP_n\coloneqq \mcP_1^{\otimes n}$ the $n$-qubit Pauli group.
Given a binary vector, $v\in\FF^n$, we define the $n$-qubit Pauli operators $X(v)\coloneqq\bigotimes_{i=1}^n X^{v_i}$ and $Z(v)\coloneqq\bigotimes_{i=1}^n Z^{v_i}$.

A quantum CSS code is defined by two classical linear codes, $ C_1, C_2$, on the same vector space such that $ C_1^\perp\subseteq C_2$. Define sets of Pauli operators 
    \begin{equation}\label{eq:SZSZ}
    \begin{aligned}
     \mcS_X\coloneqq \{ X(x) \mid x\in C_1^\perp\}\\
     \mcS_Z\coloneqq \{ Z(z) \mid z\in C_2^\perp\},
    \end{aligned}
    \end{equation}
called \emph{$X$ and $Z$ type stabilizers}, respectively (because of this,
another often used notation for the component codes is $C_X$ and $C_Z$). By the assumption $ C_1^\perp\subseteq C_2$ every operator in $\mcS_X$ commutes with every operator in $\mcS_Z$. The CSS code, $\CSS( C_1, C_2)\subseteq (\CC_2)^{\otimes n}$, is defined as the common $+1$ eigenspace of all of the $S_X$ and $S_Z$ operators defined by $ C_1^\perp$ and $ C_2^\perp$. The set $\CSS( C_1, C_2)$ is referred to as the \emph{code space} and its elements are referred to as \emph{code states}.

We further define the Pauli operators $\mcN_X\coloneqq \br{ X(x) \mid x\in C_2}$ and $\mcN_Z\coloneqq \br{ Z(z) \mid z\in C_1}$. These sets are called the \emph{undetectable $X$ and $Z$ errors}, respectively. Note that by construction $\mcS_X\subseteq\mcN_X$ and $\mcS_Z\subseteq\mcN_Z$. Elements of $\mcS_X\cup\mcS_Z$ are called ``trivial'' errors on $\CSS( C_1, C_2)$ since for every $S\in\mcS_X\cup\mcS_Z$ and every $\ket\psi\in\CSS( C_1, C_2)$ we have that $S\ket\psi=\ket\psi$. The sets $\mcE_X\coloneqq \mcN_X\setminus\mcS_X$ and $\mcE_Z\coloneqq \mcN_Z\setminus\mcS_Z$, on the other hand, are called \emph{logical $X$ and $Z$ type errors}, respectively, since for any $L\in\mcE_X\cup\mcE_Z$ and $\ket\psi\in\CSS( C_1, C_2)$, $L\ket\psi$ may not equal $\ket\psi$, but $L\ket\psi\in\CSS( C_1, C_2)$ is still a valid code state.

Suppose $ C_1$ and $ C_2$ are $[n,k_1,d_1]$ and $[n,k_2,d_2]$ codes, respectively. The dimension of $\CSS( C_1, C_2)$ as a subspace of $(\CC_2)^{\otimes n}$ is precisely 
   \begin{equation}\label{eq: dim}
    \dim\CSS( C_1, C_2)= k_1+k_2-n  . 
   \end{equation}
   The $X$ distance of the quantum code is equal to the minimal weight of a logical $X$ error, i.e., $d_X\coloneqq \min_{x\in C_2\setminus C_1^\perp} \abs{x}$, and similarly for the $Z$ distance. The distance of $\CSS( C_1, C_2)$ is defined as $d\coloneqq\min\br{d_X,d_Z}$. We say that $\CSS( C_1, C_2)$ is an $[[ n, k_1+k_2-n, d]]$ quantum CSS code.

\begin{table}[t]
    \centering
    \begin{tabular}{|c c c c c|}
        \hline & \underline{$X$ stabilizers} &  &  \underline{Undetectable $X$ errors} &\\[0.5em]
        & $C_1^\perp$ & $\subseteq$ & $C_2$ & \\
        & $\rotatebox{90}{=}$ & & $\rotatebox{90}{=}$ & \\
        & $\im H_X^T$ & &$\ker H_Z$ & \\[0.5em]\hline
        $S_X$ & $\xrightarrow[]{H_X^T}$ & $C$ & $\xrightarrow[]{H_Z}$ & $S_Z$
        \\ \hline\hline
        & \underline{$Z$ stabilizers} &  &  \underline{Undetectable $Z$ errors} &\\[0.5em]
        & $C_2^\perp$ & $\subseteq$ & $C_1$ & \\
        & $\rotatebox{90}{=}$ & & $\rotatebox{90}{=}$ &\\
        & $\im H_Z^T$ & &$\ker H_X$ & \\[0.5em] \hline
        $S_Z$ & $\xrightarrow[]{H_Z^T}$ & $C$ & $\xrightarrow[]{H_X}$ & $S_X$
        \\ \hline 
    \end{tabular}
    \caption{Inclusions between constituent $X$ and $Z$ stabilizer codes.}
    \label{tab:stabs-and-logs}
\end{table}

Logical $X$ and $Z$ type errors in a CSS code correspond precisely to elements of $ C_2\setminus  C_1^\perp$ and $ C_1\setminus  C_2^\perp$, respectively. The quotient spaces, $ C_2/  C_1^\perp$ and $ C_1/  C_2^\perp$, correspond to equivalence classes of logical errors and these spaces are isomorphic to one another. Let $x+C_2^\perp$ be a coset of $C_2^\perp$ in $C_1$. Any two elements $a,b\in x+C_2^\perp$ correspond to the same logical $X$ error on the code space, i.e., $X(a)\ket\psi= X(b)\ket\psi$ for every $\ket\psi\in\CSS(C_1,C_2)$. Thus, the equivalence classes of logical $X$ errors are uniquely determined by any set of coset representatives for the elements of $C_2/C_1^\perp$. In particular, if a set $M=\br{x_1,\dots,x_r,x_{r+1},\dots,x_{k_2}}$ is a linearly-independent basis for $C_2$, and the subset $T=\br{x_1,\dots,x_r}$ is a linearly-independent basis for $C_1^\perp\subseteq C_2$, then the set $M\setminus T=\br{x_{r+1},\dots,x_{k_2}}$ is a set of coset representatives for the quotient space $C_2/C_1^\perp$. The case for logical $Z$ errors is the analogous.

The inclusions of the various codes and subsets mentioned here as well as the maps between them are collected in \cref{tab:stabs-and-logs}. A homological view of the CSS construction is summarized by the following diagram:
\begin{equation}\label{eq:cd}
\begin{tikzcd}[wire types={n,n}]
\setwiretype{n}
\FF^{r_{\!{_{2}}}} \arrow[r,shift left,"H_Z^\intercal"]
   &\;
\FF^n \arrow[l,shift left,"H_Z"] \arrow[r,shift left,"H_X"] &\;
\FF^{r_{\!{_{1}}}} \arrow[l,shift left,"H_X^\intercal"] ,
\end{tikzcd}
\end{equation}
where $r_i=n_i-k_i,i=1,2$ are the counts of independent stabilizer generators of the $Z$ and $X$ types. In \cref{tab:stabs-and-logs} as well as in \cref{eq:cd}, $H_X\in \FF^{r_1\times n}$ is the parity-check matrix of $C_1$, and $H_Z\in \FF^{r_2\times n}$ is the same for $C_2$. We note again that $\dim \CSS(C_1,C_2)=\dim(\ker H_2^\intercal/\im H_1),$ in accordance with \cref{eq: dim}.

\subsection{The Clifford Hierarchy and Clifford errors}\label{sec: CH and errors}
The $0$-th level of the $n$-qubit Clifford Hierarchy is defined to be the Pauli group, $\Cl{0}\coloneqq \mcP_n$. 
For $k\in\NN$, the $k$-th level of the Clifford Hierarchy is defined recursively as the set
\begin{equation}
    \Cl{k}\coloneqq\br{U\in \unitary(2^n)\Bigmid U \Cl{0} U^\dagger\subseteq \Cl{k-1}}.
\end{equation}
In other words, elements of level $k$ conjugate Pauli operators to elements of level $(k-1)$. As an example, consider the rotations $\Z{k}\coloneqq \ketbra{0}+e^{i\frac{\pi}{2^k}}\ketbra{1}$ and $\X{k}\coloneqq \ketbra{+}+e^{i\frac{\pi}{2^k}}\ketbra{-}$. $\Z{k}$ and $\X{k}$ are in $\Cl{k}$ for the given $k$. The most common gates of this form are $Z=\Z{0}$, $S=\Z{1}$, i.e., the phase gate, and $T=\Z{2}$. We note that aside from $\Cl{0}$, the Pauli group, and $\Cl{1}$, the well-known \emph{Clifford group}, $\Cl{k}$ is not a subgroup of the unitary group.

Now, consider a stabilizer code $\mcC\subseteq(\CC^2)\n$. It's Pauli stabilizer group, $\mcS$, is defined as the set of Pauli operators whose mutual $+1$ eigenspace is precisely $\mcC$. In particular, it can be defined as the set
\begin{equation*}
    \stabs{0}\coloneqq\br{U\in\Cl{0}\Bigmid U\ket\psi=\ket\psi \;\;\forall\ket\psi\in\mcC},
\end{equation*}
i.e., a subset of the 0-th level of the Clifford Hierarchy that leaves the codespace invariant. This definition motivates the following extension to higher levels of the Clifford Hierarchy:
\begin{definition}
    The level-$k$ Clifford stabilizers of $\mcC$ are the operators in the $k$-th level of the Clifford Hierarchy that leave states in $\mcC$ invariant:
    \begin{equation}
        \stabs{k}\coloneqq\br{U\in\Cl{k}\Bigmid U\ket\psi=\ket\psi \;\;\forall\ket\psi\in\mcC}.
    \end{equation}
    Elements of $\stabs{k}$ perform the \emph{logical identity operator}, $\overline{\eye}$, on $\mcC$.
\end{definition}
Likewise, we can define the so-called \emph{undetectable Clifford Hierarchy errors} in the following way.
\begin{definition}
    The level-$k$ undetectable Clifford errors of $\mcC$ are the operators in the $k$-th level of the Clifford Hierarchy that conjugate Pauli stabilizers of $\mcC$ to logical identity operators in the $(k-1)$-st level of the Clifford Hierarchy:
    \begin{equation}
        \errors{k}\coloneqq\br{U\in\Cl{k}\Bigmid U\stabs{0}U^\dagger\subseteq \stabs{k-1}},
    \end{equation}
    where by convention we set $\stabs{-1}\coloneqq\stabs{0}$.
\end{definition}
Note that the set of level-$0$ errors, $\errors{0}$, is the usual space of undetectable \emph{Pauli} errors, and $\errors{k}$ corresponds to the intuitive notion of an \emph{undetectable error} on a code:
\begin{lemma}\label{lem:undetectable-error-equivalences}
    For a stabilizer code $\mcC$, let $\Pi_\mcC\coloneqq\frac{1}{\abs{\stabs{0}}}\sum_{S\in\stabs{0}} S$ denote the code space projector. Suppose $U\in\Cl{k}$ is a $k$-th level Clifford operator. The following are equivalent:
    \begin{enumerate}
        \item $U$ is an undetectable Clifford error: $U\in\errors{k}$,
        \item $U$ preserves the code space: for every $\ket\psi\in\mcC$, $U\ket\psi\in\mcC$, and
        \item $U$ commutes with the codespace projector: $U\Pi_\mcC U^\dagger = \Pi_\mcC$.
    \end{enumerate}
\end{lemma}
\begin{proof}
    ($2\Leftrightarrow 3$) This is equivalent to the statement that $U$ and $\Pi_\mcC$ preserve each others' eigenspaces if and only if they commute with each other.
    
    ($1\Rightarrow 3$) Consider $U\Pi_\mcC U^\dagger = \frac{1}{\abs{\stabs{0}}}\sum_{S\in\stabs{0}} USU^\dagger$. By definition of $\errors{k}$, each $USU^\dagger$ acts as logical identity on $\mcC$, and so for any $\ket\psi\in\mcC$ we have that $U\Pi_\mcC U^\dagger\ket\psi=\frac{1}{\abs{\stabs{0}}}\sum_{S\in\stabs{0}} \ket\psi = \ket\psi$. This implies that $U\Pi_\mcC U^\dagger$ and $\Pi_\mcC$ have the same eigenspaces. Since $U$ is unitary it preserves the rank of $\Pi_\mcC$, and so $U\Pi_\mcC U^\dagger$ must be equal to the code space projector of $\mcC$.

    ($2\Rightarrow 1$) Let $S\in\stabs{0}$. As $U\ket\psi$ can be an arbitrary element of the code space and $U\ket\psi = U S\ket\psi = (USU^\dagger) U\ket{\psi}$, $USU^\dagger$ acts as logical identity on the code space.
\end{proof}
We note that the above equivalence does not rely on the Clifford Hierarchy in any way; the property of being a logical operation is equivalent to conjugating Pauli stabilizers of $\mcC$ to logical identity operators of the code space. We restricted attention to Clifford Hierarchy operators simply because the goal of the present work is to give transversal operators that are in the Clifford Hierarchy.

An important class of undetectable Clifford Hierarchy operators are those that act non-trivially on the code space.
\begin{definition}
    The level-$k$ (non-trivial) Clifford logicals of $\mcC$ are the operators in the $k$-th level of the Clifford Hierarchy that conjugate Pauli stabilizers of $\mcC$ to logical identity operators in the $(k-1)$-st level of the Clifford Hierarchy, but which are not, themselves, level-$k$ Clifford stabilizers:
    \begin{equation}
        \logs{k}\coloneqq\errors{k}\setminus\stabs{k}.
    \end{equation}
\end{definition}

\begin{definition}\label{def: logically equivalent}
    Two operators, $U$ and $V$, are said to be \emph{logically equivalent} on a stabilizer code $\mcC$, denoted $U\equiv_\mcC V$ or simply $U\equiv V$, if $USU^\dagger=VSV^\dagger$ for every $S\in\stabs{0}$. This implies that, up to a global phase $e^{i\theta}$, $U\ket\psi = e^{i\theta} V\ket\psi$ for every $\ket\psi\in\mcC$.
\end{definition}
So, Clifford stabilizers are the operators that are logically equivalent to identity. We have the following characterization of Clifford stabilizers.

\begin{fact}\label{lem:Clifford-stabilizer-equivalences}
    Let $\mcC$ be a stabilizer code and suppose $U\in\Cl{k}$ is a $k$-th level Clifford operator. $U\in\stabs{k}$ if and only if $UPU^\dagger \equiv P$ for every $P\in\errors{0}$.
\end{fact}
Note that since $\stabs{k}\subseteq\errors{k}$ one would expect that \cref{lem:Clifford-stabilizer-equivalences} is stricter than \cref{lem:undetectable-error-equivalences}. Indeed, one recovers \cref{lem:undetectable-error-equivalences} by requiring only $P\in\stabs{0}$ in the second condition, instead of $P\in\errors{0}$.

Clifford stabilizers act as identity on the entire code space; determining the action of a non-trivial Clifford logical will require an understanding of the individual ``logical qubits'' of a stabilizer code, $\mcC$. We first note it is well-known that: (1) $\stabs{0}$ is an Abelian subgroup of $\mcP_n$ not containing $-\eye$, and (2) the center of $\errors{0}$ is precisely $\stabs{0}$.
\begin{fact}
    Let $\mcC$ be an $[[ n, \kappa]]$ stabilizer code with stabilizer group $\stabs{0}$ and undetectable Pauli-error group $\errors{0}$. There exist $2\kappa$ operators in $\errors{0}$, $\br{P_i, Q_i}_{i=1}^{\kappa}$, with the following properties:
    \begin{enumerate}
        \item $P_i$ and $P_j$ commute for every $i,j\in[\kappa]$,
        \item $Q_i$ and $Q_j$ commute for every $i,j\in[\kappa]$,
        \item $P_i$ and $Q_j$ anti-commute if and only if $i=j$, and
        \item $\errors{0}$ is generated by $\stabs{0}$ together with the $P_i$ and $Q_i$ operators: $\errors{0}=\langle\stabs{0}\cup\br{P_i, Q_i}_{i=1}^{\kappa}\rangle$.
    \end{enumerate}
    Any set $\br{P_i, Q_i}_{i=1}^{\kappa}$ that satisfies these properties is known as a \emph{symplectic basis} for $\kappa$ logical qubits, and when the symplectic basis arises from the undetectable errors of a stabilizer code we say that it is a symplectic basis for $\mcC$.
\end{fact}
Logical qubits of $\mcC$ can be indexed by pairs of anti-commuting operators in a symplectic basis. In particular, the \emph{$i$-th logical qubit} of $\mcC$ is determined by the pair $B_i\coloneqq\br{P_i,Q_i}$, and a set of logical qubits, $\mcL\subseteq[\kappa]$, is determined by the sub-basis $B_\mcL\coloneqq\bigcup_{i\in \mcL}\br{P_i,Q_i}$. 

Given a subset of logical qubits $\mcL\subseteq[\kappa]$, consider the subgroup $\errors{0}_\mcL\coloneqq\langle \stabs{0},B_\mcL\rangle\leq\errors{0}$ of the full undetectable Pauli-error group. We note that $\errors{0}_\emptyset = \stabs{0}$ and $\errors{0}_{[\kappa]}=\errors{0}$, but for non-trivial $\mcL\subset[\kappa]$ the set $\errors{0}_\mcL$ will depend on a particular choice of symplectic basis.
\begin{definition}
    Let $\mcC$ be an $[[ n, \kappa]]$ stabilizer code with a fixed symplectic basis $\mcC_{[\kappa]}\coloneqq\br{P_i,Q_i}_{i\in [\kappa]}$. Suppose $U\in\errors{k}$  is a level-$k$ Clifford logical for $Q$, and that $\mcL\subseteq[\kappa]$ is a subset of the logical qubits. $U$ is said to \emph{act trivially on the qubits in $\mcL$} if 
    $UPU^\dagger\equiv P$ for every undetectable Clifford error $P\in\errors{0}_\mcL$.
\end{definition}
Note that this definition agrees with the condition in \cref{lem:Clifford-stabilizer-equivalences} for Clifford stabilizers, which by definition act trivially on every logical qubit. It is straightforward to show that the condition $UPU^\dagger\ket\psi=P\ket\psi$ for every $P\in\errors{0}_\mcL$ is stronger than necessary; we need only check the equality holds for the symplectic basis elements that determine the qubits in $\mcL$:
\begin{fact}
    $U\in\errors{k}$ acts trivially on the logical qubits in a set $\mcL\subseteq[\kappa]$ if and only if $U$ \emph{logically commutes} with every basis operator $P\in B_\mcL$, $UP\equiv PU$.
\end{fact}

\subsection{The single-qubit case}\label{sec:Z-X-basis-logicals}

Consider the following single-qubit $Z$ and $X$ rotation gates:
\begin{align}
    \Z{k}&\coloneqq \ketbra{0}+e^{i\frac{\pi}{2^{k}}}\ketbra{1},
    & \X{k}&\coloneqq \ketbra{+}+e^{i\frac{\pi}{2^{k}}}\ketbra{-},\\
    & = \begin{bmatrix}
        1 & 0 \\
        0 & e^{i\frac{\pi}{2^{k}}}
    \end{bmatrix},
    &&=\frac{1}{2}\begin{bmatrix}
        1+e^{i\frac{\pi}{2^{k}}} & 1-e^{i\frac{\pi}{2^{k}}} \\
        1-e^{i\frac{\pi}{2^{k}}} & 1+e^{i\frac{\pi}{2^{k}}}
    \end{bmatrix}.
\end{align}
The $\Z{k}$ operators are defined so that the reproduce the natural $k$-th level Clifford Hierarchy single-qubit $Z$ basis gates: $Z(-1)=\eye$, the identity, $Z(0)=Z$, the Pauli $Z$ operator, $Z(1)=\phasegate=\sqrt{Z}$, the phase gate, $Z(2)=T=\sqrt{\phasegate}$, the $T$ gate, etc.
$\Z{k}$ and $\X{k}$ are related to each other via the Hadamard matrix, $\had \Z{k}\had = \X{k}$. Note that for $\ell\in\br{0,\dots,k+1}$, $\Z{k}^{2^\ell}=\Z{k-\ell}$ and $\X{k}^{2^\ell}=\X{k-\ell}$, implying that $\Z{k}$ and $\X{k}$ have order $2^{k+1}$.

The following conjugation identities can be verified by direct computation:
\begin{lemma}\label{lem:conjugation-identities} Let $\omega_k = e^{-i \frac \pi {2^k}}$ and suppose $k\in\ZZ_{\geq 0}$.
    \begin{align}
        \Z{k} X \Z{k}^\dagger &= {\omega_k} \Z{k-1} X,\\
        \X{k} Z \X{k}^\dagger &= {\omega_k} \X{k-1} Z,\\
        \Z{k}^\dagger X \Z{k} &= \omega_k^{-1} \Z{k-1}^\dagger X,\\
        \X{k}^\dagger Z \X{k} &= \omega_k^{-1} \X{k-1}^\dagger Z.
    \end{align}
\end{lemma}
As $X$ and $Z$ are Hermitian and conjugation preserves Hermiticity, we also have
\begin{corollary} Let $\omega_k = e^{-i \frac \pi {2^k}}$ and suppose $k\in\ZZ_{\geq 0}$.
    \begin{align}
        \Z{k} X \Z{k}^\dagger &= \omega_k^{-1} X\Z{k-1}^\dagger,\\
        \X{k} Z \X{k}^\dagger &= \omega_k^{-1} Z\X{k-1}^\dagger,\\
        \Z{k}^\dagger X \Z{k} &= {\omega_k} X\Z{k-1},\\
        \X{k}^\dagger Z \X{k} &= {\omega_k} Z\X{k-1}.
    \end{align}
\end{corollary}

As $\Z{1}=Z$, $\X{1}=X$, the conjugation identities imply the following via induction:
\begin{corollary}
    $\Z{k}$ and $\X{k}$, their adjoints, and tensors products thereof, are all in the $k$-th level of the Clifford Hierarchy.
\end{corollary}
\begin{proof}
    The latter results are a direct consequence of $\Z{k},\X{k}\in\Cl{k}$;
    we will only prove the claim for $\Z{k}$ as the proof of $\X{k}$ is analogous. Induction on $k$. Clearly for $k=0$, $\Z{0}=Z\in\mcP_1$. Suppose now that $\Z{k}\in\Cl{k}$ for $k\geq 0$ and consider $\Z{k+1}$. Since the operator is diagonal we only must prove $\Z{k+1}X\Z{k+1}\in\Cl{k}$, i.e., that $\Z{k+1}X\Z{k+1}=\omega_k \Z{k}X$ (\cref{lem:conjugation-identities}) conjugates both $X$ and $Z$ to elements of $\Cl{k-1}$.
    \begin{align}
        (\omega_k \Z{k}X) Z (\omega_k \Z{k}X)^\dagger &= \Z{k}X Z X \Z{k}^\dagger,\\
        &=  -\Z{k} Z \Z{k}^\dagger,\\
        &= -Z,\\
        \intertext{is in $\in\mcP_1\subset\Cl{k-1}$, and}
        (\omega_k \Z{k}X) X (\omega_k \Z{k}X)^\dagger &= \Z{k} X \Z{k}^\dagger,
    \end{align}
    is in $\Cl{k-1}$ since by the induction hypothesis $\Z{k}\in\Cl{k}$.
\end{proof}
Lastly, we note that conjugating an even number of $X$ (resp. $Z$) gates by even numbers of $\Z{k}$ and its adjoint (resp. $\X{k}$ and its adjoint), results in a gate with no overall phase:
\begin{fact}\label{fact:even-has-no-phase}
    \begin{align}
        \left(\Z{k}^{\otimes \frac{n}{2}} \otimes \Z{k}^{\dagger \otimes \frac{n}{2}}\right) X\n \left(\Z{k}^{\otimes \frac{n}{2}} \otimes \Z{k}^{\dagger \otimes \frac{n}{2}}\right)^\dagger &= \left(\Z{k-1}^{\otimes \frac{n}{2}} \otimes \Z{k-1}^{\dagger \otimes \frac{n}{2}}\right) X\n, \\
        &= X\n\left(\Z{k-1}^{\dagger\otimes \frac{n}{2}} \otimes \Z{k-1}^{\otimes \frac{n}{2}}\right) ,
    \end{align}
    and similarly for the case involving $\X{k}$ and $Z$.
\end{fact}

\subsection{Controlled-\texorpdfstring{$Z$}{Z} circuits}\label{sec:CZ circuits}

\begin{definition}
    For $\ell\in\NN$, the \emph{multi-controlled-$Z$} gate is defined recursively as the $\ell$-qubit unitary operator $C^{(\ell)}Z\coloneqq \ketbra{0}\otimes \eye+\ketbra{1}\otimes C^{(\ell-1)}Z$, where $C^{(0)}Z\coloneqq Z$.
\end{definition}
$C^{(\ell)}Z$ is symmetric in the $\ell$ qubits; in particular, $C^{(\ell)}Z$ is a diagonal gate that introduces a $-1$ phase to the all ones computational basis state, $\ket{1^\ell}$, and acts as identity on all other computational basis states. 

\begin{definition}
    Given a subset of $n$ qubits, $I\subset[n]$, define the \emph{$I$-controlled-$Z$}, $C^I Z$, as the $n$ qubit unitary that acts as $C^{(\abs{I})}Z$ on the qubits in $I$ and identity elsewhere, $C^I Z\coloneqq C^{(\abs{I})}Z\vert_I\otimes \eye\vert_{[n]\setminus I}$.
\end{definition}

\begin{definition}
    Given a collection of subsets of $n$ qubits, $\mcF\subseteq \powerset{[n]}$, define the \emph{$\mcF$-controlled-$Z$} operator as the circuit consisting of $C^IZ$ operators for each $I\in\mcF$, $C^\mcF Z\coloneqq\prod_{I\in\mcF}C^I Z$.
\end{definition}
Note that $C^\mcF Z$ is well-defined as each of the $C^IZ$ operators commute with each other. Further, as $(C^I Z)^2=\eye$ for each $I\subseteq[n]$, $(C^\mcF Z)^2=\eye$, as well. In particular, $C^\mcF Z^\dagger = C^\mcF Z$.

\begin{lemma}[Action of $C^{I}Z$ on $\mcP_n$]\label{lem: multi-controlled-Z action}
    \hspace{0em}
    \begin{enumerate}
        \item For every $i\in [n]$, $(C^{I}Z) Z_i (C^IZ) = Z_i$.
        \item For every $i\in [n]\setminus I$, $(C^{I}Z) X_i (C^IZ) = X_i$.
        \item For every $i\in I$, $(C^{I}Z) X_i (C^IZ) = X_i C^{I\setminus\br{i}}Z$.
    \end{enumerate}
\end{lemma}
\begin{proof}
    1 and 2 are trivial. For 3, consider without loss of generality $(C^{(\ell)}Z)X_1 (C^{(\ell)}Z)$. We compute
    \begin{align*}
        (C^{(\ell)}Z) X_1 (C^{(\ell)}Z) &= \left(\ketbra{0}\otimes \eye+\ketbra{1}\otimes C^{(\ell-1)}Z\right) X_1 \left(\ketbra{0}\otimes \eye+\ketbra{1}\otimes C^{(\ell-1)}Z\right),\\
        &= \ketbra{0}{1}\otimes C^{(\ell-1)}Z + \ketbra{1}{0}\otimes C^{(\ell-1)}Z,\\
        &= X\otimes C^{(\ell-1)}Z,\\
        &= X_1C^{(\ell-1)}Z\vert_{[\ell]\setminus\br{1}}.
    \end{align*}
\end{proof}

Let $\mcF\subseteq\powerset{[n]}$ be a collection of subsets of qubits. Given a qubit $i\in[n]$, the collection $\mcF_{\sim i}\subseteq\powerset{[n]}$ is defined as
\begin{equation}
    \mcF_{\sim i} \coloneqq \br{I\setminus \br{i}\Bigmid I\in\mcF, i\in I}.
\end{equation}
That is, $\mcF_{\sim i}$ consists of the sets in $\mcF$ that contain $i$, but with $i$ removed. Note that for a single subset, $I\subseteq[n]$, $\br{I}_{\sim i}$ is equal to $I\setminus\br{i}$ if $i\in I$ and empty otherwise.

\begin{lemma}[Action of $C^{\mcF}Z$ on $\mcP_n$]\label{lem: controlled-Z circuit action}
    \hspace{0em}
    \begin{enumerate}
        \item For every $i\in [n]$, $(C^{\mcF}Z) Z_i (C^\mcF Z) = Z_i$.
        \item For every $i\in [n]$, $(C^{\mcF}Z) X_i (C^\mcF Z) = X_i C^{\mcF_{\sim i}}Z$.
    \end{enumerate}
\end{lemma}
\begin{proof}
    1 is trivial. Let $i\in [n]$. We prove the identity in 2 by induction on $\abs{\mcF}$. Clearly if $\mcF=\emptyset$ then $\mcF_{\sim i}=\emptyset$ and the identity is true. Suppose now the identity holds for all $\mcG$ with $\abs{G}=m\geq 0$, and consider $\mcF\subseteq\powerset{[n]}$ with $\abs{\mcF}=m+1$. Pick an arbitrary $I\in\mcF$.
    \begin{align*}
        (C^\mcF Z) X_i (C^\mcF Z) &= (C^{\mcF\setminus I} Z) (C^I Z) X_i (C^I Z)(C^{\mcF\setminus I}Z), \\
        \since{\cref{lem: multi-controlled-Z action}}&=\begin{cases}
            (C^{\mcF\setminus I} Z) X_i (C^{I\setminus\br{i}} Z)(C^{\mcF\setminus I}Z), &\text{if }i\in I\\
            (C^{\mcF\setminus I} Z) X_i (C^{\mcF\setminus I}Z), &\text{otherwise}
        \end{cases},\\
        \since{Def. of $\br{I}_{\sim i}$ and commuting operators}&=(C^{\mcF\setminus I} Z) X_i (C^{\mcF\setminus I}Z) (C^{\br{I}_{\sim i}} Z),\\
        \since{I.H.}&= X_i (C^{(\mcF\setminus I)_{\sim i}}Z) C^{\br{I}_{\sim i}}Z, \\
        &= X_i (C^{\mcF_{\sim i}}Z).
    \end{align*}
\end{proof}
\begin{corollary}
    For $k_\mcF\coloneqq \max_{I\in\mcF} \abs{I}$, $C^\mcF Z\in\Cl{k_\mcF-1}$.
\end{corollary}

\section{Classical and quantum codes from hypercubes}\label{sec:codes and hypercubes}

\subsection{Classical Reed--Muller codes}\label{sec:classical-RM-codes}
We now proceed to define the Reed--Muller code family $RM(r,m)$ using the structure of the $m$-dimensional hypercube. 
The presentation we give here is close to the classic description of RM codes that appears in \cite[Ch.13]{MS77} or \cite{AK98}. The key difference arises because we use the generating set
of the code formed by the (indicator vectors) of the standard subcubes while the cited works rely on a
larger set formed by all the flats in the  geometry $AG(m,2)$. 
We find our presentation is more intuitive when describing transversal logic operations on \emph{quantum} Reed--Muller codes.

We will denote the binary field by $\FF\coloneqq\FF_2$, throughout, to avoid confusing the standard $n$-dimensional binary vector space, $\FF^n$, with the space of bit strings, $\ZZ_2^m$, which we use to index elements of the hypercube. As is often the case with RM codes, $n=2^m$ below.

Recall that a binary linear code of length $n$, dimension $k$, and distance $d$ is defined as a $k$-dimensional subspace, $V\subseteq \FF^n$, whose shortest non-zero vector has Hamming weight $\abs{w}=d$. Consider now the $n\coloneqq 2^m$-dimensional vector space $\br{f\colon \ZZ_2^m\rightarrow\FF_2}$ defined as binary functions on the hypercube, that is, the space of length-$n$ vectors whose coordinates are indexed by the elements of $\ZZ_2^m$. Given a subcube, $A\subcubeeq\ZZ_2^m$, the indicator function of $A$ is defined via
\begin{equation}
    \indicator{A}(x) \coloneqq\begin{cases}
        1, &\text{if } x\in A\\
        0, &\text{otherwise}.
    \end{cases}
\end{equation}

\begin{definition}[$RM(r,m)$]\label{def:classical-RM}
    For $r\in\br{-1,0,\dots, m}$ let $B_r$ denote the set of standard $\ell$-cubes where $m-r\leq \ell\leq m$, i.e., the subsets of $\ZZ_2^m$ given by $B_{r}\coloneqq\br{\standard{J}\mid J\subseteq S, \abs{J}\geq m-r}$. By convention, $B_{-1}\coloneqq\emptyset$. The Reed--Muller code of \emph{order} $r$, denoted $RM(r,m)$, is defined as the linear code generated by the indicator functions of $B_r$:
    \begin{equation*}
        RM(r,m)\coloneqq \br{\sum_{A\in B_r} c_A \indicator{A}\biggmid c_A\in\FF}.
    \end{equation*}
The code $RM(r,m)$ has length $n=2^m,$ dimension $\sum_{i=0}^r\binom mi$, and minimum distance $2^{m-r}$. 
\end{definition}

\begin{remark} \hspace{0em}

\begin{enumerate}[label=(\alph*)]
    \item Recall that the standard definition of RM codes relies on polynomial evaluation: the code $RM(r,m)$ 
is the set of evaluations of all $m$-variate polynomials $f(x_1,\dots,x_m)$ of degree at most $r$ on all possible bit strings $\bits{m}$. These two definitions
are equivalent. Consider the following mapping: send the standard subcube, $\standard{J}$, to the monomial $\standard{J}\mapsto x_{S\setminus J}\coloneqq\prod_{i\notin J}x_i$. Then (up to string reversal) $\indicator{\standard{J}}=\evaluate(x_{S\setminus J})$ and the two definitions of $RM(r,m)$ yield the same code by linearity. 
    \item The inclusion of $r=-1$ in \cref{def:classical-RM} deviates from the standard presentation of RM codes, which assumes that
$0\le r\le m$. This extension is convenient for our arguments below.
\end{enumerate}
\end{remark}

While the generating set $\mcI(B_r)\coloneqq\br{\indicator{A}\mid A\in B_r}$ is formed on linearly independent vectors and thus constitutes a genuine basis of the code, for stabilizer codes we will prefer a more operationally-useful (albeit redundant) generating set. For $r\in\br{-1,0,\dots, m}$ let $M_r$ denote the set of all $(m-r)$-cubes, 
\begin{align}
    M_r\coloneqq \br{x+\standard{J}\mid x\in\ZZ_2^m, J\subseteq S, \abs{J} = m-r}.
\end{align}
That is, $M_r$ contains all subcubes of dimension \emph{exactly} equal to $m-r$, whereas $B_r$ contains all \emph{standard} subcubes with dimension \emph{greater than or equal to} $m-r$. In particular, every subcube in $M_r$ contains precisely $2^{m-r}$ elements of $\ZZ_2^m$, whereas a subcube in $B_r$ contains $2^i$ elements for some $i\in \br{m-r,\dots, m}$. We will show later that $\mcI(M_r)$ is a valid generating set for $RM(r,m)$. First, we begin with some useful results on subcubes.

\begin{lemma}\label{prop:counts} We have  
   $
    |M_r|=2^{r}\binom mr.
    $
\end{lemma}
\begin{proof} The count of $(m-r)$-cubes follows since they are in 1--1 correspondence with the cosets of the standard subcubes (see the proof of \cref{lemma: all subcubes} below).
\end{proof}

A standard cube $\standard J$ has a unique element of minimum Hamming weight, namely the origin.
This property obviously extends to all of the cosets of $\standard J$, as shown in the next lemma.
\begin{lemma}\label{lemma: all subcubes}  Let $A\subcubeeq\ZZ_2^m$ be an arbitrary subcube of type $J\subseteq S$. It contains
a unique element, $x$, such that $|x|<|y|$ for all $y\in A\setminus\br{x}$.
\end{lemma}
\begin{proof}
Consider the subcubes $A$ of type $J$. We claim that each of them can be constructed as $x+\standard J$, where $x$ depends on $A$ and satisfies $\supp(x)\cap J=\emptyset$. Indeed, there are $2^{m-|J|}$ such vectors $x$, and for $x\ne x'$, the subcubes (cosets) $x+\standard J$ and $x'+\standard J$ are disjoint. Alternatively, any $y$ that has a nonempty overlap with $J$ is contained in the cube $x+\standard J$, where $x$ is obtained from $y$ by replacing the ones located in $\supp(y)\cap J$ with zeros. Moreover, each $x$ as defined above clearly is a unique element of minimum Hamming weight in its coset.
\end{proof}


We will denote the support of the minimal-weight element of $A$ by $I_A\coloneqq\supp(x)$.
In particular, we note that for all $y\in A$, $y_i=1$ for all $i\in I_A$ and $y_i=0$ for all $i\notin I_A\cup J$.
$y\in A$. We also note that $I_A\cap J=\emptyset$.

Using the set $I_A$, we can give a decomposition for the indicator function of an arbitrary subcube in terms of the indicator functions for standard subcubes.

\begin{restatable}{lemma}{indicatordecomposition}\label{lem: decomposition of subcube indicators}
 Let $A\subcubeeq\ZZ_2^m$ be a subcube of type $J$. The indicator $\indicator{A}$ can be decomposed into indicator functions of standard subcubes as 
    \begin{equation}\label{eq: subcube decomposition}
        \indicator{A} = \sum_{I\subseteq I_A} \indicator{\standard{I\cup J}}.
    \end{equation}
\end{restatable}
\begin{proof}
    Two proofs are given in \cref{app: indicator proofs}: \hyperlink{pf: indicator decomposition}{Proof 1}, \hyperlink{pf: indicator decomposition alternate}{Proof 2}.
\end{proof}

We now show that the indicator functions corresponding to $(m-r)$-cubes do, in fact, generate $RM(r,m)$:

\begin{lemma}\label{lem: minimal subcubes generate code}
    The indicator functions of $M_r$, $\mcI(M_r)$, form a redundant generating set for $RM(r,m)$, i.e., $RM(r,m)=\langle \mcI(M_r)\rangle$. Further, each function in $\mcI (M_r)$ is a minimum-weight codeword of $RM(r,m)$.
\end{lemma}
\begin{proof}
    ($\subseteq$) Suppose $\standard{J}$ is a standard subcube with $\abs{J}\geq m-r$, so that $\indicator{\standard{J}}$ is a generator for $RM(r,m)$. Let $J'\subseteq J$ be any subset of $J$ of size $\abs{J'}=m-r$. Then each  coset of $\standard{J'}$ in $\standard{J}$, $A\in\standard{J}/\standard{J'}$, is an $(m-r)$-cube in $\ZZ_2^m$, and the collection of these subcubes $\standard{J}/\standard{J'}$ forms a disjoint cover of $\standard{J}$. Thus, $\indicator{\standard{J}}=\sum_{A\in\standard{J}/\standard{J'}}\indicator{A}$. As each generator of $RM(r,m)$ can be decomposed into indicators of $(m-r)$-cubes, the desired result holds. 

    ($\supseteq$) Suppose $A$ is an arbitrary subcube with $\dim A=m-r$. By \cref{lem: decomposition of subcube indicators} we have
    \begin{equation}
        \indicator{A}=\sum_{I\subseteq I_A}\indicator{\standard{I\cup J}}.
    \end{equation}
    As each $I\cup J$ appearing in the summation has $\abs{I\cup J}\geq \abs{J}=m-r$, $\indicator{A}\in RM(r,m)$, by definition.
    
    The claim that $\mcI(M_r)$ is a redundant generating set follows by \cref{prop:counts}. Because the distance of $RM(r,m)$ is $2^{m-r}$, the statement that elements of $\indicator{M_r}$ have minimal-weight in $RM(r,m)$ is true by definition of $M_r$.
\end{proof}
\newcommand{\dstirling}[2]{\genfrac{[}{]}{0pt}{1}{#1}{#2}}
We note that Theorem 13.12 in \cite{MS77} says that the set of \emph{all} minimum-weight codewords of $RM(r,m)$--- which correspond to incidence vectors of the set $F_r$ of $(m-r)$-dimensional flats in $AG(m,2)$--- generates the code $RM(r,m)$. The content of \cref{lem: minimal subcubes generate code} is therefore that a smaller class of minimum-weight codewords--- those corresponding to $(m-r)$-cubes--- suffices to generate $RM(r,m)$. We mention in passing that a different redundant set of minimum-weight codewords--- those corresponding to $(m-r)$-dimensional subspaces--- also suffices to generate $RM(r,m)$ \cite{assmus1996berman,AK98}.

Finally, we recall some well-known facts about Reed--Muller codes, which make them particular useful from the perspective of stabilizer codes. Denote by $RM(r,m)^\bot$ the dual code 
of $RM(r,m)$, i.e., the set $\{y\in\ZZ_2^n\mid (y,x)=0\; \forall x\in RM(r,m)\}$, where $(\cdot,\cdot)$ is the dot product over $\ZZ_2^n$. In terms of the hypercube, $RM(r,m)^\bot$ is the set of subsets $D\subset \ZZ_2^m$ 
such that $\abs{D\cap B}\equiv 0\,(\text{mod\,}2)$ for all subsets $B$ that support codewords of $RM(r,m)$.
Using the language of the generating sets, we state and prove a standard description of $RM(r,m)^\bot$ \cite[Thm. 13.4]{MS77}.
\begin{lemma}\label{fact:dual-of-RM} The dual of $RM(r,m)$ is given by $RM(m-r-1,m)=RM(r,m)^\perp$.
\end{lemma}
\begin{proof}
    As is typical when proving the duality relation for RM codes, we begin by showing $RM(r,m)\subseteq RM(m-r-1,m)^\perp$.
    
    Let $A=x+\standard{J}$ be a $(m-r)$-cube and $B=y+\standard{K}$ an $(r+1)$-cube so that $\indicator{A}$ and $\indicator{B}$ are generators of $RM(r,m)$ and $RM(m-r-1,m)$, respectively.
    We wish to prove that $\abs{\indicator{A}\cdot\indicator{B}}\equiv 0\pmod{2}$, which happens precisely when $\abs{A\cap B}\equiv0\pmod{2}$.
    If $A\cap B$ is empty then clearly this is true. On the other hand, if $A\cap B$ is non-empty then their intersection is a subcube $A\cap B=z+\standard{J\cap K}$ for some $z\in\ZZ_2^m$. $\abs{A\cap B}=2^{\abs{J\cap K}}$, which is even as long as $J\cap K\neq\emptyset$. As $\abs{J}+\abs{K}>(m-r)-r>m$, but $J$ and $K$ are both subsets of $S$, which has $m$ elements, clearly $J\cap K\neq\emptyset$.

    The prove equality it suffices to show that the dimensions of $RM(r,m)^\perp$ and $RM(m-r-1,m)$. This is straightforward to show by noting that $\dim RM(r,m)^\perp = 2^m-\dim RM(r,m)$.
\end{proof}

The following is also obvious from the definition of $RM(r,m)$ using standard subcubes. 
\begin{fact}\label{fact:containment-in-RM}
    For integers $q\leq r$, $RM(q,m)\subseteq RM(r,m)$.
\end{fact}


In order to understand the logic of transversal operators coming from higher levels of the Clifford Hierarchy, we will need to consider indicator functions whose outputs are integers, rather than bits.

For fixed $m\in\NN$, consider the additive group of integer-valued functions on the Boolean hypercube, $\ZZ[\ZZ_2^m]\coloneqq\br{f\colon \ZZ_2^m\rightarrow \ZZ}$ (which is also closed under scalar multiplication by $\ZZ$). Let $A\subcubeeq \ZZ_2^m$ be a subcube.
With a view toward describing logical operators for quantum RM codes, define the \emph{unsigned indicator function} on $A$, $\barindicator{A}\in\ZZ[\ZZ_2^m]$, as $\barindicator{A}(x)\coloneqq 1$ if $x\in A$ and 0 otherwise. Note that we are using the same notation, $\indicator{A}$, to denote both the $\FF$-valued and $\ZZ$-valued indicator functions on $A$. The difference, though minor, is that $2\cdot\indicator{A}(x)=0$ for all $x\in\ZZ_2^m$ in the $\FF$-valued case, whereas this does not hold in the $\ZZ$-valued case. Nonetheless, the codomain of the indicator function should be clear from context. In particular, we will never multiply an $\FF$-valued function by a scalar.

Further, define the \emph{signed indicator function} on $A$, $\tildeindicator{A}\in\ZZ[\ZZ_2^m]$, as follows:
\begin{equation}
    \tildeindicator{A}(x)\coloneqq\begin{cases}
        (-1)^\abs{x}, & x\in A\\
        0, &\text{otherwise}.
    \end{cases}
\end{equation}
That is, $\tildeindicator{A}$ is the indicator function on $A$ except that odd-weight elements of $A$ have a minus sign.

The following is a generalization of \cref{lem: decomposition of subcube indicators} to the case of $\ZZ$-value indicator functions:
\hypertarget{lem: bar indicator expression return}{}
\begin{restatable}{lemma}{barindicatordecomposition}\label{lem: bar indicator expression}
    Let $A\subcubeeq\ZZ_2^m$ be a subcube of type $J$. The unsigned indicator function can be decomposed as
    \begin{equation}\label{eq: bar indicator equality}
        \barindicator{A} = \sum_{I\subseteq I_A} \barindicator{\standard{I\cup J}} - \sum_{{ i=1}}^\abs{I_A}\sum_{I\subseteq I_A\colon \abs{I}=i} 2^i\cdot\barindicator{e_{I_A\setminus I}+\standard{I\cup J}}.
    \end{equation}
\end{restatable}
\begin{proof}
    \hyperlink{pf: bar indicator expression}{See \cref{app: indicator proofs}.}
\end{proof}
Note that if we take the functions $\indicator{A}$ modulo 2, we recover the $\FF$-valued indicator functions used earlier when discussing classical RM codes.
If we take the right-hand side of \cref{eq: bar indicator equality} modulo 2, the large summation term vanishes as every function in the sum is scaled by $2^i$ where $i\geq 1$. Thus, \cref{lem: bar indicator expression} exactly reproduces \cref{lem: decomposition of subcube indicators} when the functions are taken modulo 2.

An analogous version of \cref{lem: bar indicator expression} is true for the signed indicator functions, as well:

\begin{restatable}{lemma}{tildeindicatordecomposition}\label{lem: tilde indicator expression}
    Let $A\subcubeeq\ZZ_2^m$ be a subcube of type $J$. The signed indicator function can be decomposed as
    \begin{equation}\label{eq: tilde indicator equality}
        \tildeindicator{A} = \sum_{I\subseteq I_A} \tildeindicator{\standard{I\cup J}} - \sum_{{ i=1}}^\abs{I_A}\sum_{I\subseteq I_A\colon \abs{I}=i} 2^i\cdot\tildeindicator{e_{I_A\setminus I}+\standard{I\cup J}}.
    \end{equation}
\end{restatable}
\begin{proof}
    The proof is unchanged from the proof of \cref{lem: bar indicator expression} \hyperlink{pf: bar indicator expression}{here}.
\end{proof}

Lastly, the following result relates the unsigned indicator function on a particular standard subcube, $\standard{K}$, to \emph{signed} indicator functions on standard subcubes contained \emph{within} $\standard{K}$.

\begin{restatable}{lemma}{strongtransversalindicator}{\label{lem: strong transversal indicator}}
    For $m\in\NN$,
    \begin{equation}
        \barindicator{\ZZ_2^m} = \sum_{i=0}^m\sum_{J\subseteq S\colon \abs{J}=i} 2^{m-i}(-1)^i \cdot \tildeindicator{\standard{J}}.
    \end{equation}
    More generally, for $m\in\NN$ and $K\subseteq[S]$,
    \begin{equation}\label{eq: standard subcube bar indicator decomposition}
        \barindicator{\standard{K}} = \sum_{i=0}^{\abs{K}}\sum_{J\subseteq K\colon \abs{J}=i} 2^{\abs{K}-i}(-1)^i \cdot \tildeindicator{\standard{J}}.
    \end{equation}
\end{restatable}
\begin{proof}
    \hyperlink{pf: strong transversal indicator}{See \cref{app: indicator proofs}.}
\end{proof}

The utility of these results will become clear in \cref{sec: signed logic} and \cref{sec: unsigned logic}, where will will use them to prove operator decomposition lemmas for transversal diagonal gates in the Clifford Hierarchy. We will restate them later as needed.

\subsection{Quantum Reed--Muller codes}\label{sec: QRM-codes}

\cref{fact:dual-of-RM} and \cref{fact:containment-in-RM} give an obvious way to define a quantum Reed--Muller code of the CSS type. To define it,
we need two classical codes, $C_1$ and $C_2$, such that $C_1^\perp\subseteq C_2$. In the case of Reed--Muller codes, we will choose the codes as follows:
   \begin{align*}
   \begin{array}{c@{\hspace*{.3in}}c}
     Z\emph{ logicals: }C_1=RM(m-q-1,m) & X\emph{ stabilizers: }C_1^\bot=RM(q,m)\\
           \cup      & \cap\\
     Z\emph{ stabilizers: }C_2^\bot=RM(m-r-1,m) & X\emph{ logicals: }C_2=RM(r,m),
    \end{array}
   \end{align*}
where for strict inclusions to hold we take $q<r\le m$. For the definitions below in this section we note 
the dimension of the subcubes that give rise to minimum-weight codewords of the codes:
    \begin{equation}\label{eq: min dist}
  \begin{array}{cccc}
     C_1^\bot&C_2&C_2^\bot&C_1\\
     m-q&m-r&r+1&q+1
  \end{array}.
 \end{equation}

\begin{definition}\label{def: QRM}
    For integers $0\leq q\leq r\leq m$, the quantum Reed--Muller code $QRM_m(q,r)$ of order $(q,r)$ is defined to be the code $\CSS(RM(m-q-1,m),RM(r,m))$. The code encodes $\kappa_m$ qubits into $n=2^m$ qubits, where 
    $\kappa_m=\dim(C_2)-\dim(C_1^\bot)=\sum_{i=q+1}^r\binom mi$.
    \end{definition}
Although algebraically complete, this definition can be intuitively difficult to work with. Instead, we will use the $m$-dimensional hypercube to give a geometric interpretation of $QRM_m(q,r)$. Consider the $n=2^m$ qubit space, where we index the qubits via elements of $\ZZ_2^m$. Given a subcube $A\subcubeeq\ZZ_2^m$, we define an $n$-qubit Pauli $X$-type operator $X_A$ that acts as $X$ on qubits in $A$ and as $\eye$ elsewhere,
\begin{equation*}
    (X_A)_x = \begin{cases}
        X, &\text{if } x\in A\\
        \eye, &\text{otherwise}
    \end{cases},
\end{equation*}
which in the notation of \cref{sec:Pauli group and quantum codes} says that $X_A= X(\indicator{A})$. For any single-qubit unitary, $U\in\unitary(2)$, we define the $n$-qubit operator $U_A$ in the analogous way. Now, the results of \cref{sec:classical-RM-codes} imply that the $X$ stabilizer group, given by $RM(q,m)$, can be generated by the $X_A$ operators acting on subcubes of dimension exactly $m-q$. Likewise, the $Z$ stabilizer group, given by $RM(m-r-1,m)$ can be generated by the $Z_A$ operators acting on subcubes of dimension exactly equal to $m-(m-r-1)=r+1$. Thus, we call the definition of quantum RM codes given in the introduction, which is equivalent to \cref{def: QRM}:

\quantumrmcode*

Using the duality of Reed--Muller codes, we can also define generators of the undetectable $X$ and $Z$ errors using subcubes of $\ZZ_2^m$ via \cref{lem: minimal subcubes generate code}:
\begin{equation}
    \begin{aligned}
        N_X&\coloneqq \br{X_A\Bigmid A \text{ is an $(m-r)$-cube}},\\
        N_Z&\coloneqq \br{Z_A\Bigmid A \text{ is a $(q+1)$-cube}}.
    \end{aligned}
\end{equation}
Using the standard definition of RM codes, combined with the stabilizers the following give independent bases for the logical $X$ and $Z$ operators of $QRM_m(q,r)$:
\begin{equation}\label{eq: LX-LZ}
    \begin{aligned}
        L_X'&\coloneqq \br{X_A\Bigmid A \text{ is a standard $\ell$-cube for $m-r\leq\ell\leq m-q-1$}},\\
        L_Z&\coloneqq \br{Z_A\Bigmid A \text{ is a standard $\ell$-cube for $q+1\leq\ell\leq r$}}.
    \end{aligned}
\end{equation}

We can rewrite these sets to better reflect the symmetry between them:

\begin{equation}
    \begin{aligned}
        L_Z&\coloneqq \big\{Z_{\standard{J}}&\Bigmid J\subseteq S,\; q+1\leq\abs{J}\leq r\big\},\\
        L_X'&\coloneqq \big\{X_{\standard{S\setminus J}}&\Bigmid J\subseteq S,\; q+1\leq\abs{J}\leq r \big\}.
    \end{aligned}
\end{equation}

In order to understanding logic on $QRM_m(q,r)$ we must describe a symplectic basis for the logical Pauli operators. Unfortunately, the sets $L_Z$ and $L_X'$ defined above \emph{are not} symplectic except in the case $q=r-1$. To see this, take $J\subseteq S$ to be any subset of size $q+1$ and $J'\supset J$ to be a superset of $J$ of size at most $r$ (which only exists in the case $q<r-1$). By construction, both $X_{\standard{S\setminus J}}$ and $X_{\standard{S\setminus J'}}$ are in the set $L_X'$ and they both anti-commute with $Z_{\standard{J}}\in L_Z$, so $\{L_Z, L_X'\}$ cannot be a symplectic basis for $QRM_m(q,r)$.

In some sense, this problem arises because all elements in $\{L_Z, L_X\}$ are guaranteed to overlap at least on $0^m\in\ZZ_2^m$. To fix this, we will shift the $X$ logical Pauli operators away from the $0^m$ vertex. For $J\subseteq S$ let $e_J\coloneqq \sum_{i\in J} e_i$ denote the incidence bit string corresponding to $J$, and consider the following sets of logical Pauli operators:
\begin{equation}
    \begin{aligned}
        L_Z&\coloneqq \Big\{Z_{\standard{J}}&\Bigmid J\subseteq S,\; q+1\leq\abs{J}\leq r\Big\},\\
        L_X&\coloneqq \Big\{X_{e_J +\standard{S\setminus J}}&\Bigmid J\subseteq S,\; q+1\leq\abs{J}\leq r\Big\}.
    \end{aligned}
\end{equation}

\begin{lemma}\label{lem: standard basis for QRM}
    $\{L_Z, L_X\}$ is a symplectic basis for $QRM_m(q,r)$. In particular, operators
    $Z_{\standard{J}}\in L_Z$ and $X_{e_K+\standard{S\setminus K}}\in L_X$ anti-commute if and only if $J=K$.
\end{lemma}
\begin{proof}
    ($\Rightarrow$) Rephrasing the claim, we will show that if $J\neq K$, then the operators commute. Consider the set of qubits that are acted on by both operators and that is given by $A\coloneqq \standard{J}\cap(e_K+\standard{S\setminus K})$. We proceed in cases:
    \begin{enumerate}[label=\textbf{\Roman*.}, leftmargin=*]
        \item ($J\subset K$)\footnote{Note that this case can only occur when $q<r-1$.} Suppose that $x\in A\neq\emptyset$. Then there exists $J'\subseteq J$ and $M\subseteq (S\setminus K)$ such that $x=e_{J'}=e_K+e_M$, implying that $e_{J'}+e_K = e_M$. Now since $J\subset K$, we are guaranteed that ${J'}\cap M=\emptyset$, and for the equality to hold it must be that $e_{J'}+e_K=e_M=0$. Thus, we have that $e_{J'}=e_K$. But by assumption, $K$ is strictly larger than $J'$, so this equation cannot be satisfied and no such $x$ can exist. Thus $A=\emptyset$ and the operators commute.
        
        \item ($J\setminus K \neq\emptyset$) Recall that either $A=\emptyset$ or else there is an $x\in\ZZ_2^m$ such that $A = x+\standard{J\cap (S\setminus K)}$. We are guaranteed in this case that $J\cap (S\setminus K)\neq\emptyset$, so $\abs{A}\in\br{0,2^\abs{J\cap (S\setminus K)}}$ is even and the operators commute.
    \end{enumerate}
    
    ($\Leftarrow$) Assuming $J=K$, we have that $\standard{J}\cap(e_J+\standard{S\setminus J})=\br{e_J}$, implying that $Z_{\standard{J}}$ and $X_{\standard{S\setminus K}}$ have overlapping support on a single qubit and therefore anti-commute.
\end{proof}
\cref{lem: standard basis for QRM} allows us to use the subsets, $J\subseteq S$, with $q+1\leq \abs{J}\leq r$ to uniquely index the logical qubits of the $QRM_m(q,r)$ code as mentioned previously in \cref{sec: geometry}:
\QRMdef*


\section{Transversal logic via subcube operators}\label{sec: sufficiency conditions}
As always, suppose $0\leq q < r\leq m$ are non-negative integers, and consider the quantum code $QRM_m(q,r)$.
In this section we prove necessary and sufficient conditions for subcube operators to act on $QRM_m(q,r)$ as either Clifford stabilizers or undetectable Clifford errors. Recall that $QRM_m(q,r)$ has $2^m$ physical qubits indexed by the elements of $\ZZ_2^m$. Given a subcube $A\subcubeeq\ZZ_2^m$ and a non-negative integer $k\in\ZZ_{\geq 0}$, we defined the unsigned and signed $\Z{k}$ operators on $A$, respectively, as
\begin{align}
    \left(\Z{k}_A\right)_x&\coloneqq \begin{cases}
        \Z{k}, &\text{ if } x\in A \\
        \eye, &\text{ otherwise}.
    \end{cases}\\
    \left(\tildeZ{k}_A\right)_x&\coloneqq \begin{cases}
        \Z{k}, &\text{ if $x\in A$ and $\abs{x}$ is even} \\
        \Z{k}^\dagger, &\text{ if $x\in A$ and $\abs{x}$ is odd} \\
        \eye, &\text{ otherwise}.
    \end{cases}
\end{align}

When $k=0$ and $\Z{0}_A = \tildeZ{0}_A=Z_A$ is a $Z$ operator acting on $A$, the following lemma is a direct consequence of the definition of quantum Reed--Muller codes, \cref{def: QRM}, and the equivalence of the constructions of classical Reed--Muller codes given in \cref{sec:classical-RM-codes}.
\begin{lemma}\label{lem: base cases}
    Consider $QRM_m(q,r)$ and let $A\subcubeeq\ZZ_2^m$ be a subcube. The following are true:
    \begin{itemize}
        \item $Z_A\in\stabs{0}$ if and only $\dim A\geq r+1$.
        \item $Z_A\in\logs{0}$ if and only if $\dim A\geq q+1$.
        \item $X_A\in\stabs{0}$ if and only if $\dim A\geq m-q$.
        \item $X_A\in\logs{0}$ if and only if $\dim A\geq m-r$.
    \end{itemize}
\end{lemma}
The aim of the present section is to prove the following generalization of \cref{lem: base cases} to $\Z{k}_A$ and $\tildeZ{k}_A$ operators for arbitrary values of $k\in\ZZ_{\geq 0}$,
\begin{restatable}{theorem}{signedlogicconditions}\label{thm: subcube dimension implies logic}
    Consider $QRM_m(q,r)$, $k\in\ZZ_{\geq 0}$, and $A\subcubeeq\ZZ_2^m$ be a subcube. The following are true:
    \begin{enumerate}
        \item \emph{(\cref{clm: trivial logic signed})} $\tildeZ{k}_A\in\stabs{k}$ if and only if $\dim A\geq (k+1)r+1$.
        \item \emph{(\cref{clm: non-trivial logic signed})} $\tildeZ{k}_A\in\logs{k}$ if and only if $q+kr+1\leq\dim A\leq (k+1)r$.
        \item \emph{(\cref{clm: trivial logic unsigned})} $\Z{k}_A\in\stabs{k}$ if and only if $\dim A\geq (k+1)r+1$.
        \item \emph{(\cref{clm: non-trivial logic unsigned})} $\Z{k}_A\in\logs{k}$ if and only if $q+kr+1\leq\dim A\leq (k+1)r$.
    \end{enumerate}
\end{restatable}

Note that the conclusion of \cref{thm: subcube dimension implies logic} is limited to the claim that the operators 
$\tildeZ{k}_A$ and $\Z{k}_A$ preserve the code space of $QRM_m(q,r)$, while saying nothing about \emph{what} logic they perform when $q+kr+1\leq\dim A\leq (k+1)r$. We will detail their logical circuits below in \cref{sec: signed logic} and \cref{sec: unsigned logic}. In this section we will prove \cref{thm: subcube dimension implies logic}.

We will make frequent use of the following simple statement that details intersections of various subcubes in $\ZZ_2^m$.
\begin{lemma}\label{lem:overlap-guarantee}
    For $\ell\in\br{0,\dots,m}$, let $\mcB_{A,\ell}$ denote the collection of $\geq\ell$-cubes that have a non-trivial overlap with a subcube $A\subcubeeq\ZZ_2^m$: 
    \begin{equation}
        \mathcal{B}_{A,\ell}\coloneqq\br{B\mid \dim B\geq \ell, A\cap B\neq\emptyset}.
    \end{equation}
    For $p\ge 1$, $\dim A\cap B\geq p$ for \emph{every} $B\in\mcB_{A,\ell}$ if and only if $\dim A\geq m-\ell+p$.
\end{lemma}
\begin{proof}
    Without loss of generality assume that $A=\standard{J}$ for $J\subseteq S$. 

    ($\Rightarrow$) First note that $\abs{J}=\dim A\geq p$; otherwise $\dim A\cap B \leq \dim A < p$. Suppose for contradiction that $p\leq\abs{J}<m-\ell+p$. Define $K\subseteq S$ to be the union of $[m]\setminus J$ and any $p-1$ elements of $J$, so that $\abs{K}>m-(m-\ell+p)+p-1=\ell-1$. But then $\standard{K}$ is a $\geq \ell$-cube for which $\dim{A\cap B}=\abs{J\cap K} =p-1< p$.

    ($\Leftarrow$) For arbitrary $B\in\mcB_{A,\ell}$ there exists a $K\subseteq S$, $\abs{K}\geq \ell$, and a $w\in\ZZ_2^m$ such that $B=w+\standard{K}$. $A\cap B\neq\emptyset$ by definition, so there is a $w'\in\ZZ_2^m$ such that $A\cap B = w'+\standard{J\cap K}$ and $\dim A\cap B=\abs{J\cap K}$. Since $J$ and $K$ are both subsets of $[m]$ and $\abs{A}\geq m-\ell+p$, by the pigeonhole principle it must be that $\abs{J\cap K}\geq p$.
\end{proof}

Throughout this section we will consider conjugating $X$ operators by either $\tildeZ{k}_A$ or $\Z{k}_A$ operators. We will focus on cases for which the operator obtained through conjugation does not contain a global phase:
\begin{definition}
    For subcubes $A,B\subcubeeq\ZZ_2^m$ and $k\in\ZZ_{\geq 0}$, the operator $\tildeZ{k}_A X_B\tildeZ{k}_A^\dagger$ is said to be \emph{phase-free} if
    \begin{align}
        \tildeZ{k}_A X_B \tildeZ{k}_A^{\dagger} &= \tildeZ{k-1}_{A\cap B} X_B,\\
        &= X_B\tildeZ{k-1}_{A\cap B}^\dagger.  
    \end{align}
    Analogously, $\Z{k}_A X_B\Z{k}_A^\dagger$ is said to be \emph{phase-free} if
    \begin{align}
        \Z{k}_A X_B \Z{k}_A^{\dagger} &= \Z{k-1}_{A\cap B} X_B,\\
        &= X_B\Z{k-1}_{A\cap B}^\dagger.  
    \end{align}
\end{definition}

We will now prove \cref{clm: trivial logic signed} and \cref{clm: non-trivial logic signed}, which deal with the case of signed subcube operators. Proofs of the results for unsigned operators, \cref{clm: trivial logic unsigned} and \cref{clm: non-trivial logic unsigned}, appear in \cref{sec: unsigned dimension conditions}, as they
are nearly identical to the proofs in the signed operator case.

A direct consequence of the conjugation identities for $\Z{k}$ and $\Z{k}^\dagger$ given in \cref{lem:conjugation-identities} is the following:
\begin{fact}\label{fact: tildeZk-conjugation-identity}
    For subcubes $A,B\subcubeeq\ZZ_2^m$ with non-trivial intersection, $\tildeZ{k}_A X_B \tildeZ{k}_A^{\dagger}$ is phase-free if and only if $\dim A\cap B\geq 1$. 
\end{fact}
See \cref{fig:conjugation_identity} for a proof by illustration.

\begin{figure}[ht]
    \centering
    \includegraphics[width=\textwidth]{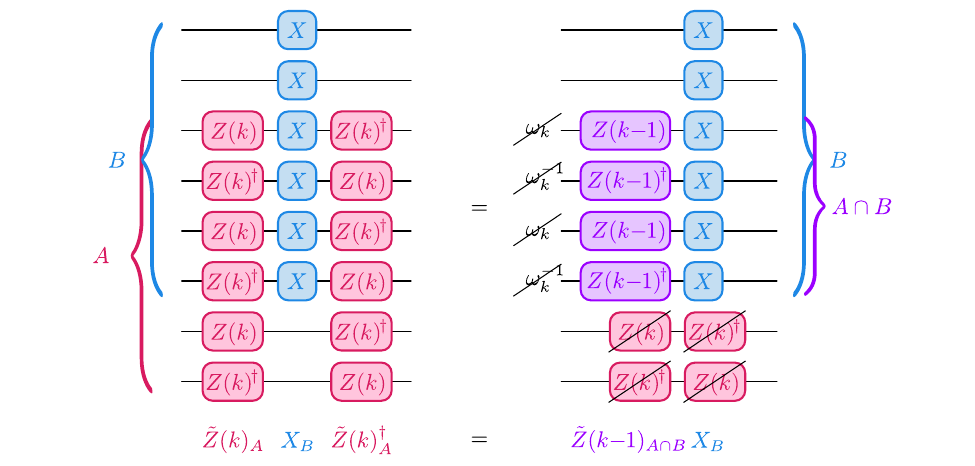}
    \caption{A visual ``proof'' of \cref{fact: tildeZk-conjugation-identity}}
    \label{fig:conjugation_identity}
\end{figure}

\begin{lemma}\label{lem: signed operator phase free condition}
    For an arbitrary subcube $A\subcubeeq\ZZ_2^m$, $\tildeZ{k}_A X_B \tildeZ{k}_A^{\dagger}$ is phase-free for every $X$ stabilizer generator $X_B$ of $QRM_m(q,r)$ if and only if $\dim A\geq q+1$.
\end{lemma}
\begin{proof}
    By definition of $QRM_m(q,r)$, the $X_B$'s that are stabilizer generators are precisely those with $\dim B=m-q$. By \cref{fact: tildeZk-conjugation-identity}, the statement of the lemma can be rephrased as follows: $\tildeZ{k}_A X_B \tildeZ{k}_A^{\dagger}$ is phase-free if and only if $\dim A\cap B\geq 1$ for every subcube $B$ that has non-trivial intersection with $A$ and satisfies $\dim B\geq m-q$. The desired result therefore holds by \cref{lem:overlap-guarantee}.
\end{proof}

\begin{claim}\label{clm: trivial logic signed}
    For $k\in\ZZ_{\geq 0}$ and a subcube $A$, $\tildeZ{k}_A$ is a level-$k$ Clifford stabilizer for $QRM_m(q,r)$ if and only if $\dim A\geq (k+1)r+1$.
\end{claim}
\begin{proof}
We will prove both parts of the claim by induction on $k$. When $k=0$, $\tildeZ{0}_A=Z_A$ and the statement is true by \cref{lem: base cases}. Let us  suppose that the statement is true for $k\geq 0$ and consider the statement for $k+1$.

    ($\Rightarrow$) Suppose for contradiction that there exists a subcube $A$ such that (1) $\tildeZ{k+1}_A\in\stabs{k+1}$, but (2) $\dim A\leq (k+2)r$. As $A$ is a subcube, there exist $x\in\ZZ_2^m$
    and $K\subseteq S$ ($\abs{K}=\dim A$), such that $A=x+\standard{K}$.

    Now, by assumption of $\tildeZ{k}_A\in\stabs{k+1}$, it must be true that for every $X$ logical subcube operator $X_B$, the operators $\tildeZ{k}_A X_B \tildeZ{k}_A^\dagger$ and  $X_B$ are equivalent. Given our assumption for the dimension of $A$, we will obtain a contradiction by constructing a logical $X_B$ for which $\tildeZ{k}_A X_B \tildeZ{k}_A^\dagger\not\equiv X_B$.

    Let $K^*\subseteq S$ be any subset of $S$ with $\abs{K^*}=m-r$ elements such that $S\setminus K\subseteq K^*$. Define the subcube $B\coloneqq x+\standard{K^*}$, so that $A\cap B = x +\standard{K\cap K^*}\neq \emptyset$. Using the conjugation identities for $\Z{k}$ and $\Z{k}^\dagger$ we have that $\tildeZ{k+1}_A X_B \tildeZ{k+1}_A^\dagger = \alpha \tildeZ{k}_{A\cap B} X_B$, where $\alpha$ is some global phase factor dependent on $k$, $A$, and $B$. This implies that for $\tildeZ{k+1}_A X_B \tildeZ{k+1}_A^\dagger \in\stabs{k}$ to be true it must be that $\tildeZ{k}_{A\cap B}\in\stabs{k}$, as otherwise $\alpha\tildeZ{k}_{A\cap B} X_B\ket\psi$ cannot equal $X_B\ket\psi$ for every code state $\ket\psi\in QRM_m(q,r)$. 
As it turns out, $\tildeZ{k}_{A\cap B}\in\stabs{k}$ contradicts our induction hypothesis that 
$\dim A\cap B\geq (k+1)r+1$. Indeed, we can upper bound this dimension as
    \begin{align}
        \dim A\cap B &= \abs{K\cap K^*}, \\
        &= \abs{K^*\setminus (S\setminus J)},\\
        \since{$S\setminus J\subseteq K$}&= \abs{K^*} - \abs{S\setminus J},\\
        &= (m - r)- (m -\abs{J}),\\
        &= \abs{J}-r,\\
        \since{$\dim A\leq (k+2)r$ by (2)} &\leq  (k+1)r,
    \end{align}
  Thus, $\tildeZ{k}_{A\cap B}\notin\stabs{k}$ implying that $\tildeZ{k+1}_A\notin \stabs{k+1}$.

    ($\Leftarrow$) Assume that $\tildeZ{k}_A$ is a level-$k$ Clifford stabilizer for all $A$ satisfying $\dim A\le (k+1)r+1$. Now
    suppose that $A$ is a subcube with $\dim A\geq (k+2)r+1$. Let $B$ be an arbitrary subcube for which $X_B$ is an undetectable $X$ error, which by \cref{lem: base cases} occurs if and only if $\dim B\geq m-r$. By \cref{lem:Clifford-stabilizer-equivalences}, the desired result, $\tildeZ{k+1}_A\in\stabs{k+1}$, holds if and only if $\tildeZ{k+1}_A X_B \tildeZ{k+1}_A^{\dagger}\equiv X_B$. Thus, we consider the operator $\tildeZ{k+1}_A X_B \tildeZ{k+1}_A^{\dagger}$.

    As $\dim A\geq r+1\geq q+1$, \cref{lem: signed operator phase free condition} implies that the operator is phase-free, and so $\tildeZ{k+1}_A X_B \tildeZ{k+1}_A^{\dagger} = \tildeZ{k}_{A\cap B} X_B$. Now notice that, since $\dim A\geq (k+2)r+1=m-(m-r)+(k+1)r+1$, by \cref{lem:overlap-guarantee} we have that $\dim A\cap B\geq (k+1)r+1$, so $\tildeZ{k}_{A\cap B}$ is a level-$k$ Clifford stabilizer for the code and the desired result holds by the induction hypothesis.
\end{proof}

\begin{claim}\label{clm: non-trivial logic signed}
    For $k\in\ZZ_{\geq 0}$ and a subcube $A$, $\tildeZ{k}_A$ is a level-$k$ undetectable Clifford error for $QRM_m(q,r)$ if and only if $\dim A\geq q+kr+1$.
\end{claim}
\begin{proof}
    By definition,  $\tildeZ{k}_A\in\errors{k}$ if and only if $\tildeZ{k}_A X_B \tildeZ{k}_A^\dagger \in \stabs{k-1}$ for every $X_B$ with $\dim B= m-q$. Let $X_B$ be an arbitrary stabilizer generator. Using the conjugation identities for $\Z{k}$ and $\Z{k}^\dagger$ we have that $\tildeZ{k}_A X_B \tildeZ{k}_A^\dagger = \alpha \tildeZ{k-1}_{A\cap B} X_B$, where $\alpha$ is some global phase factor dependent on $k$, $A$, and $B$. Since $X_B$ is a stabilizer we have that $\tildeZ{k}_A X_B \tildeZ{k}_A^\dagger\in\stabs{k}$ if and only if $\alpha \tildeZ{k-1}_{A\cap B}\in\stabs{k-1}$ for every subcube $B$ with $\dim B=m-q$.

    ($\Rightarrow$) We assume that $\alpha \tildeZ{k-1}_{A\cap B}\in\stabs{k-1}$ for every subcube $B$ with $\dim B=m-q$, and we seek to show that $\dim A\geq q+kr+1$. If the global phase factor $\alpha_k\neq 1$, then $\alpha \tildeZ{k-1}_{A\cap B}$ cannot fix the code space, so by \cref{lem: signed operator phase free condition} we have that $\dim A \geq q+1$ in order for $\tildeZ{k}_A X_B \tildeZ{k}_A^\dagger$ to be phase-free. Now, we must show that $\tildeZ{k-1}_{A\cap B}\in\stabs{k-1}$ for every $B$ such that $\dim B = m-q$. Using \cref{clm: trivial logic signed}, $\tildeZ{k-1}_{A\cap B}\in\stabs{k-1}$ if and only if $\dim A\cap B\geq kr+1$. By \cref{lem:overlap-guarantee} we have that $\dim A\cap B\geq kr+1$ for every $B$ with $\dim B = m-q$ only if $\dim A\geq m- (m-q)+kr+1 = q+kr+1$, as desired.

    ($\Leftarrow$) We assume that $\dim A\geq q+kr+1$, and we seek to show that $\alpha \tildeZ{k-1}_{A\cap B}\in\stabs{k-1}$ for every subcube $B$ with $\dim B=m-q$. As $k\geq 0$, $\dim A\geq q+1$ and \cref{lem: signed operator phase free condition} implies that $\tildeZ{k}_A X_B \tildeZ{k}_A^\dagger$ is phase-free, and so $\alpha=1$. By \cref{lem:overlap-guarantee}, since $\dim A\geq q+kr+1$ we have that $\dim A\cap B\geq kr+1$ for every $B$ with $\dim B = m-q$. \cref{clm: trivial logic signed} thus implies that $\tildeZ{k-1}_{A\cap B}\in\stabs{k-1}$, as desired.
\end{proof}


\section{Signed subcube operator logic}\label{sec: signed logic}
\cref{thm: subcube dimension implies logic} gives necessary and sufficient conditions for when a $\tildeZ{k}_A$ operator performs non-trivial logic on $QRM_m(q,r)$; the aim of this section is to determine the logical circuit implemented by a signed subcube operator. A simple corollary of \cref{thm: subcube dimension implies logic} gives one hint toward the structure of the logical circuits:
\begin{corollary}\label{cor: tilde operators are logically Hermitian}
    If $\tildeZ{k}_A\in\errors{k}$, then $\tildeZ{k}_A^2\in\stabs{k-1}$ and $\tildeZ{k}_A\equiv\tildeZ{k}_A^\dagger$.
\end{corollary}
\begin{proof}
    The first implication follows by \cref{thm: subcube dimension implies logic} since $\dim A\geq q+kr+1\geq \big((k-1)+1\big)+1$ and $\tildeZ{k}_A^2 = \tildeZ{k-1}_A$. Since $\tildeZ{k}_A$ is unitary, the logical involution property implies logical Hermiticity.
\end{proof}
If $\tildeZ{k}_A$ is a logical operator on the code space then by \cref{thm: subcube dimension implies logic} it is logically Hermitian. One may expect that, as a diagonal operator in the $k$-th level of the Clifford Hierarchy, such a $\tildeZ{k}_A$ would implement a logical diagonal operator in the $k$-th level, as well. The only diagonal $k$-th level Clifford Hierarchy operators that are Hermitian are circuits of multi-controlled-$Z$ gates where the number of controls is at most $k-1$ for any gate \cite{CGK17}. So, \cref{cor: tilde operators are logically Hermitian} seems to indicate that $\tildeZ{k}_A$ will implement logical multi-controlled-$Z$ circuits. As stated in \cref{sec: overview} this turns out to be the case.

Consider first the case of $k=0$ where $\tildeZ{0}_A=Z_A$ is the $Z$ operator acting on the subcube $A$. It is a simple consequence of the structure of classical Reed--Muller codes that $Z_A$ can be written as a product of operators acting on \emph{standard} subcubes, $Z_{\standard{J}}$, as the standard subcubes correspond precisely to the basis elements of Reed--Muller codes. It turns out that an analogous statement is true for $\tildeZ{k}_A$ operators: Every $\tildeZ{k}_A$ operator can be written as a product of standard subcube operators $\tildeZ{k'}_{\standard{K'}}$, where $k'\leq k$. We will prove this fact in \cref{sec: basis for k-th level}. 

Thus, we can describe the logical circuits for arbitrary subcube operators by describing the logic of standard subcube operators. To give some intuition, in \cref{sec: phasetilde operator logic} we will focus solely on the case of $k=1$, where $\tildeZ{1}_{\standard{K}}=\widetilde{\phasegate}_{\standard{K}}$. These operators will implement logical controlled-$Z$ circuits. Our main result on the logic implemented by standard signed subcube operators is proven in \cref{sec: standard subcube logic}. Such operators will act on $QRM_m(q,r)$ as circuits of multi-controlled-$Z$ circuits.

\subsection{A basis for \texorpdfstring{$k$}{k}-th level subcube logic}\label{sec: basis for k-th level}

To build intuition we will begin with the simple case of $\tildeZ{1}_A=Z_A$ logical operators. Recall the following lemma on classical Reed--Muller codes from \cref{sec:classical-RM-codes}.

\indicatordecomposition*

The multiplicative group of $Z$ operators, $\langle Z_A=Z(\indicator{A})\rangle$, is isomorphic to the additive group of indicator functions, $\langle\indicator{A}\rangle$, and so \cref{thm: subcube dimension implies logic} combined with \cref{lem: decomposition of subcube indicators} yields the following decomposition:
\begin{lemma}\label{lem: Z logical decomposition}
    Let $A\subcubeeq\ZZ_2^m$ be a subcube of type $J\subset S$, and suppose $Z_A\in\logs{0}$ for $QRM_m(q,r)$. Then
    \begin{equation}
        Z_A\equiv \prod_{I\subseteq I_A\colon \abs{I}+\abs{J}\leq r} Z_{\standard{I\cup J}}.
    \end{equation}
\end{lemma}
\begin{proof}
    Let $\ket\psi$ be a codestate. By \cref{lem: decomposition of subcube indicators}, $Z_A=\prod_{I\subseteq I_A}Z_{\standard{I\cup J}}$. By \cref{thm: subcube dimension implies logic}, for each $I\subseteq I_A$ such that $\abs{I\cup J}=\abs{I}+\abs{J}>r$ the operator $Z_{\standard{I\cup J}}$ is a Pauli stabilizer for $QRM_m(q,r)$, and so
    \begin{align}
        Z_A\ket\psi &= \prod_{I\subseteq I_A}Z_{\standard{I\cup J}}\ket\psi, \\
        &=  \prod_{I\subseteq I_A\colon \abs{I}+\abs{J}\leq r} Z_{\standard{I\cup J}}\ket\psi,
    \end{align}
    as desired.
\end{proof}
The intuition behind \cref{lem: Z logical decomposition} is not new: the set of standard subcubes with dimension at least $q$ and less than $r$ is precisely a basis for the space of logical $Z$ operators, so each $Z_A\in\logs{0}$ must have a decomposition in terms of standard subcube operators. The content of \cref{lem: Z logical decomposition} is to give the explicit decomposition of a $Z_A$ operator in terms of the basis logicals.

Recall now the signed indicator function, $\tildeindicator{A}:\ZZ_2^m\rightarrow \ZZ$, that maps even-weight elements of $A$ to $1$ and odd-weight elements of $A$ to $-1$. In the same way that the space of $\FF$-valued indicator functions on subcubes is isomorphic to the space of $Z$ operators on subcubes, the space of signed indicator functions modulo $2^{k+1}$ corresponds to the space of $\tildeZ{k}_A$ operators:

\begin{restatable}{lemma}{Zkisomorphism}\label{lem: Zk isomorphism}
    Given $k\in\NN$, the multiplicative group generated by the $\tildeZ{k}_A$ operators is isomorphic to the additive group generated by the $\tildeindicator{A}$ functions taken modulo ${2^{k+1}}$:
    \begin{equation}
        \Big\langle \tildeZ{k}_A \Big\rangle \cong \Big\langle \tildeindicator{A}\pmod{2^{k+1}}\Big\rangle.
    \end{equation}
\end{restatable}
\begin{proof}
    The map sending $\tildeZ{k}_A\mapsto \tildeindicator{A}\pmod{2^{k+1}}$ is a (surjective) homomorphism since the operators all commute and $\abs{\tildeZ{k}_A}= \abs{\tildeindicator{A}\pmod{2^{k+1}}}= 2^{k+1}$, and it is clearly injective as only $\tildeZ{k}_\emptyset=\eye$ maps to $\tildeindicator{\emptyset}=0^{2^m}$, the all zero string.
\end{proof}

We recall the decomposition of $\tildeindicator{A}$ given in \cref{sec:classical-RM-codes}:

\tildeindicatordecomposition*
Using \cref{lem: tilde indicator expression} and the isomorphism in \cref{lem: Zk isomorphism}, we can now prove the following generalization of \cref{lem: Z logical decomposition} to the case of arbitrary values of $k\in\ZZ_{\geq 0}$:
\begin{restatable}{theorem}{Zkdecomposition}\label{thm: Zk logical decomposition}
    Let $k\in\ZZ_{\geq 0}$, and consider $QRM_m(q,r)$. The standard subcube operators 
    \begin{equation}
        \br{\tildeZ{k}_{\standard{K}}\Bigmid K\in\mcQ_k}
    \end{equation}
    form a basis for the space of logical $\tildeZ{k}_A$ operators on $QRM_m(q,r)$. In particular, let $A\coloneqq x+\standard{K}\subcubeeq\ZZ_2^m$ be a subcube and let $x$ have minimal-weight in $A$. Recalling that $I_A\coloneqq \supp(x)$, we have
    \begin{equation}
        \tildeZ{k}_A\equiv \prod_{I\subseteq I_A\colon \abs{I}+\abs{K}\leq (k+1)r} \tildeZ{k}_{\standard{I\cup K}},
    \end{equation}
    up to Clifford stabilizers.
\end{restatable}
\begin{proof}
    Using \cref{lem: tilde indicator expression}, we have
    \begin{equation}\label{eq: tilde indicator equality modulo}
        \tildeindicator{A}\pmod{2^{k+1}} \equiv \sum_{I\subseteq I_A} \tildeindicator{\standard{I\cup J}} - \sum_{{ i=1}}^{2^k}\sum_{I\subseteq I_A\colon \abs{I}=i} 2^i\cdot\tildeindicator{e_{I_A\setminus I}+\standard{I\cup J}}.
    \end{equation}
    Recalling that $\Z{k}^{2^\ell}=\Z{k-\ell}$, applying the isomorphism from \cref{lem: Zk isomorphism} to \cref{eq: tilde indicator equality modulo} yields
    \begin{align}
        \tildeZ{k}_A &= \left(\prod_{I\subseteq I_A}\tildeZ{k}_{\standard{I\cup J}}\right)\cdot \left(
        \prod_{i=1}^{2^k}\bigg( \prod_{I\subseteq I_A\colon \abs{I}=i} \tildeZ{k-i}_{e_{I_A\setminus I}+\standard{I\cup J}}
        \bigg)
        \right)^\dagger,\\
        &= \left(\prod_{I\subseteq I_A}\tildeZ{k}_{\standard{I\cup J}}\right)\cdot \left(
        \prod_{i=1}^{2^k} U_i^\dagger,
        \right)
    \end{align}
    where we've defined
    \begin{equation}
        U_i\coloneqq \prod_{I\subseteq I_A\colon \abs{I}=i} \tildeZ{k-i}_{e_{I_A\setminus I}+\standard{I\cup J}}.
    \end{equation}
    Each $U_i$ is the product of $\tildeZ{k-i}$ operators acting on subcubes of dimension $\abs{I}+\abs{J}\geq i+q+kr+1$, since by \cref{thm: subcube dimension implies logic} $\abs{J}\geq q+kr+1$. Clearly $i+q+kr+1\geq (k-i + 1)r+1$, so \cref{thm: subcube dimension implies logic} implies that, in fact, each $U_i$ is the product of $(k-i)$-level Clifford stabilizers. As the $U_i$ operators are unitary their adjoints are also logical identity on the code space, so we have
    \begin{equation}
        \tildeZ{k}_A \equiv \prod_{I\subseteq I_A}\tildeZ{k}_{\standard{I\cup J}}.
    \end{equation}
    The desired result holds as by \cref{thm: subcube dimension implies logic}, for each $I$ such that $\abs{I\cup J}=\abs{I} + \abs{J}>(k+1)r$, $\tildeZ{k}_{\standard{I\cup J}}$ is also logical identity on the code space.
\end{proof}

\subsection{Standard subcube logic --- the phase operator case}\label{sec: phasetilde operator logic}
We first look at the case when $k=1$, so that $\tildeZ{1}_{\standard{K}}=\tildephase_{\standard{K}}$. From \cref{thm: subcube dimension implies logic}, $\tildephase_{\standard{K}}$ is a non-trivial logical Clifford operator if and only if $q+r+1\leq \abs{K}<2r+1$. For such a $K\subseteq S$, we will determine the logical circuits performed on $QRM_m(q,r)$ via the physical implementation of $\tildephase_{\standard{K}}$.

Consider the $J$-th logical qubit of $QRM_m(q,r)$ determined by the $\overline{Z}_J$ and $\overline{X}_J$ operators. In \cref{sec: standard subcube logic} we will prove \cref{clm: standard subcube operator conjugation rule}, which describes how the $\tildeZ{k}_{\standard{K}}$ operators conjugate logical $X$ operators. For now, we will simply state the implication of \cref{clm: standard subcube operator conjugation rule} for $\tildephase_{\standard{K}}$:
\begin{enumerate}
    \item If $J\not\subseteq K$ then $\tildephase_{\standard{K}}$ acts trivially on the $J$-th qubit.
    \item If $J\subset K$ then $\tildephase_{\standard{K}}$ commutes with $\overline{Z}_J$, and conjugates $\overline{X}_J$ as
    \begin{equation}\label{eq: phase tilde conjugation rule}
        \tildephase_{\standard{K}} \overline{X}_J \tildephase_{\standard{K}}^\dagger = \overline{X}_J \left(\prod_{I\subseteq J\colon \abs{I}+\abs{K}-\abs{J}\leq r} Z_{\standard{I\cup (K\setminus J)}}\right).
    \end{equation}
\end{enumerate}
Let $\mcF(K)_{\sim J}$ denote the collection of sets appearing in the product on the right-hand side of \cref{eq: phase tilde conjugation rule},
\begin{equation}\label{eq: local sets 1}
    \mcF(K)_{\sim J}\coloneqq\br{ J'\in\mcQ \Bigmid J' = I\cup (K\setminus J),\; I\subseteq J},
\end{equation}
so that
\begin{equation}\label{eq: phase tilde conjugation rule simple}
    \tildephase_{\standard{K}} \overline{X}_J \tildephase_{\standard{K}}^\dagger = \overline{X}_J \left(\prod_{J'\in\mcF(K)_{\sim J}} Z_{\standard{J'}}\right).
\end{equation}
We note a few things about the terms in the product on the right-hand side of \cref{eq: phase tilde conjugation rule simple}:
\begin{enumerate}
    \item We are guaranteed that $q+1\leq\abs{J'}\leq r$ for each $J'\in\mcF(K)_{\sim J}$, so $Z_{\standard{J'}}$ is the logical $Z$ operator $\overline{Z}_{J'}$.
    \item The $J$-th logical $Z$ operator cannot appear in the product. This would require setting $I=J$, but $\abs{J}+\abs{K}-\abs{J}\geq q+r+1> r$ is too large to be included in the product. 
    \item A set $J_2\in\mcF(K)_{\sim J_1}$ if and only if $J_1\in\mcF(K)_{\sim J_2}$: Take $I\subseteq J_1$ such that $J_2=I\cup (K\setminus J_1)$. Then $K\setminus J_2 = J_1\setminus I$, and so using the same $I$ we have that $I\cup (K\setminus J_2)= J_1$, and $\abs{I}+\abs{K}-\abs{J_2}=\abs{J_1}\leq r$, implying that $J_1\in\mcF(K)_{\sim J_2}$.
\end{enumerate}
These conditions are precisely the defining features of a circuit consisting of $CZ$ operators, and, indeed, $\tildephase_{\standard{K}}$ implements a logical $\overline{C^{\mcF(K)}Z}$ where $\mcF(K)$ is a collection of pairs of logical qubits acted on by $\overline{CZ}$ gates.

\begin{definition}\label{def: CZ sets}
    Given $K\subseteq S$ satisfying $q+r+1\leq\abs{K}\leq 2r$, define a collection of pairs of logical qubit indices, $\mcF(K)$, via
    \begin{equation*}
        \mcF(K) \coloneqq \br{\br{J_1,J_2}\subseteq \mcQ \Bigmid J_1\cup J_2=K },
    \end{equation*}
    where we recall that $\mcQ\coloneqq\br{J\mid q+1\leq\abs{J}\leq r}$.
\end{definition}

Recall from \cref{sec:CZ circuits} that given a collection of subsets of some index set $\mcI$, $\mcF\subseteq\powerset{\mcI}$, the set $\mcF_{\sim i}$ for $i\in \mcI$ is defined as
\begin{equation}\label{eq: local sets 2}
    \mcF_{\sim i} \coloneqq \br{I\setminus \br{i}\Bigmid I\in\mcF, i\in I}.
\end{equation}
That is, $\mcF_{\sim i}$ is the collection of all sets in $\mcF$ that contain $i$, but that have $i$ removed.

\begin{lemma}\label{lem: consistent definitions of sets}
    Given $K\subseteq S$ with $q+r+1\leq\abs{K}\leq 2r$ and $J\in\mcQ$, the definitions of $\mcF(K)_{\sim J}$ given by \cref{eq: local sets 1} and \cref{eq: local sets 2} are consistent.
\end{lemma}
\begin{proof}
    ($\subseteq$) Suppose $J'\in \mcF(K)_{\sim J}$ as defined in \cref{eq: local sets 1}, i.e., $J'\in \mcQ$ and there is an $I\subseteq J$ such that $J'=I\cup (K\setminus J)$. Clearly, we have that $J\cup J' = J\cup (K\setminus J) = K$.

    ($\supseteq$) Suppose $J'\in \mcF(K)_{\sim J}$ as defined using \cref{eq: local sets 2} applied to the definition of $\mcF(K)$, i.e., $J\in\mcQ$ and $J\cup J'=K$. Defining $I\coloneqq J\cap J'$, we see that $J'= I\cup( (J\cup J')\setminus J)= I\cup (K\setminus J)$.
\end{proof}

We are now prepared to give the logical controlled-$Z$ circuit implemented by $\tildephase_{\standard{K}}$:

\begin{proposition}\label{prop: logical CZ circuit}
    If $K\subseteq S$ satisfies $q+r+1\leq \abs{K}\leq 2r$, then $\tildephase_{\standard{K}}$ implements the logical $CZ$ circuit corresponding to the pairs of qubits in $\mcF(K)$: 
    \begin{equation}
        \tildephase_{\standard{K}}\equiv\overline{C^{\mcF(K)}Z}=\prod_{\br{J_1,J_2}\subseteq\mcQ \colon J_1\cup J_2=K} \overline{CZ}_{J_1,J_2}.
    \end{equation}
\end{proposition}
\begin{proof}
    Clearly both $\tildephase_{\standard{K}}$ and $\overline{C^{\mcF(K)}Z}$ commute with the logical $Z$ operators, so we only must verify that they conjugate the logical $X$ operators in the same way. We have already shown via \cref{eq: phase tilde conjugation rule simple} and \cref{lem: consistent definitions of sets} that
    \begin{equation*}
        \tildephase_{\standard{K}} \overline{X}_J \tildephase_{\standard{K}}^\dagger = \overline{X}_J \left(\prod_{J'\in\mcF(K)_{\sim J}} \overline{Z}_{{J'}}\right),
    \end{equation*}
    which is precisely the conjugation rule for $\overline{C^{\mcF(K)}Z}$ as proven in \cref{lem: controlled-Z circuit action}.
\end{proof}

\subsection{Standard subcube logic}\label{sec: standard subcube logic}
By \cref{thm: Zk logical decomposition}, in order to fully characterize the logic that $\tildeZ{k}_A$ operators perform on $QRM_m(q,r)$, we need only understand the logic performed by standard subcube operators, $\tildeZ{k}_{\standard{K}}$. Recall now the definition of the logical space for the quantum Reed--Muller codes:
\QRMdef*
In particular, the logical $Z$ and $X$ operators are given by $\br{\overline{Z}_J\mid J\in \mcQ}$ and $\br{\overline{X}_J\mid J\in \mcQ}$, respectively. We further recall the collections $\mcQ_k\subseteq\powerset{S}$:
\kgatesets*

\cref{thm: Zk logical decomposition} justifies the ``index set'' terminology, as operators from $\br{\tildeZ{k}_{\standard{K}} \mid K\in\mcQ_k}$ can be used to construct any signed subcube operator $\tildeZ{k}_A\in\logs{k}$. Note that in the case $k=0$ this corresponds precisely to the fact that $\mcQ_0=\mcQ$ indexes the logical $Z$ operators of $QRM_m(q,r)$.

The logic implemented by $\tildeZ{k}_{\standard{K}}$ for $K\in\mcQ_k$ will ultimately be related to sets of logical qubits that form so-called ``minimal covers'' of $K$:
\minimalcovers*

Recall from \cref{sec:CZ circuits} the definition of $\mcF(K)_{\sim J}$:
\begin{equation}\label{eq: local sets inductive proof}
    \mcF(K)_{\sim J} = \br{\mcJ\setminus\br{J}\Bigmid J\in\mcJ\subseteq \mcQ ,\; \mcJ\in\mcF(K)}.
\end{equation}
Each $\mcJ'\in\mcF(K)_{\sim J}$ is called a \emph{partial minimal cover for $K$ relative to $J$}; they are precisely the collections of logical qubit indices for which $\mcJ'\cup\br{J}$ is a minimal cover for $K$.

In order to prove the main result of this section, \cref{thm: logical multi-controlled-Z circuit}, we will need to define a certain subset of generators that depend on a particular $K\in\mcQ_k$ as well as a particular logical qubit index $J\in\mcQ$:

\begin{definition}
    Let $J\in\mcQ$ be an index for a logical qubit, $K\in\mcQ_k$ an index of a $k$-th level logical operator, and $K'\subseteq K$ an arbitrary subset of $K$. The set $K'$ is said to be \emph{$\mcQ$-dense in $K$ relative to $J$} if (1) $K'$ is an index of a $(k-1)$-st level logical operator, $K'\in\mcQ_{k-1}$, and (2) if the union of $K'$ and $J$ is all of $K$, $K=K'\cup J$. The set $J\in\mcQ$ is often a fixed logical qubit index, so we will simply say that $K'$ \emph{is dense in} $K$ if it satisfies the mentioned conditions for the fixed choice of $J$.

    The collection of all subsets $K'\subset K$ that are dense in $K$ will be denoted by $\mcD_J(K)$:
    \begin{equation}
        \mcD_J(K) \coloneqq \br{ K'\subseteq K \Bigmid K'\in\mcQ_{k-1},\; K = K'\cup J}.
    \end{equation}
\end{definition}
Dense subsets of $K$ will typically appear in a different, yet equivalent, form:
\begin{lemma}\label{lem: alternate definition of dense}
    The collection $\mcD_J(K)$ can alternatively be defined as:
    \begin{equation}
        \mcD_J(K)\coloneqq\br{ K'\subset K \Bigmid K'\in\mcQ_{k-1},\; K' = I\cup (K\setminus J) \text{ for some } I\subseteq J}.
    \end{equation}
\end{lemma}
\begin{proof}\hspace{0em}

    ($\subseteq$) Suppose $K'$ satisfies the first definition. For $I\coloneqq K'\cap J\subseteq J$, $K'=I\cup(K\setminus J)$.

    ($\supseteq$) Suppose $K'$ satisfies the second definition. Clearly $K'\cup J = I\cup(K\setminus J)\cup J=K$.
\end{proof}

We note that, as a set of logical qubits, any minimal cover for $K\in\mcQ_k$, $\mcJ$, can be used to define a logical multi-controlled-$Z$ gate acting on $k+1$ logical qubits of $QRM_m(q,r)$:
\begin{equation*}
    \overline{C^\mcJ Z} \coloneqq \overline{C^{(k+1)}Z}_{\mcJ}\otimes \overline{\eye}_{\mcQ\setminus\mcJ}.
\end{equation*}

Similarly, a partial minimal cover for $K$ relative to $J$, $\mcJ'\subseteq\mcQ$, can be used to define a multi-controlled-$Z$ gate acting on $k$ qubits, which is guaranteed to act as identity on the $J$-th qubit:
\begin{equation*}
    \overline{C^{\mcJ'} Z} \coloneqq \overline{C^{(k)}Z}_{\mcJ'}\otimes \overline{\eye}_{\mcQ\setminus\br{\mcJ'\cup\br{J}}} \otimes\overline{\eye}_{\br{J}}.
\end{equation*}

Therefore, we can define logical multi-controlled-$Z$ circuits using the collections of minimal covers and partial minimal covers, $\prod_{\mcJ\in\mcF(K)}\overline{C^{\mcJ}Z}$ and $\prod_{\mcJ'\in\mcF(K)_J}\overline{C^{\mcJ'}Z}$, the latter of which acts as identity on the $J$-th qubit. The remainder of the section will be dedicated to proving our main result on the logic implemented by standard signed subcube operators:

\begin{restatable}{theorem}{logicalcircuitsigned}\label{thm: logical multi-controlled-Z circuit}
    For every $K\in\mcQ_k$, $\tildeZ{k}_{\standard{K}}$ implements the logical multi-controlled-$Z$ circuit corresponding to the collection of minimal covers of $K$:
    \begin{equation}
        \tildeZ{k}_{\standard{K}} \equiv \overline{C^{\mcF(K)}Z}.
    \end{equation}
\end{restatable}
The proof will follow from two claims. The first claim is given in \textbf{Step \hyperlink{Step I}{I}}, and concerns the composition of multi-controlled-$Z$ circuits defined using the three collections we have defined--- minimal covers, partial minimal covers, and dense sets. The second claim, given in \textbf{Step \hyperlink{Step I}{II}}, details how a given $\tildeZ{k}_{\standard{K}}$ operator conjugates the logical $X$ operators of $QRM_m(q,r)$. We restate and prove \cref{thm: logical multi-controlled-Z circuit} in \textbf{Step \hyperlink{Step III}{III}}, using these two ingredients.

\hypertarget{Step I}{\paragraph{Step I: Composition of multi-controlled-$Z$ circuits.}}
For a fixed logical qubit index, $J\in\mcQ$, and a fixed $k$-th level operator index, $K\in\mcQ_k$, our first goal is to relate the collection of partial minimal covers of $K$ to the 
collection of dense subsets of $K$. The culmination of this effort will be the following result on the composition of logical multi-controlled-$Z$ circuits defined using these collections:

\begin{claim}\label{clm: multi-controlled-Z circuit composition}
    Let $J\in\mcQ$ and $K\in\mcQ_k$ with $J\subseteq K$. Then
    \begin{equation}\label{eq: logical circuit composition}
        \overline{C^{\mcF(K)_{\sim J}}Z} = \prod_{K'\in\mcD_J(K)} \overline{C^{\mcF(K')}Z}.
    \end{equation}
\end{claim}
That is, given a dense set $K'\in\mcD_J(K)$, we can consider the collection of sets of logical qubits that it defines via \emph{its own} collection of minimal covers, $\mcF(K')$. \cref{clm: multi-controlled-Z circuit composition} says that the composition of the multi-controlled-$Z$ circuits defined by the minimal covers of {all} dense subsets of $K$ is precisely the multi-controlled-$Z$ circuit defined by the collection of partial minimal covers of $K$.

\cref{clm: multi-controlled-Z circuit composition} follows from two lemmas. First, it is a simple fact that no set of logical qubits can be a minimal cover for two different $K_1,K_2\in\mcQ_k$:

\begin{lemma}\label{lem: minimal covers are disjoint}
    Let $K_1,K_2\subseteq S$ be disjoint subsets of generators, $K_1\neq K_2$. The intersection of the collections of their minimal covers is empty, $\mcF(K_1)\cap \mcF(K_2)=\emptyset$.
\end{lemma}
\begin{proof}
    If $\mcJ\in\mcF(K_1)\cap \mcF(K_2)$ then it covers both $K_1$ and $K_2$, and so $K_1=\bigcup_{J\in\mcJ}J=K_2$.
\end{proof}

Next, we show that the collection of partial minimal covers of $K$ is given by the union of all minimal covers of dense subsets of $K$.
\begin{lemma}\label{lem: union of minimal covers}
    Let $J\in\mcQ$ and $K\in\mcQ_k$ with $J\subseteq K$. Then $\bigcup_{K'\in\mcD_J(K)} \mcF(K') = \mcF(K)_{\sim J}$.
\end{lemma}
\begin{proof}
    ($\subseteq$) Take $\mcJ'\in\mcF(K')$ for some $K'\in\mcD_J(K)$. By definition, $\mcJ'\in\mcF(K)_{\sim J}$ if and only if $\mcJ'\cup\br{J}$ is a minimal cover for $K$. We verify the two conditions for $\mcJ'$ to be a minimal cover: 
    \begin{enumerate}[leftmargin=*]
        \item $\mcJ'\cup\br{J}$ is a cover for $K$: We have assumed that $\mcJ'$ is a cover for $K'$ and that $K'$ dense in $K$ relative to $J$. Thus, by definition, $\left(\bigcup_{J'\in\mcJ'}J'\right)\cup J=K'\cup J = K$.
        \item $\abs{\mcJ'\cup\br{J}} =k+1$: First, note that $J\notin\mcJ$; otherwise $K'=K$, which cannot happen as $K'\in\mcQ_{k-1}$, $K\in\mcQ_k$, and $\mcQ_{k-1}\cap\mcQ_k=\emptyset$. So, $\abs{\mcJ'\cup\br{J}} =\abs{\mcJ'}+1$. Now, since $\mcJ'$ is a minimal cover for $K'\in\mcQ_{k-1}$ we know that $\abs{\mcJ'}=(k-1)+1 =k$, so $\abs{\mcJ'\cup\br{J}} =k+1$.
    \end{enumerate}

    ($\supseteq$) Take $\mcJ'\in \mcF(K)_{\sim J}$ and define $K'\coloneqq \bigcup_{J'\in\mcJ'} J'$. We will show (1) $\mcJ'$ is a minimal cover for $K'$ and (2) $K'$ is an almost-covering set for $K$, which together imply the desired result. We will first show that $K'\in\mcQ_{k-1}$.
    
    Note that as $\mcJ'\in \mcF(K)_{\sim J}$, by assumption we know that $J\cup K'=K$ and that $\abs{\mcJ'}=k$. 
    As any set in $\mcQ$ has at most $r$ elements, we can upper bound
    \begin{align*}
        \abs{K'} &\leq \sum_{J'\in\mcJ'}\abs{J'},\\
        &\leq kr,\\
        &<\left((k-1)+1\right)r+1.
    \end{align*}
    As $K' \supseteq \left(J\cup K' \right)\setminus J=K\setminus J$, so we can lower bound
    \begin{align*}
        \abs{K'} &\geq \abs{K}-\abs{J},\\
        \since{$K\in\mcQ_k$} &\geq q+kr+1-\abs{J},\\
        \since{$J\in\mcQ$} &\geq q+\big(k-1\big)r+1.
    \end{align*}
    The given bounds on $\abs{K'}$ are precisely the conditions for $K'\in\mcQ_{k-1}$.

    Now, $\mcJ'$ is, by definition, a cover of $K'$, and since $\abs{\mcJ'}=k=(k-1)+1$ we have that $\mcJ'$ is a minimal cover for $K'$. Further, since $K'\cup J=K$ and $K'\in\mcQ_{k-1}$, $K'$ satisfies the condition for an almost-covering of $K$. Thus $\mcJ'\in\mcF(K')$ for some $K'\in\mcD_J(K)$.
\end{proof}

We now prove the desired composition result:

\begin{proof}[Proof of \cref{clm: multi-controlled-Z circuit composition}]
    Consider $\prod_{K'\in\mcD_J(K)} \overline{C^{\mcF(K')}Z}$. By \cref{lem: minimal covers are disjoint} we are guaranteed that no $\overline{C^{\mcJ_1} Z}$ for $\mcJ_1\in\mcF(K_1)$ can cancel with a $\overline{C^{\mcJ_2} Z}$ for $\mcJ_2\in\mcF(K_2)$, $K_2\neq K_1$ for $K_1,K_2\in\mcD_J(K)$. Thus
    \begin{align*}
        \prod_{K'\in\mcD_J(K)} \overline{C^{\mcF(K')}Z} &= \prod_{K'\in\mcD_J(K)} \prod_{\mcJ'\in\mcF(K')}\overline{C^{\mcJ'}Z},\\
        \since{\cref{lem: minimal covers are disjoint}}&= \prod_{\mcJ'\in \bigcup_{K'\in\mcD_J(K)} \mcF(K')} \overline{C^{\mcJ'}Z},\\
        \since{\cref{lem: union of minimal covers}} &= \prod_{\mcJ'\in\mcF(K)_{\sim J}}\overline{C^{\mcJ'}Z}.
    \end{align*}
\end{proof}

\hypertarget{Step II}{\paragraph{Step II: Conjugating logical Pauli operators.}}
Our next goal is to understand how $\tildeZ{k}_{\standard{K}}$ conjugates the logical Pauli operators of $QRM_m(q,r)$.

\begin{claim}\label{clm: standard subcube operator conjugation rule}
    Let $J\subseteq K$ for $J\in\mcQ$, $K\in\mcQ_k$, and $k\geq 1$. Consider the operator $\tildeZ{k}_{\standard{K}}$ acting on $QRM_m(q,r)$.
    \begin{enumerate}
        \item $\tildeZ{k}_{\standard{K}}$ commutes with every logical $Z$ operator.
        \item Up to Clifford stabilizers, $\tildeZ{k}_{\standard{K}}$ conjugates a logical $\overline{X}_J$ operator as
        \begin{equation}
            \tildeZ{k}_{\standard{K}} \overline{X}_J \tildeZ{k}_{\standard{K}}^\dagger \equiv \begin{cases}
                \overline{X}_J, &\text{if } J\not\subseteq K,\\
                \overline{X}_J \left(\prod_{K'\in\mcD_J(K)} \tildeZ{k-1}_{\standard{K'}}\right),  &\text{otherwise}.
            \end{cases}
        \end{equation}
    \end{enumerate} 
\end{claim}
\begin{proof}
    The first assertion is trivial, so we only prove the second.
    
    As $e_J$ is the minimal Hamming weight element of $e_J+\standard{S\setminus J}$, clearly $\standard{K}\cap (e_J+\standard{S\setminus J})\neq\emptyset$ if and only if $e_J\in\standard{K}$. Thus, $J\not\subseteq K$ is precisely the case that $\tildeZ{k}_{\standard{K}}$ and $\overline{X}_J$ have disjoint support and therefore commute. For the rest of the proof suppose that $J\subseteq K$, so 
    \begin{align}
        \standard{K}\cap (e_J+\standard{S\setminus J})&=e_J+\standard{K\cap (S\setminus J)},\\
        &=e_J+\standard{K\setminus J}.
    \end{align}
    We now apply \cref{thm: subcube dimension implies logic} to both $K$ and $J$ to bound \begin{align}
        \dim(e_J+\standard{K\setminus J})&=\abs{K}-\abs{J},\\
        &\geq q+kr+1 -\abs{J},\\
        &\geq q+kr+1-r,\\
        &= q + (k-1)r+1,
    \end{align}
    We have assumed $k\geq 1$, thus $\dim(e_J+\standard{K\setminus J})\geq 1$ and $\tildeZ{k}_{\standard{K}} \overline{X}_J \tildeZ{k}_{\standard{K}}^\dagger$ is phase-free by \cref{lem: signed operator phase free condition}. By definition, this means that
    \begin{equation}
        \tildeZ{k}_{\standard{K}} \overline{X}_J \tildeZ{k}_{\standard{K}}^\dagger= \left( \tildeZ{k-1}_{e_{J}+\standard{K\setminus J}} \right)\overline{X}_J.
    \end{equation}
    Further, the lower bound $\dim(e_J+\standard{K\setminus J})\geq q+(k-1)r+1$ implies via \cref{thm: subcube dimension implies logic} that $\tildeZ{k-1}_{e_{J}+\standard{K\setminus J}}\in\errors{k-1}$, and so by \cref{cor: tilde operators are logically Hermitian} the operator $\tildeZ{k-1}_{e_{J}+\standard{K\setminus J}}$ is logically equivalent to its Hermitian conjugate. Thus, we have
    \begin{align}
        \tildeZ{k}_{\standard{K}} \overline{X}_J \tildeZ{k}_{\standard{K}}^\dagger&=\overline{X}_J\left( \tildeZ{k-1}_{e_{J}+\standard{K\setminus J}} \right)^\dagger,\\
        \since{\cref{cor: tilde operators are logically Hermitian}}&\equiv \overline{X}_J\left( \tildeZ{k-1}_{e_{J}+\standard{K\setminus J}} \right),\\
        \since{\cref{thm: Zk logical decomposition}}
        &\equiv \overline{X}_J \left(\prod_{I\subseteq J\colon \abs{I}+\abs{K}-\abs{J}\leq kr} \tildeZ{k-1}_{\standard{I\cup (K\setminus J)}}\right). \label{eq: conjugation final step}
    \end{align}
    In the last step we utilized \cref{thm: Zk logical decomposition}, which gives the decomposition of an arbitrary subcube operator as a product of standard subcube operators. The claim now holds by \cref{lem: alternate definition of dense}, which says that the sets appearing in the product in \cref{eq: conjugation final step} correspond precisely to the collection of dense subsets of $K$ relative to $J$, $\mcD_J(K)$.
\end{proof}

\hypertarget{Step III}{\paragraph{Step III: Implemented logic.}}
We are now prepared to describe the logic performed by standard subcube operators using \cref{clm: multi-controlled-Z circuit composition} and \cref{clm: standard subcube operator conjugation rule}.

\logicalcircuitsigned*
\begin{proof}
    Induction on $k$. If $k=0$ then $\mcF(K)=\br{\br{K}}$ and $\tildeZ{0}_{\standard{K}}=\overline{Z}_K = \overline{C^{\br{\br{K}}}Z}$, so the result holds in the base case. Suppose now that the result holds for all $K'\in\mcQ_{k}$, and choose $K\in\mcQ_{k+1}$. We seek to show that $\tildeZ{k}_{\standard{K}}$ and $\overline{C^{\mcF(K)}Z}$ conjugate logical Pauli operators in the same way.
    
    Consider $J\in\mcQ$. By \cref{clm: standard subcube operator conjugation rule}, $\tildeZ{k}_{\standard{K}}$ commutes with $\overline{Z}_J$ and maps
    \begin{equation}
        \tildeZ{k}_{\standard{K}} \overline{X}_J \tildeZ{k}_{\standard{K}}^\dagger \equiv \begin{cases}
            \overline{X}_J, &\text{if } J\not\subseteq K,\\
            \overline{X}_J \left(\prod_{K'\in\mcD_J(K)} \tildeZ{k-1}_{\standard{K'}}\right),  &\text{otherwise}.
        \end{cases}
    \end{equation}

    By \cref{lem: controlled-Z circuit action}, $\overline{C^{\mcF(K)}Z}$ commutes with $\overline{Z}_J$ and maps
    \begin{equation*}
        \left(\overline{C^{\mcF(K)}Z} \right)\overline{X}_J \left(\overline{C^{\mcF(K)}Z}\right) = \overline{X}_J \left(\overline{C^{\mcF(K)_{\sim J}}Z}\right).
    \end{equation*}
    We proceed in cases:
    
    \noindent\textbf{I.} ($J\not\subseteq K$) $\mcF(K)_{\sim J}=\emptyset$ in this case, as any union $J\cup\left(\bigcup_{J'\in\mcJ'} J'\right)$ for $\mcJ'\subseteq\mcQ$ is guaranteed to contain an element outside of $K$. Thus $\tildeZ{k}_{\standard{K}}$ and $\overline{C^{\mcF(K)}Z}$ each commute with both $\overline{Z}_J$ and $\overline{X}_J$.

    \noindent\textbf{II.} ($J\subseteq K$) In this case,
    \begin{align}
        \tildeZ{k}_{\standard{K}} \overline{X}_J \tildeZ{k}_{\standard{K}}^\dagger &\equiv \overline{X}_J \left(\prod_{K'\in\mcD_J(K)} \tildeZ{k-1}_{\standard{K'}}\right). \label{eq: pre-induction standard cube}\\
        \intertext{By definition of $\mcD_J(K)$, each subcube, $\standard{K'}$, appearing in the product on the right-hand side of \cref{eq: pre-induction standard cube} satisfies $K'\in\mcQ_{k-1}$, so we can use the induction hypothesis to compute}
        &\equiv \overline{X}_J \left(\prod_{K'\in\mcD_J(K)} \overline{C^{\mcF(K')}Z}\right),\\
        \intertext{which by \cref{clm: multi-controlled-Z circuit composition} is precisely}
        &\equiv \overline{X}_J \overline{C^{\mcF(K)_{\sim J}}Z}.
    \end{align}

    As $\overline{C^{\mcF(K)}Z}$ and $\tildeZ{k}_{\standard{K}}$ conjugate the logical Pauli operators of $QRM_m(q,r)$ in the same way, by definition they are equivalent logical operators for the code.
\end{proof}

\subsection{Arbitrary subcubes}\label{sec: arbitrary subcube signed logic}

Consider an arbitrary subcube $A\coloneqq x+\standard{K}\subcubeeq\ZZ_2^m$, where $K\subseteq S$ is the type of $A$. \cref{thm: subcube dimension implies logic} tells us that the operator $\tildeZ{k}_A$ will be a logical operator for $QRM_m(q,r)$ if and only if $K\in\mcQ_k$. Given the decomposition of $\tildeZ{k}_A$ into standard subcube operators (\cref{thm: Zk logical decomposition}) as well as the description of logical multi-controlled-$Z$ circuits implemented by these operators (\cref{thm: logical multi-controlled-Z circuit}), it may be natural to wonder if interesting logical circuits can be constructed via $\tildeZ{k}_A$. As as example, it is straightforward to verify from \cref{thm: logical multi-controlled-Z circuit} and the definition of minimal covers that every $\tildeZ{k}_{\standard{K}}$ implements a circuit containing more than one logical gate.\footnote{Other than the case of $QRM_m(0,1)$, which corresponds to the well-known family of hypercube codes .} Can $\tildeZ{k}_A$ implement single multi-controlled-$Z$ gates when $A$ is no longer a standard subcube? Or perhaps more generally, can some product of $\tildeZ{k}_{\standard{K}}$ operators implement a single multi-controlled-$Z$ gate?

Unfortunately, the standard subcube operators are, in some sense, the fundamental logical multi-controlled-$Z$ circuits that can be implemented on $QRM_m(q,r)$. By this, we mean that the logical circuit defined by a product, $\tildeZ{k_1}_{\standard{K_1}}\cdot \tildeZ{k_2}_{\standard{K_2}}$, can never have cancellations of logical gates. More formally:
\begin{theorem}\label{thm: no cancellations in the signed case}
    Let $\br{k_1,\dots,k_\ell}$ be a set of non-negative integers, and suppose $\br{K_i}_{i\in[\ell]}$ is a collection of (distinct) subsets $K_i\subseteq S$, such that $K_i\in\mcQ_{k_i}$ for every $i\in[\ell]$. Then
    \begin{equation}
        \prod_{i\in[\ell]}\tildeZ{k_i}_{\standard{K_i}} \equiv \overline{C^{\mcF}Z},
    \end{equation}
    where 
    \begin{equation}
        \mcF\coloneqq\br{\mcJ\subseteq \mcQ\Bigmid \mcJ \text{ is a minimal cover for some }K_i}
    \end{equation} 
    is simply the union of all collections of minimal covers of the $K_i$ sets.
\end{theorem}
\begin{proof}
    A direct application of \cref{thm: logical multi-controlled-Z circuit} implies that
    \begin{equation}\label{eq: basic product of subcube operators}
        \prod_{i\in[\ell]}\tildeZ{k_i}_{\standard{K_i}} \equiv \prod_{i\in[\ell]}\prod_{\mcJ\in\mcF(K_i)}\overline{C^{\mcJ}Z}.
    \end{equation}
    The result will hold by proving that any $\overline{C^{\mcJ}Z}$ in the right-hand side of \cref{eq: basic product of subcube operators} can only appear once. It is trivial that for a \emph{particular} $i\in[\ell]$, $\overline{C^{\mcJ}Z}$ can only appear once in the product $\prod_{\mcJ\in\mcF(K_i)}\overline{C^{\mcJ}Z}$. The operator $\overline{C^{\mcJ}Z}$ can also only appear for a \emph{single} $i\in[\ell]$, as otherwise the cover property of $\mcJ$ would imply that $K_i=\bigcup_{J\in\mcJ}J=K_j$, contradiction the fact the $K_i$ sets are all distinct.
\end{proof}

So, while the operator $\tildeZ{k}_{x+\standard{K}}$ is non-trivial whenever $K\in\mcQ_k$, \cref{thm: no cancellations in the signed case} implies that the circuit it defines necessarily contains more gates than $\tildeZ{k}_{\standard{K}}$.

\section{Unsigned subcube operator logic}\label{sec: unsigned logic}
We now arrive at the case of $\Z{k}$ subcube operators, i.e., transversal operators which act as $\Z{k}$ on vertices of a subcube and identity elsewhere, with no adjoints.
As in the case of signed subcube operators, 
\cref{thm: subcube dimension implies logic} gives necessary and sufficient conditions for when a $\Z{k}_A$ operator performs non-trivial logic on $QRM_m(q,r)$; the aim of this section is to determine the logical circuit implemented by an \emph{unsigned} subcube operator. As in \cref{sec: signed logic}, we can first use a decomposition lemma from \cref{sec:classical-RM-codes} to reduce to the case of standard subcube operators.

\begin{fact}[Unsigned operator version of \cref{thm: Zk logical decomposition}]\label{fact: transversal Zk subcube decomposition}
    Let $A\coloneqq x+\standard{K}\subcubeeq\ZZ_2^m$ be a subcube with $K\in\mcQ_k$, and let $x$ have minimal-weight in $A$. Recalling that $I_A\coloneqq \supp(x)$, we have
    \begin{equation}
        \Z{k}_A\equiv \prod_{I\subseteq I_A\colon \abs{I}+\abs{K}\leq (k+1)r} \Z{k}_{\standard{I\cup K}},
    \end{equation}
    up to Clifford stabilizers.
\end{fact}
We will not prove \cref{fact: transversal Zk subcube decomposition} as its proof is exactly the same as the proof \cref{thm: Zk logical decomposition}, with applications of \cref{lem: tilde indicator expression} replaced by \cref{lem: bar indicator expression}.

In \cref{sec: signed logic} we directly proved the logical circuits implemented by $\tildeZ{k}_{\standard{K}}$ operators by examining how they conjugated the logical Pauli operators of $QRM_m(q,r)$. In the present section, we will instead show that the circuits implemented by unsigned operators can be deduced by those implemented by signed operators. 

Analogously to \cref{lem: Zk isomorphism}, given $k\in\NN$, the multiplicative group generated by the $\Z{k}_A$ operators is isomorphic to the additive group generated by the $\barindicator{A}$ functions modulo ${2^{k+1}}$:
\begin{equation}
    \Big\langle \Z{k}_A \Big\rangle \cong \Big\langle \barindicator{A}\pmod{2^{k+1}}\Big\rangle.
\end{equation}

As unsigned indicator functions correspond to the space of $\Z{k}_A$ operators, and \emph{signed} indicator functions correspond to the space of $\tildeZ{k}_A$ operators, if an unsigned indicator function can be written as the sum of signed indicator functions then we can use the mentioned isomorphism to deduce that a $\Z{k}_{\standard{K}}$ operator can be decomposed into a product of $\tildeZ{k'}_{\standard{K'}}$ operators.
Indeed, we stated in \cref{sec:classical-RM-codes} the following:

\strongtransversalindicator*

Recalling that for $\ell\in\br{0,\dots,k}$, $\Z{k}^{2^\ell}=\Z{k-\ell}$, for a standard subcube $\standard{K}\subcubeeq\ZZ_2^m$ we can restate \cref{lem: strong transversal indicator} as:
\begin{equation}\label{eq: transversal Zk decomposition into tilde top-down}
    \Z{k}_{\standard{K}} = \prod_{i=\abs{K}-k}^{\abs{K}}\prod_{J\subseteq K\colon \abs{J}=i} \tildeZ{k-(\abs{K}-i)}_{\standard{J}}^{(-1)^i}.
\end{equation}

The index in \cref{eq: transversal Zk decomposition into tilde} starts from $i=\abs{K}-k$ instead of $i=0$ as by definition $\Z{k-(\abs{K}-i)}=\eye$ whenever $i<\abs{K}-k$. For simplicity of notation, we will find it easier to re-index the outer product in \cref{eq: transversal Zk decomposition into tilde} so that it does begin from $0$:
\begin{equation}\label{eq: transversal Zk decomposition into tilde}
    \Z{k}_{\standard{K}} = \prod_{j=0}^{k}\left(\prod_{J\subseteq K\colon \abs{J}=\abs{K}-j} \tildeZ{k-j}_{\standard{J}}^{(-1)^{\abs{K}-j}}\right).
\end{equation}

The $j=0$ term in \cref{eq: transversal Zk decomposition into tilde} is $\tildeZ{k}_{\standard{K}}$, up to a possible inversion depending on the parity of $\abs{K}$. Ignoring the possible inversions, each time $j$ increases by one \cref{eq: transversal Zk decomposition into tilde} says that we include all of the signed subcube operators in one level lower on the Clifford Hierarchy that act on standard subcubes of $K$ with one dimension less than the previous.

It turns out that the inversions in \cref{eq: transversal Zk decomposition into tilde} are not necessary. It is straightforward to show that if $K\in\mcQ_k$ is the index of a $k$-th level logical operator for $QRM_m(q,r)$ then every term in \cref{eq: transversal Zk decomposition into tilde} is a logical operator for the code. In particular, let $\tildeZ{k-j}_{\standard{J}}$, $\abs{J}=\abs{K}-j$, be one of the terms in \cref{eq: transversal Zk decomposition into tilde}, ignoring the possible inversion. We can lower bound $\abs{J}$ as
\begin{align}
    \abs{J} &= \abs{K}-j,\\
    \since{$K\in\mcQ_k$}&\geq q+kr+1 - j,\\
    \since{$r\geq 1$} &\geq q+(k-j)r+1,
\end{align}
which by \cref{thm: subcube dimension implies logic} implies that $\tildeZ{k-j}_{\standard{J}}\in\errors{k-j}$. Therefore, \cref{cor: tilde operators are logically Hermitian} implies that each $\tildeZ{k-j}_{\standard{J}}$ is logically equivalent to its Hermitian, and for $K\in\mcQ_k$ we have 
\begin{equation}\label{eq: transversal Zk decomposition into tilde (no phases)}
    \Z{k}_{\standard{K}} = \prod_{j=0}^{k}\left(\prod_{J\subseteq K\colon \abs{J}=\abs{K}-j} \tildeZ{k-j}_{\standard{J}}\right).
\end{equation}

It turns out that in many cases, the operators $\tildeZ{k-j}_{\standard{J}}$ are actually stabilizers of $QRM_m(q,r)$ instead of non-trivial logical operators. Recall that the collections $\mcQ_k$ are defined via:
\begin{equation}
    \mcQ_k\coloneq\br{K\subseteq S\Bigmid q+kr+1\leq \abs{K}\leq (k+1)r}.
\end{equation}
The subsets $\{q+kr+1,\dots,(k+1)r\}\subset \NN$ defining the $\mcQ_k$ are disjoint and if $q=0$, they also form
a partition of $\NN$. In particular, and $\{q+(k-1)r+1,\dots,kr\}$ and $\{q+kr+1,\dots,(k+1)r\}$  are separated by $q$ integers. 

Now, note that each time the index $j$ increases in \cref{eq: transversal Zk decomposition into tilde (no phases)} we decrease the level of the Clifford Hierarchy of operators by one while only decreasing the dimension of the subcubes they act on by one, as well. One simple consequence of this is that if $q\geq 1$, then \emph{only} the $j=0$ term can act non-trivially on $QRM_m(q,r)$. For example, if $J=K\setminus\br{i}$ then $\abs{J}$ is too large for $J$ to be in $\mcQ_{k-1}$, and, in particular, its dimension is large enough to imply $\tildeZ{k-1}_{\standard{J}}\in\stabs{k-1}$ by \cref{thm: subcube dimension implies logic}. In the next few results we enumerate all possibilities for the logic implemented by $\Z{k}_{\standard{K}}\in\errors{k}$ in terms of the logic implemented by the signed operators.

\begin{theorem}[Conditions when $\Z{k}_{\standard{K}}\equiv\tildeZ{k}_{\standard{K}}$]\label{thm: Zk and tildeZk equivalence conditions}
    Consider $QRM_m(q,r)$ and let $K\in\mcQ_k$. The following are true: 
    \begin{enumerate}
        \item If $q\geq 1$ then $\Z{k}_{\standard{K}}\equiv \tildeZ{k}_{\standard{K}}$. \emph{(\cref{subfig:a})}
        \item If $\abs{K} \geq q+kr+2$ then $\Z{k}_{\standard{K}}\equiv \tildeZ{k}_{\standard{K}}$. \emph{(\cref{subfig:b})}
    \end{enumerate}
\end{theorem}
\begin{proof}
    Consider an operator $\tildeZ{k-j}_{\standard{J}}$, $\abs{J}=\abs{K}-j$, in the decomposition of $\Z{k}_{\standard{K}}$ given in \cref{eq: transversal Zk decomposition into tilde (no phases)}. The assertions will hold if for each $j\geq 1$, $\tildeZ{k-j}_{\standard{J}}\in\stabs{k-j}$. By \cref{thm: subcube dimension implies logic}, it is sufficient to show $\abs{J}\geq(k-j+1)r+1$ in both cases.

    As $r > q\geq 0$, the inequality $(j-1)r\geq j-1$ holds for all integers $j\geq 1$. Starting from this inequality, we see that
    \begin{align}
        (j-1)r&\geq j-1,\\
        \Longleftrightarrow 1-j&\geq -jr+r,\\
        \Longleftrightarrow kr+2-j&\geq kr-jr+r+1 ,\\
        \Longleftrightarrow kr+2-j&\geq (k-j+1)r+1.\label{eq: Zk tildeZk equivalence LHS}
    \end{align}
    In both cases, $\abs{J}$ is larger than the term on the left-hand side of \cref{eq: Zk tildeZk equivalence LHS}:
    \begin{enumerate}
        \item If $q\geq 1$ then $\abs{J}\geq q+kr+1-j\geq kr+2-j$.
        \item Since $q\geq 0$, if $\abs{K}\geq q+kr+2$ then, once again, $\abs{J}\geq q+kr+2-j\geq kr+2-j$.
    \end{enumerate}
    Thus, $j\geq 1$ implies that $\abs{J}\geq(k-j+1)r+1$, which by \cref{thm: subcube dimension implies logic} forces $\tildeZ{k-j}_{\standard{J}}\in\stabs{k-j}$. So $\Z{k}_{\standard{K}}\equiv \tildeZ{k}_{\standard{K}}$, as desired.
\end{proof}

\cref{thm: Zk and tildeZk equivalence conditions} and \cref{thm: subcube dimension implies logic} imply that in order for $\Z{k}_A$ and $\tildeZ{k}_A$ to perform different (non-trivial) actions on $QRM_m(q,r)$, it is necessary that $q=0$ and $\abs{K}=kr+1$.
\begin{lemma}[\cref{subfig:c}]\label{lem: q=0 r>1 unsigned logic}
    Consider $QRM_m(0,r)$ and suppose $K\subseteq S$ satisfies $\abs{K}=kr+1$ for $k\in\ZZ_{\geq 0}$. If $r\geq 2$ then
    \begin{equation}
        \Z{k}_{\standard{K}}\equiv \overline{C^{\mcF(K)}Z} \cdot\prod_{i\in K} \overline{C^{\mcF(K\setminus\br{i})}Z}.
    \end{equation}
\end{lemma}
\begin{proof}
    When $r\geq 2$ the inequality $(j-1)r\geq j$ holds for all integers $j\geq 2$, and with rearranging, is equivalent to $kr+1-j\geq (k-j+1)r+1$ for all $j\geq 2$.
    Thus, for each $\tildeZ{k-j}_{\standard{J}}$ in the decomposition of $\Z{k}_{\standard{K}}$ in \cref{eq: transversal Zk decomposition into tilde (no phases)} with $j\geq 2$, we have that $\abs{J}=kr+1-j\geq (k-j+1)r+1$ implying by \cref{thm: subcube dimension implies logic} that $\tildeZ{k-j}_{\standard{J}}\in\stabs{k-j}$. Thus, 
    \begin{equation}
        \Z{k}_{\standard{K}}\equiv \tildeZ{k}_{\standard{K}}\cdot\prod_{J\subseteq K\colon \abs{J}=kr} \tildeZ{k-1}_{\standard{J}}.
    \end{equation}
    Now for each such $J$, $\abs{J}\geq(k-1)r+1$, implying that $J\in\mcQ_{k-1}$. The desired result then holds by \cref{thm: logical multi-controlled-Z circuit}.
\end{proof}

We consider now the remaining case when $q=0$ and $r=1$:

\begin{lemma}[\cref{subfig:d}]\label{lem: hypercube codes unsigned logic}
    Consider the hypercube code family, $QRM_m(0,1)$. For each $K\subseteq S$,
    \begin{equation}
        \Z{\abs{K}-1}_{\standard{K}} \equiv \overline{C^{\powerset{K}}Z}.
    \end{equation}
\end{lemma}
\begin{proof}
    By \cref{eq: transversal Zk decomposition into tilde (no phases)},
    \begin{equation}
        \Z{\abs{K}-1}_{\standard{K}} = \prod_{J\subseteq K\setminus\emptyset} \tildeZ{\abs{J}-1}_{\standard{J}}.
    \end{equation}
    As $\abs{J}=(\abs{J}-1)+1$ we have that $J\in\mcQ_{\abs{J}-1}$, so by \cref{thm: subcube dimension implies logic} it must be that $\tildeZ{\abs{J}-1}_{\standard{J}}\in\logs{\abs{J}-1}$. 
    The result holds by \cref{thm: logical multi-controlled-Z circuit} and by definition of the composition of multi-controlled-$Z$ circuits.
\end{proof}

\begin{figure}[b]
    \subfloat[(\cref{thm: Zk and tildeZk equivalence conditions}.1) Decomposition of $\Z{2}_{\standard{K}}$ where $\abs{K}=5$, acting on $QRM_6(1,2)$. 
    Since $q\geq 1$, the dimensions that admit logical operators \emph{do not} partition $\NN$. As a result, every operator $\tildeZ{k-j}_{\standard{J}}$ with $j\geq 1$ necessarily has size \emph{larger} than the bounds for its given logical index set $\mcQ_{k-j}$. That is, other than the $\tildeZ{k}_{\standard{K}}$ operator, every signed operator in the decomposition acts trivially on the code.]{%
        \includegraphics[width=.49\linewidth]{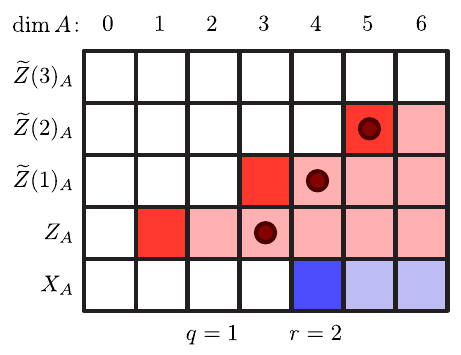}%
        \label{subfig:a}%
    }\hfill
    \subfloat[(\cref{thm: Zk and tildeZk equivalence conditions}.2) Decomposition of $\Z{2}_{\standard{K}}$ where $\abs{K}=6$, acting on $QRM_6(0,2)$. As $\abs{K}=5$ is \emph{not} the lowest size for a set in $\mcQ_2$, reducing dimension/level of the CH by one immediately implies that an operator is trivial. This case can only happen when $r-q\geq 2$, but is independent of the choice of $q$.]{%
        \includegraphics[width=.49\linewidth]{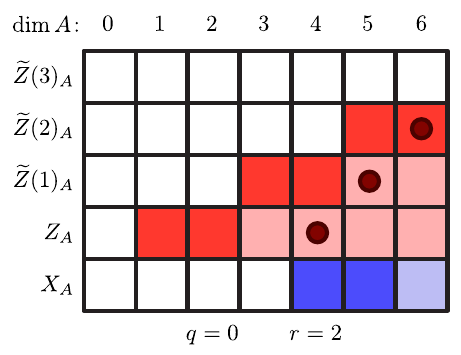}%
        \label{subfig:b}%
    }\\
    \subfloat[(\cref{lem: q=0 r>1 unsigned logic}) Decomposition of $\Z{2}_{\standard{K}}$ where $\abs{K}=5$, acting on $QRM_6(0,2)$. Since $q=0$, every dimension does admit a logical operator. As $\abs{K}=5$ is the lowest size for a set in $\mcQ_2$, reducing dimension/level of the CH by one remains a logical operation. However, as $r\geq 2$ every operator $\tildeZ{k-j}_{\standard{J}}$ with $j\geq 2$ necessarily acts trivially on the code.]{%
        \includegraphics[width=.49\linewidth]{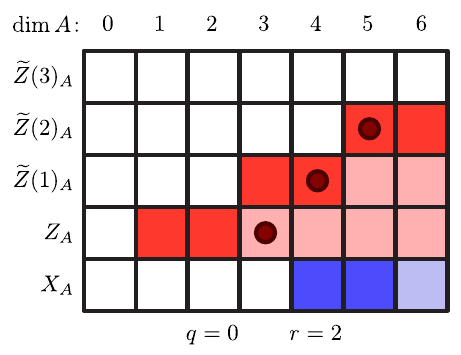}%
        \label{subfig:c}%
    }\hfill
    \subfloat[(\cref{lem: hypercube codes unsigned logic}) Decomposition of $\Z{5}_{\standard{K}}$ where $\abs{K}=4$, acting on the hypercube code $QRM_6(0,1)$. Each time the dimension/level of the CH is reduced by one, the operator in the decomposition remains a logical operator for the code.]{%
        \includegraphics[width=.49\linewidth]{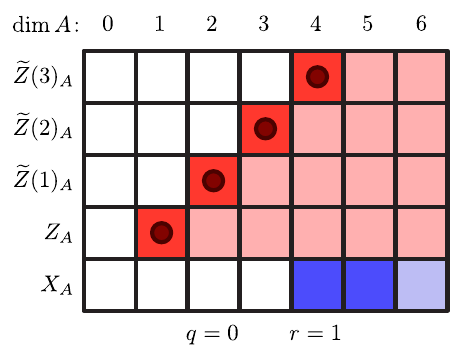}%
        \label{subfig:d}%
    }
    \caption{For $K\in\mcQ_k$, the unsigned subcube operator $\Z{k}_{\standard{K}}$ can be decomposed as a product of signed subcube operators $\tildeZ{k-j}_{\standard{J}}$ ($J\subseteq K$ with $\abs{J}=\abs{K}-j$) via \cref{eq: transversal Zk decomposition into tilde}. In the above figures, a dark box indicates a dimension where the subcube operator of the given level of the Clifford Hierarchy acts as a logical operator, a light box indicates a dimension where the operator acts trivially, and a white box indicates that the operator does not preserve the code space. In each of them, we consider the decomposition of an unsigned operator (specified in the subcaption) into a product of signed operators, represented by red dots.}
    \label{fig: Zk visuals}
\end{figure}

\newpage
\bibliography{ref,doms_refs}
\newpage
\appendix

\section{Proofs of indicator function identities}\label{app: indicator proofs}
\subsection*{Identities over \texorpdfstring{$\FF$}{F2}}

\hypertarget{pf: indicator decomposition}{}
\indicatordecomposition*

For $w\in\ZZ_2^m$, consider the set $\mcI_w\coloneqq\br{I\subseteq I_A\mid w\in\standard{I\cup J}}$. We prove two lemmas, which together will imply \cref{lem: decomposition of subcube indicators}.

\begin{lemma}\label{lem: condition for in A}
    $w\in A$ if and only if $\abs{\mcI_w}=1$. 
\end{lemma}
\begin{proof} For a subset $L\subset[m]$ we write $e_L:=\{e_i,i\in L\}$.

    ($\Rightarrow$) We will show that $\mcI_w=\{I_A\}.$ Let $w=e_{I_A}+e_K$ for $K\subseteq J$, and let $I\in\mcI_w$, which by definition implies that there exist $I'\subseteq I$, $J'\subseteq J$, such that $w=e_{I'}+e_{J'}$. Since $e_{I_A}+e_K = e_{I'}+e_{J'}$ and since $I_A\cap J=\emptyset$, we have that $e_{I_A}+e_{I'}=e_K+e_{J'}=0$. Thus $I' = I_A$ and $J'=K$. Therefore $\mcI_w$ is formed of a single element, $I_A$, and so $\abs{\mcI_w}=1$.
    
    ($\Leftarrow$) We first note that if $w\in\standard{K}$ for some $K\subseteq S$ then $w\in \standard{K'}$ for every $K\subseteq K'\subseteq S$, as well. Thus, assuming $\abs{\mcI_w}=1$ we have that $\mcI_w=\br{I_A}$. As $w\in\standard{I_A\cup J}$ but $w\notin\standard{I\cup J}$ for any $I\subsetneq I_A$, we have $w=e_{I_A}+e_{K}$ for some $K\subseteq J$. Since $I_A\cap J=\emptyset$, this is precisely the condition for $w\in A$.
\end{proof}

\begin{lemma}\label{lem: size of Iw is a power of 2}
    For all $w\in\ZZ_2^m$, either $\abs{\mcI_w}=\emptyset$ or $\abs{\mcI_w}=2^\ell$ for some $\ell\in\ZZ_{\geq0}$.
\end{lemma}
\begin{proof} If $w\not\in \standard{I_A\cup J},$ then clearly $\mcI_w=0$. Otherwise, there is an $I\subset I_A$ such that $w\in\standard{I\cup J}$. Further, there is a unique subcube $S(w)\coloneqq I^\ast\cup J^\ast$
of minimum cardinality such that $w\in\standard{S(w)}$ (if there are two or more, then $w$ is also contained
in their intersection, itself a subcube), and $I^\ast\subset I$. Uniqueness together with $I_A\cap J=I^*\cap J=\emptyset$ implies that $I^*\in \mcI_w$ is the unique element in $\mcI_w$ with minimum cardinality, as well. As noted above, for any $I\in\mcI_w$, $I^*\subseteq I$, and so $\mcI_w = \br{I\subseteq S\mid I^*\subseteq I\subseteq I_A}$, implying that $\abs{\mcI_w}=2^{\abs{I_A}-\abs{I^*}}$.
\end{proof}

\begin{proof}[Proof of \cref{lem: decomposition of subcube indicators}]
 For $w\in\ZZ_2^m$, $\sum_{I\subseteq I_A} \indicator{\standard{I\cup J}}(w) = \abs{\mcI_w} \pmod{2}$.  For all $w\in A$, by definition of the indicator function and by \cref{lem: condition for in A} the two functions in \cref{eq: subcube decomposition} are equal to 1. For all $w\notin A$, \cref{lem: condition for in A} and \cref{lem: size of Iw is a power of 2} together imply that $\abs{\mcI_w}\equiv 0\pmod{2}$, so the two functions are again equal.
\end{proof}

\hypertarget{pf: indicator decomposition alternate}{}
We give one more proof of \cref{lem: decomposition of subcube indicators}, which will later be
generalized to prove other decomposition results in the next section.
\begin{proof}[Alternate proof of \cref{lem: decomposition of subcube indicators}]
    We will prove identity \cref{eq: subcube decomposition} by induction on $\abs{I_A}$. If $\abs{I_A}=0$ then $A=\standard{J}$ and the statement is clearly true. Now suppose it holds for all $A'$ with $\abs{I_{A'}}=k\geq 0$. Take an $A=e_{I_A}+\standard{J}$ with $\abs{I_{A}}=k+1$, and assume without loss of generality that $1\in I_{A}$ (which also implies that $1\notin J$). Now for $B\coloneqq e_{I_A\setminus\br{1}}+\standard{J}$ and $C\coloneqq e_{I_A\setminus\br{1}}+\standard{\br{1}\cup J}$, by the induction hypothesis we have
    \begin{align}
        \sum_{I\subseteq I_A} \indicator{\standard{I\cup J}} &= \sum_{I\subseteq I_A\setminus\br{1}} \indicator{\standard{I\cup J}} + \sum_{I\subseteq I_A\setminus\br{1}} \indicator{\standard{ I\cup\br{1}\cup J}},\\
        &= \indicator{B} + \indicator{C},\\
        &= \indicator{C\setminus B}.
    \end{align}
The vectors $x$ in the cubes $A,B$, and $C$ have $1$'s in the positions of $I_A\backslash\{1\}$ and
all the possible entries in the positions of $J$. Additionally, $C$ allows both 1 and 0 in position 1, 
$$
C=\{x=(\underbrace{\ast11\dots1}_{I_A}00\dots0\underbrace{\ast\ast\dots\ast}_{J}0\dots0)\},
$$
while the vectors $x\in A$ have $x_1=1$ and the vectors in $B$ have $x_1=0$.
Writing this formally, we obtain
    \begin{align*}
       A &= \br{e_{I_A}+e_{J'}\mid J'\subseteq J}, \\
       &= \br{e_{I_A\setminus\br{1}}+e_{\br{1}}+e_{J'}\mid J'\subseteq J}, \\
       &= \br{e_{I_A \setminus \br{1}}+e_{J'\cup\br{1}}\mid J'\subseteq J},\\
       &= C\setminus B,
    \end{align*}
and so $\indicator{C\setminus B}=\indicator{A}$, as desired.
\end{proof}


\subsection*{Identities over \texorpdfstring{$\ZZ$}{Z}}
We recall the statement of \cref{lem: bar indicator expression} which gives the decomposition of a $\ZZ$-valued indicator function for an arbitrary subcube $A\subcubeeq\ZZ_2^m$ into indicator functions on standard subcubes. First, recall that given $A\coloneqq x+\standard{J}$ of type $J$, there is a unique minimal-weight element $x\in A$ whose support is disjoint from $J$. Denoting the support of this element by $I_A\coloneq\supp(x)$, we have the following:
\hypertarget{pf: bar indicator expression}{}
\barindicatordecomposition*

We note that
$\barindicator{A}=\barindicator{e_{I_A\setminus\emptyset}+\standard{\emptyset\cup J}}$, so adding the second sum to the left-hand side of \cref{eq: bar indicator equality} we see that \cref{lem: bar indicator expression}
can be equivalently stated as follows.

\begin{lemma}[Equivalent version of \cref{lem: bar indicator expression}]\label{lem: bar indicator expression with proof}
Let $A\subcubeeq\ZZ_2^m$ be a subcube of type $J$. Then
    \begin{equation}\label{eq: bar indicator equality proven}
        \sum_{I\subseteq I_A} \barindicator{\standard{I\cup J}} = \sum_{i=0}^\abs{I_A}\sum_{I\subseteq I_A\colon \abs{I}=i} 2^i\cdot \barindicator{e_{I_A\setminus I}+\standard{I\cup J}}.
    \end{equation}
\end{lemma}

This lemma is proved by induction using the following straightforward observation:
\begin{align*}
    \barindicator{\standard{J}}+\barindicator{\standard{J\cup\br{1}}} &= \barindicator{e_1+\standard{J}}+2\cdot\barindicator{\standard{J}}. \tag{$*$}
\end{align*}
In other words, if we add the indicator functions of two standard subcubes, one of which is one dimension larger and contains the other, we double count the terms in the smaller subcube ($2\cdot\barindicator{\standard{J}}$), but only count the remaining terms, precisely the elements of the coset $e_1+\standard{J}$, once ($\indicator{e_1+\standard{J}}$).
Using this relation repeatedly accounts for the accumulating powers of 2 in \cref{eq: bar indicator equality proven}.


\begin{proof}[Proof of \cref{lem: bar indicator expression with proof}]
 Induction on $\abs{I_A}$. If $\abs{I_A}=0$ then $A=\standard{J}$ and the statement is clearly true. Now suppose it holds for all $A'$ with $\abs{I_{A'}}=k\geq 0$. Take an $A=e_{I_A}+\standard{J}$ with $\abs{I_{A}}=k+1$, and assume without loss of generality that $1\in I_{A}\setminus J$.
    
    We compute:

    \begin{align}
        \sum_{I\subseteq I_A} \barindicator{\standard{I\cup J}} &= \sum_{I\subseteq I_A\setminus\br{1}} \barindicator{\standard{I\cup J}} + \sum_{I\subseteq I_A\setminus\br{1}} \barindicator{\standard{ I\cup\br{1}\cup J}}, \\
        \since{I.H.}&= \sum_{i=0}^\abs{I_A\setminus\br{1}}\sum_{I\subseteq I_A\setminus\br{1}\colon \abs{I}=i} 2^i\cdot \barindicator{e_{I_A\setminus (I\cup\br{1})}+\standard{I\cup J}} \\
        &\hspace{2em}+ \sum_{i=0}^\abs{I_A\setminus\br{1}}\sum_{I\subseteq I_A\setminus\br{1}\colon \abs{I}=i} 2^i\cdot \barindicator{e_{I_A\setminus (I\cup\br{1})}+\standard{I\cup\br{1}\cup J}}\nonumber\\
        \since{Combining sums} &= \sum_{i=0}^\abs{I_A\setminus\br{1}}\sum_{I\subseteq I_A\setminus\br{1}\colon \abs{I}=i} 2^i\cdot\left(\barindicator{e_{I_A\setminus (I\cup\br{1})}+\standard{I\cup J}} + 
        \barindicator{e_{I_A\setminus (I\cup\br{1})}+\standard{I\cup\br{1}\cup J}}
        \right), \\
        \since{$*$}&= \sum_{i=0}^\abs{I_A\setminus\br{1}}\sum_{I\subseteq I_A\setminus\br{1}\colon \abs{I}=i}\left( 2^i\cdot \barindicator{e_{I_A\setminus I}+\standard{I\cup J}} + 2^{i+1}\cdot \barindicator{e_{I_A\setminus (I\cup\br{1})}+\standard{I\cup\br{1}\cup J}}\right),\label{eq: main step decomp lemma}\\
        &=\sum_{i=0}^\abs{I_A}\sum_{I\subseteq I_A\colon \abs{I}=i} 2^i\cdot \barindicator{e_{I_A\setminus I}+\standard{I\cup J}},
    \end{align}
    as desired.
\end{proof}

\hypertarget{pf: strong transversal indicator}{}
\strongtransversalindicator*

\begin{proof}[Proof of \cref{lem: strong transversal indicator}]
    The former statement is a simple corollary of the latter. We will prove the latter statement by showing the two sides evaluate to the same number for every $x\in\ZZ_2^m$. Let $x\in\ZZ_2^m$. We proceed in cases.
    \begin{enumerate}[label=\textbf{\Roman*.},leftmargin=*]
        \item ($x\notin\standard{K})$. Clearly $x\notin\standard{J}$ for any $J\subseteq K$, so the left and right hand sides of \cref{eq: standard subcube bar indicator decomposition} both evaluate to 0.
        \item ($x\in\standard{K}$) In this case $\supp(x)\subseteq K$. 
        For $i\in\br{0,\dots,\abs{K}}$ define the collection $\mcK_i\subseteq\powerset{K}$ as
        \begin{equation}
            \mcK_i\coloneqq \br{J\subseteq K\Bigmid J\supseteq\supp(x),\; \abs{J}=i},
        \end{equation}
        so that 
        \begin{equation}\label{eq: sub-collections appendix proof}
            \sum_{i=0}^{\abs{K}}\sum_{J\subseteq K\colon \abs{J}=i} 2^{\abs{K}-i}(-1)^i \cdot \tildeindicator{\standard{J}}(x) = \sum_{i=0}^{\abs{K}} 2^{\abs{K}-i} (-1)^{i+\abs{x}} \abs{\mcK_i}.
        \end{equation}
        Our goal is to show that the expression on the right hand side evaluates to 1.

        Clearly if $i<\abs{x}$ then $\abs{\mcK_i}=0$. Consider now $i\geq \abs{x}$. 
        We seek to determine the number of subsets $J$ for which $\supp(x)\subseteq J\subseteq K$, which have precisely $i\geq \abs{x}$ elements. Since $J$ must contain all of $\supp(x)$, we must pick $i-\abs{x}$ additional elements from the set $K\setminus\supp(x)$ to complete a set $J\in\mcK_i$. As this set has size exactly $\abs{K\setminus\supp(x)}=\abs{K}-\abs{x}$, we have that
        \begin{equation}
            \abs{\mcK_i} = \binom{\abs{K}-\abs{x}}{i-\abs{x}}.
        \end{equation}
        Continuing from \cref{eq: sub-collections appendix proof}, we therefore have
        \begin{align}
            \sum_{i=0}^{\abs{K}}2^{\abs{K}-i} (-1)^{i+\abs{x}} \abs{\mcK_i} &= \sum_{i=\abs{x}}^{\abs{K}} 2^{\abs{K}-i}(-1)^{i+\abs{x}} \binom{\abs{K}-\abs{x}}{i-\abs{x}} ,\\
            \since{$(-1)^n=(-1)^{-n}$}&= \sum_{i=\abs{x}}^{\abs{K}}  2^{\abs{K}-i}(-1)^{i-\abs{x}} \binom{\abs{K}-\abs{x}}{i-\abs{x}}.\\
            \intertext{Re-indexing the summation and using the binomial theorem we have}
            &=\sum_{j=0}^{\abs{K}-\abs{x}}  2^{(\abs{K}-\abs{x})-j} (-1)^j \binom{\abs{K}-\abs{x}}{j},\\
            &= (2-1)^{\abs{K}-\abs{x}},\\
            &= 1,
        \end{align}
        as desired.
    \end{enumerate}
\end{proof}

\section{Classifying quantum Reed--Muller code logic}
\subsection{Unsigned \texorpdfstring{$\Z{k}$}{Zk} subcube operators}\label{sec: unsigned dimension conditions}
We now turn to the case of unsigned subcube operators. As mentioned before, proofs in this case will be nearly unchanged from the corresponding versions of the previous section, except that the condition for an operator $\Z{k}_A X_B \Z{k}_A^{\dagger}$ to be phase-free is slightly different. The condition does not change the validity of the previous proofs, but we will repeat them here with the minor changes made.

Recalling that $\omega_k = e^{-i \frac \pi {2^k}}$, a direct consequence of the conjugation identities for $\Z{k}$ and $\Z{k}^\dagger$ given in \cref{lem:conjugation-identities} is the following:
\begin{fact}\label{fact: unsigned conjugation identity}
    Let $A,B\subcubeeq\ZZ_2^m$ be subcubes.
    \begin{align}
        \Z{k}_A X_B \Z{k}_A^{\dagger} &= \omega_k^{\abs{A\cap B}}\Z{k-1}_{A\cap B} X_B.  
    \end{align}
\end{fact}

As the order of $\omega_k$ is $2^{k+1}$, and the number of elements in a subcube is two to the power of its dimension, the following is immediate:

\begin{corollary}\label{cor: Zk-conjugation-identity}
    For subcubes $A,B\subcubeeq\ZZ_2^m$, $\Z{k}_A X_B \Z{k}_A^{\dagger}$ is phase-free if and only if $\dim A\cap B\geq k+1$.
\end{corollary}

\begin{lemma}\label{lem: unsigned operator phase free condition}
    For an arbitrary subcube $A\subcubeeq\ZZ_2^m$, $\Z{k}_A X_B \Z{k}_A^{\dagger}$ is phase-free for every $X$ stabilizer generator, $X_B$, of $QRM_m(q,r)$ if and only if $\dim A\geq q+k+1$.
\end{lemma}
\begin{proof}
    By definition of $QRM_m(q,r)$, the $X_B$'s that are stabilizer generators are precisely those with $\dim B=m-q$. By \cref{cor: Zk-conjugation-identity}, the statement is equivalent to: $\Z{k}_A X_B \Z{k}_A^{\dagger}$ is phase-free if and only if $\dim A\cap B\geq k+1$ for every subcube $B$ that has non-trivial intersection with $A$ and satisfies $\dim B\geq m-q$. The desired result therefore holds by \cref{lem:overlap-guarantee}.
\end{proof}

\begin{claim}\label{clm: trivial logic unsigned}
    For $k\in\ZZ_{\geq 0}$ and a subcube, $A$, $\tildeZ{k}_A$ is a level-$k$ Clifford stabilizer for $QRM_m(q,r)$ if and only if $\dim A\geq (k+1)r+1$.
\end{claim}
\begin{proof}
    Induction on $k$. When $k=0$, $\tildeZ{0}_A=Z_A$ and the statement is true by \cref{lem: base cases}. We suppose now that the statement is true for $k\geq 0$, and we consider $k+1$.

    ($\Rightarrow$) Suppose for contradiction that there exists a subcube $A$ such that (1) $\Z{k+1}_A\in\stabs{k+1}$, but (2) $\dim A\leq (k+2)r$. As $A$ is a subcube, there exist $x\in\ZZ_2^m$
    and $K\subseteq S$ ($\abs{K}=\dim A$), such that $A=x+\standard{K}$.

    Now, by assumption of $\Z{k}_A\in\stabs{k+1}$, it must be true that for every $X$ logical subcube operator, $X_B$, the operators $\Z{k}_A X_B \Z{k}_A^\dagger\equiv X_B$ are equivalent. Given our assumption on the upper bound on the dimension of $A$ we will give a contradiction by constructing a logical $X_B$ for which $\Z{k}_A X_B \Z{k}_A^\dagger\not\equiv X_B$.

    Let $K^*\subseteq S$ be any subset of $S$ with $\abs{K^*}=m-r$ elements such that $S\setminus K\subseteq K^*$. Define the subcube $B\coloneqq x+\standard{K^*}$, so that $A\cap B = x +\standard{K\cap K^*}\neq \emptyset$. By \cref{fact: unsigned conjugation identity} we have that $\Z{k+1}_A X_B \Z{k+1}_A^\dagger = \omega_k^{\abs{A\cap B}}\Z{k}_{A\cap B} X_B$. This implies that for $\Z{k+1}_A X_B \Z{k+1}_A^\dagger \in\stabs{k}$ to be true it must be that $\Z{k}_{A\cap B}\in\stabs{k}$, as otherwise $\omega_k^{\abs{A\cap B}}\Z{k}_{A\cap B} X_B\ket\psi$ cannot equal $X_B\ket\psi$ for every code state $\ket\psi\in QRM_m(q,r)$. Therefore, we have deduced that $\Z{k}_{A\cap B}\in\stabs{k}$. This, however, contradicts our induction hypothesis: we can upper bound
    \begin{align}
        \dim A\cap B &= \abs{K\cap K^*}, \\
        &= \abs{K^*\setminus (S\setminus J)},\\
        \since{$S\setminus J\subseteq K$}&= \abs{K^*} - \abs{S\setminus J},\\
        &= (m - r)- (m -\abs{J}),\\
        &= \abs{J}-r,\\
        \since{$\dim A\leq (k+2)r$ by (2)} &\leq  (k+1)r,
    \end{align}
    but by the induction hypothesis it must be that $\dim A\cap B\geq (k+1)r+1$. Thus, $\Z{k}_{A\cap B}\notin\stabs{k}$ implying that $\Z{k+1}_A\notin \stabs{k+1}$.

    ($\Leftarrow$) Suppose $A$ is a subcube with $\dim A\geq (k+2)r+1$. Let $B$ be an arbitrary subcube for which $X_B$ is an undetectable $X$ error, which by \cref{lem: base cases} occurs if and only if $\dim B\geq m-r$. By \cref{lem:Clifford-stabilizer-equivalences}, the desired result, $\Z{k+1}_A\in\stabs{k+1}$, holds if and only if $\Z{k+1}_A X_B \Z{k+1}_A^{\dagger}\equiv X_B$. Thus, we consider the operator $\Z{k+1}_A X_B \Z{k+1}_A^{\dagger}$.

    As $\dim A\geq k+r+1\geq k+q+1$, \cref{lem: unsigned operator phase free condition} tells us that the operator is phase-free, and so $\Z{k+1}_A X_B \Z{k+1}_A^{\dagger} = \Z{k}_{A\cap B} X_B$. If we can show that $\Z{k}_{A\cap B}$ is a level-$k$ Clifford stabilizer for the code then the desired result holds. Indeed, since $\dim A\geq (k+2)r+1=m-(m-r)+(k+1)r+1$, by \cref{lem:overlap-guarantee} we have that $\dim A\cap B\geq (k+1)r+1$ and the result holds by the induction hypothesis.
\end{proof}

\begin{claim}\label{clm: non-trivial logic unsigned}
    For $k\in\ZZ_{\geq 0}$ and a subcube, $A$, $\Z{k}_A$ is a level-$k$ undetectable Clifford error for $QRM_m(q,r)$ if and only if $\dim A\geq q+kr+1$.
\end{claim}
\begin{proof}
    By definition,  $\Z{k}_A\in\errors{k}$ if and only if $\Z{k}_A X_B \Z{k}_A^\dagger \in \stabs{k-1}$ for every $X_B$ with $\dim B= m-q$. Let $X_B$ be an arbitrary stabilizer generator. By \cref{fact: unsigned conjugation identity} we have that $\Z{k+1}_A X_B \Z{k+1}_A^\dagger = \omega_k^{\abs{A\cap B}}\Z{k}_{A\cap B} X_B$. Since $X_B$ is a stabilizer we have that $\Z{k}_A X_B \Z{k}_A^\dagger\in\stabs{k}$ if and only if $\omega_k^{\abs{A\cap B}}\Z{k}_{A\cap B}\in\stabs{k-1}$ for every subcube $B$ with $\dim B=m-q$.

    ($\Rightarrow$) We assume that $\omega_k^{\abs{A\cap B}}\Z{k}_{A\cap B}\in\stabs{k-1}$ for every subcube $B$ with $\dim B=m-q$, and we seek to show that $\dim A\geq q+kr+1$. If the global phase factor $\omega_k^{\abs{A\cap B}}\neq 1$, then $\omega_k^{\abs{A\cap B}}\Z{k-1}_{A\cap B}$ cannot fix the code space, so by \cref{lem: unsigned operator phase free condition} we have that $\dim A \geq k+q+1$ in order for $\Z{k}_A X_B \Z{k}_A^\dagger$ to be phase-free. Now, we must show that $\Z{k-1}_{A\cap B}\in\stabs{k-1}$ for every $B$ such that $\dim B = m-q$. Using \cref{clm: trivial logic unsigned}, $\Z{k-1}_{A\cap B}\in\stabs{k-1}$ if and only if $\dim A\cap B\geq kr+1$. By \cref{lem:overlap-guarantee} we have that $\dim A\cap B\geq kr+1$ for every $B$ with $\dim B = m-q$ only if $\dim A\geq m- (m-q)+kr+1 = q+kr+1$, as desired.

    ($\Leftarrow$) We assume that $\dim A\geq q+kr+1$, and we seek to show that $\omega_k^{\abs{A\cap B}}\Z{k}_{A\cap B}\in\stabs{k-1}$ for every subcube $B$ with $\dim B=m-q$. As $r\geq 1$, $\dim A\geq q+k+1$ and \cref{lem: unsigned operator phase free condition} implies that $\Z{k}_A X_B \Z{k}_A^\dagger$ is phase-free. By \cref{lem:overlap-guarantee}, since $\dim A\geq q+kr+1$ we have that $\dim A\cap B\geq kr+1$ for every $B$ with $\dim B = m-q$. \cref{clm: trivial logic unsigned} thus implies that $\Z{k-1}_{A\cap B}\in\stabs{k-1}$, as desired.
\end{proof}

\subsection{\texorpdfstring{$\X{k}$}{Xk} subcube operators}\label{app: X basis operators}

It is natural to wonder if $QRM_m(q,r)$ can support transversal Clifford Hierarchy logic in the $Z$ and $X$ bases \emph{simultaneously}. Unfortunately, the bounds given in \cref{thm: subcube dimension implies logic} and \cref{clm: subcube dimension implies logic (X basis)} are largely incompatible with one another. Since the dimension of any subcube is at most $m$, \cref{thm: subcube dimension implies logic} tells us that $\tildeZ{k}_A\in\mcE^*$ only when $q+kr+1\leq m$. For the same reason, \cref{clm: subcube dimension implies logic (X basis)} tells us that $\tildeX{k}_A\in\mcE^*$ only when $m-r+k(m-q-1)\leq m$. Adding these two bounds and rearranging terms, we see
\begin{align}
    q+kr+1+m-r+k(m-q-1)&\leq 2m,\\
    q+kr+1-r+km-kq-k - m&\leq 0,\\
    (k-1)r-(k-1)q+(k-1)m+-(k-1)&\leq 0,\\
    (k-1)(m-1 + r-q)&\leq 0.
\end{align}
This inequality is always satisfied when $k=0$, which is indicative of the fact that logical Pauli operators always exist for CSS codes. For $k\geq 2$, this inequality can never be satisfied. Therefore, a quantum Reed--Muller code can never simultaneously admit a transversal logical $\Z{k}_A\in\mcN^*$ and $\X{k}_B\in\mcN^*$ when $k\geq 2$. 

Consider now the remaining case where $k=1$ and the inequality is satisfied. In this case, the bound from \cref{thm: subcube dimension implies logic} implies that $m\geq q+r+1$, whereas the bound from \cref{clm: subcube dimension implies logic (X basis)} implies that $m\leq q+r+1$. Therefore, for $m=q+r+1$ the codes $QRM_{q+r+1}(q,r)$ \emph{do} simultaneously support global transversal $\Z{1}_{\ZZ_2^{q+r+1}}\in\mcN^*$ and $\X{1}_{\ZZ_2^{q+r+1}}\in\mcN^*$. However, by both \cref{thm: subcube dimension implies logic} and \cref{clm: subcube dimension implies logic (X basis)}  they cannot support any logical $\Z{k}_A$ or $\X{k}_A$ for any value of $k\geq 2$.

\subsection{Diagonal and transversal gates in the Clifford Hierarchy}\label{app: dtc}
Recall the space of integer-valued functions on the Boolean hypercube, $\ZZ[\ZZ_2^m]\coloneqq\br{f\colon\ZZ_2^m\rightarrow \ZZ}$, along with the unsigned and signed indicator functions of $A\subcubeeq\ZZ_2^m$: $\barindicator{A}(x)\coloneqq 1$ if $x\in A$ and 0 otherwise, and $\tildeindicator{A}(x)\coloneqq (-1)^\abs{x}$ if $x\in A$ and 0 otherwise. These indicator functions are defined so that
\begin{align}
    \Big\langle \Z{k}_A \Big\rangle &\cong \Big\langle \barindicator{A}\pmod{2^{k+1}}\Big\rangle, \\
    \Big\langle \tildeZ{k}_A \Big\rangle &\cong \Big\langle \tildeindicator{A}\pmod{2^{k+1}}\Big\rangle.
\end{align}
\cref{thm: subcube dimension implies logic} implies the following:
\begin{align}
    \left\langle e^{i\theta}\tildeZ{k}_{\standard{K}}\Bigmid k\in\ZZ_{\geq 0},\;K\subseteq S,\;\abs{K}\geq (k+1)r+1,\; \theta\in[0,2\pi)\right\rangle &\subseteq \bigcup_{k\in\ZZ_{\geq 0}}\stabs{k}, \\
    \intertext{and}
    \left\langle e^{i\theta}\tildeZ{k}_{\standard{K}}\Bigmid k\in\ZZ_{\geq 0},\;K\subseteq S,\;\abs{K}\geq q+kr+1,\; \theta\in[0,2\pi)\right\rangle &\subseteq \bigcup_{k\in\ZZ_{\geq 0}}\errors{k}.
\end{align}
A natural open question of our work is whether or not the \emph{reverse} inclusion holds, as well. For simplicity we will consider only the undetectable error set, and without global phases. That is, consider the subgroup generated by the $\tildeZ{k}_{\standard{K}}$ operators for $K\in\mcQ_k$, the index set of $k$-th level logicals:
\begin{equation}\label{eq: undetectable Clifford error group}
    \left\langle \tildeZ{k}_{\standard{K}} \Bigmid k\in\ZZ_{\geq 0},\; K\in\mcQ_k\right\rangle.
\end{equation}
\begin{question}\label{question: matrix version}
    Does the group generated by standard subcube operators given in \cref{eq: undetectable Clifford error group} fully characterize the group of undetectable Clifford errors for the code $QRM_m(q,r)$ that are diagonal and transversal? That is, up to global phases, does
    \begin{equation}
        \left\langle \tildeZ{k}_{\standard{K}} \Bigmid k\in\ZZ_{\geq 0},\; K\in\mcQ_k\right\rangle = \bigcup_{k\in\ZZ_{\geq 0}} \left(\DTC{k}\cap\:\errors{k}\right)\;?
    \end{equation}
\end{question}

Our goal in this section is to give a coding-theoretic interpretation of the group in \cref{eq: undetectable Clifford error group}. In particular, we begin by defining so-called linear codes over $\ZZ_{2^{k+1}}$.

For $k\in\NN$, let $R_k\coloneqq\ZZ_{2^{k+1}}$ denote the ring of integers modulo $2^{k+1}$. Consider now the space of $R_k$-valued functions on the Boolean hypercube, $R_k[\ZZ_2^m]\coloneqq\br{f\colon \ZZ_2^m\rightarrow R_k}$. $R_k[\ZZ_2^m]$ is a free module over $R_k$. A submodule, $C\subseteq R_k[\ZZ_2^m]$, is called an \emph{$R_k$-linear code of length $2^m$}, or simply a linear code. For two linear codes, $A,B\subseteq R_k[\ZZ_2^m]$, their sum $A+B\subseteq R_k[\ZZ_2^m]$ is a linear code defined by taking sums of elements in $A$ and $B$.

As $R_k$ is not a principal ideal domain (for $k>0$), the space $R_k[\ZZ_2^m]$ fails to have many properties that a vector space over a field does. For example, not every linear code will have a basis, and even if a code has a basis the number of elements in two given bases may be different. In other words, the dimension of an $R_k$-linear code may not be well-defined.


Given a function $f\in R_k[\ZZ_2^m]$ it is straightforward to construct a diagonal and transversal operator in the $k$-th level of the Clifford Hierarchy. In particular, define $Z(f)\in\Cl{k}$ to be
\begin{equation}
    Z(f) = \bigotimes_{x\in\ZZ_2^m} \Z{k}^{f(x)}.
\end{equation}
That is, $Z(f)$ acts as $\Z{k}^{f(x)}$ on the physical qubit indexed by $x\in\ZZ_2^m$. Similarly, supposing that $U=\bigotimes_{x\in\ZZ_2^m} \Z{k}^{P_x}\in\Cl{k}$ is a diagonal and transversal operation in the $k$-th level of the Clifford Hierarchy where each $P_x\in\ZZ$, we can define a function $f_U\in R_k[\ZZ_2^m]$ via
\begin{equation}
    f_U(x)\coloneqq P_x \pmod{2^{k+1}}.
\end{equation} 
The space of \emph{diagonal, transversal, and $k$-th level Clifford operators}, $\DTC{k}$, is defined as
\begin{equation}
    \DTC{k} \coloneqq \br{\bigotimes_{x\in\ZZ_2^m} \Z{k}^{f(x)} \biggmid f\in\ZZ[\ZZ_2^m]}
\end{equation}
As the diagonal and transversal operators in the $k$-th level of the Clifford Hierarchy admit a natural action via $R_k$ by taking powers, we see that
\begin{equation}
    R_k[\ZZ_2^m] \cong \DTC{k}
\end{equation}
as modules over $R_k$. 

We now proceed to generalize construction of classical RM codes:

\begin{definition}[$RM_k(r,m)$]\label{def:generalized-RM}
    For $r\in\br{-1,0,\dots, m}$, the Generalized Reed--Muller (GRM) 
    code
    of \emph{order} $r$ over $R_k$, denoted $RM_k(r,m)$, is defined as the linear code generated by the signed indicator functions of standard subcubes of dimension at least $m-r$, taken modulo $2^{k+1}$:
    \begin{equation*}
        RM_k(r,m)\coloneqq \left\langle \tildeindicator{\standard{J}} \pmod{2^{k+1}} \Bigmid J\subseteq S,\; \abs{J}\geq m-r \right\rangle.
    \end{equation*}
\end{definition}
We note that for all $k\in\ZZ_{\geq 0}$, 
\begin{equation}
    \left\langle \tildeZ{k}_{\standard{K}}\Bigmid K\subseteq S,\;\abs{K}\geq (k+1)r+1 \right\rangle\cong RM_k\big(m-((k+1)r+1),m\big).
\end{equation}
Thus, sufficiency in \cref{thm: subcube dimension implies logic} can be rephrased as:
\begin{theorem}
Consider $QRM_m(q,r)$.
    \begin{enumerate}
        \item If $f\in RM_k\big(m-((k+1)r+1),m\big)$, then $Z(f)\in\stabs{k}$.
        \item If $f\in RM_k\big(m-(q+kr+1),m\big)$, then $Z(f)\in\errors{k}$.
    \end{enumerate}
\end{theorem}
Ideally the converse would hold as well, that any operator $U\in\DTC{k}$ that is a Clifford error for $QRM_m(q,r)$ \emph{must} arise in this way. That is, it would be convenient if the code $RM_k\big(m-(q+kr+1),m\big)$ completely characterized $\errors{k}\cap\DTC{k}$. This, however, cannot be the case. In particular, every function in one level below, $k-1$, gives rise to an operator in $\errors{k}$ via \emph{squaring}: if $f\in RM_{k-1}\big(m-(q+(k-1)r+1),m\big)$ then $Z(f)^2\in\errors{k}$. However, $RM_{k-1}\big(m-(q+(k-1)r+1),m\big)$ is \emph{not} a subcode of $RM_k\big(m-(q+kr+1),m\big)$. In fact, a priori it is not even an $R_k$-module! At every increased level of the Clifford Hierarchy we must include all \emph{squares} from the level below. We will do so through the use of the sum of codes, $A+B$.

To begin with, we must turn $RM_{k-1}\big(m-(q+(k-1)r+1),m\big)$ into an $R_k$-module. 
For $i\in\br{-1,\dots,k}$, define the linear code, $2^{k-i}\cdot RM_k(r,m)\subseteq R_{k}[\ZZ_2^m]$, via
\begin{equation}
    2^{k-i}\cdot RM_k(r,m) \coloneqq \left\langle 2^{k-i}\cdot\tildeindicator{\standard{J}} \pmod{2^{k+1}} \Bigmid J\subseteq[m],\; \abs{J}\geq m-r \right\rangle.
\end{equation}
Note that if $i=-1$ then $2^{k-i}\cdot RM_k(r,m)=\br{0}$, and for $i\geq 0$ we have $2^{k-i}\cdot RM_k(r,m)\cong RM_{i}(r,m)$ as Abelian groups.

For any fixed $k\in\ZZ_{\geq 0}$ we now have the following isomorphisms of Abelian groups:
\begin{align}
    \bigoplus_{i=0}^{k} 2^{k-i}\cdot RM_{k}\big(m-((i+1)r+1),m\big) &\cong \left\langle \tildeZ{i}_{\standard{K}}\Bigmid i\in\br{0,\dots,k},\;K\subseteq S,\;\abs{K}\geq (i+1)r+1\right\rangle, \\
    \intertext{and}
    \bigoplus_{i=0}^{k} 2^{k-i}\cdot RM_{k}\big(m-(q+ir+1),m\big) &\cong \left\langle \tildeZ{i}_{\standard{K}}\Bigmid i\in\br{0,\dots,k},\;K\subseteq S,\;\abs{K}\geq q+ir+1\right\rangle.
\end{align}

We can therefore rephrase \cref{question: matrix version} in a coding-theoretic language:
\begin{question}\label{question: code version}
    Suppose that $U\in\DTC{k}$ for $k\in\ZZ_{\geq 0}$ and consider $QRM_m(q,r)$. If $U\in\errors{k}$ then does it hold that 
    \begin{equation}
        f_U\in\bigoplus_{i=0}^{k} 2^{k-i}\cdot RM_{k}\big(m-(q+ir+1),m\big)?
    \end{equation}
    In particular, does the following isomorphism hold:
    \begin{equation}
         \DTC{k}\cap\:\errors{k}\cong \bigoplus_{i=0}^{k} 2^{k-i}\cdot RM_{k}\big(m-(q+ir+1),m\big)  \;?
    \end{equation}
\end{question}
Note that the base case of this statement is already true, i.e., when $k=0$, \cref{question: code version} asks if
\begin{equation}
    \mcN_Z\cong RM(m-q-1,m),
\end{equation}
where $\mcN_Z$ is the group of Pauli $Z$ errors for the code and $RM(m-q-1,m)$ is a classical binary RM code. This statement is true, and it can be shown that treated as vectors of $\FF^{2^m}$ a $Z$ logical must be dual to every $X$ stabilizer. In other words, one shows that $\mcN_Z\cong RM(m-q-1,m)$ by, in fact, proving that $\mcN_Z\cong RM(q,m)^\perp$, and using the duality relation of RM codes the desired characterization holds. 

A potential proof of \cref{question: code version} would likely follow the same line of thought: given a $Z(f)\in\errors{k}$ can one prove that $f\in R_k[\ZZ_2^m]$ is \emph{dual} to every function $g$ where 
\begin{equation}\label{eq: RHS general RM}
    g\in\left(\bigoplus_{i=0}^{k} 2^{k-i}\cdot RM_{k}\big(m-(q+ir+1),m\big)\right)^\perp?
\end{equation}
We have not defined a symmetric bilinear form on $R_k[\ZZ_2^m]$ here, though a natural one can be defined and dual spaces behave as one may expect, e.g., $(C^\perp)^\perp = C$, even in the case of $R_k$-modules. A natural first step towards answering \cref{question: code version} in the affirmative would be characterizing the space in \cref{eq: RHS general RM} in terms of generalized RM codes, as opposed to the \emph{dual} of generalized RM codes.

\section{The hyperoctahedral view}\label{app: hyperoctahedral}
We have described our results in terms of the $m$-dimensional hypercube and its complex of subcubes, but there is another equivalent description of quantum RM codes in terms of the \emph{hyperoctahedral complex} in $m$ dimensions, which is dual to the hypercube construction. We will now detail this alternative construction, and also show its connection to the ball codes of \cite{vasmer2022morphing}.

In three dimensions there is a well-known duality between the cube and octahedron, where the vertices of the cube correspond to the triangular faces of the octahedron and the vertices of the octahedron correspond to the squares of the cube. This duality is easily extended to the higher dimensional \emph{hyperoctahedron}. Consider $m$-dimensional real Euclidean space, $\RR^m$, and define the coordinate, $p_i$, which has a 1 in the $i$-th position and 0's elsewhere. In short, the hyperoctahedron in $m$-dimensional space can be defined as the convex hull of the $2m$ points $\br{\pm p_i}_{i\in[m]}$. This definition, although simple, does not capture the full \emph{simplicial} structure of the hyperoctahedron, in the same way that the Boolean hypercube $\ZZ_2^m$ does not immediately capture the subcube structure of the hypercube.

Geometrically, an \emph{$\ell$-simplex} $\sigma$ is a collection of $\ell+1$ (affinely) independent set of points in Euclidean space, i.e., the convex hull of $\sigma$ does not lie in an $\ell$-dimensional flat. The dimension of a simplex is defined as $\dim\sigma=\abs{\sigma}-1$. If one simplex $\rho$ is contained in another, $\rho\subseteq \sigma$, the $\rho$ is said to be \emph{incident to $\sigma$} and this incidence is denoted by $\rho\preceq\sigma$.
A \emph{geometric simplicial complex} is a collection of simplices $\mcX=\br{\sigma}$ which is (1) downward closed (if $\sigma\in\mcX$ and $\rho\preceq\sigma$ then $\rho\in\mcX$), and (2) closed under intersections. We note that the downward closure property implies that $\emptyset\in\mcX$ for any simplicial complex.

The $2m$ points $V\coloneqq\br{\pm p_i}_{i\in[m]}$ can be used to define the \emph{$m$-dimensional hyperoctahedral (simplicial) complex}, $H_m$: A subset of points $\sigma\subseteq V$ is a simplex in $H_m$ if $\sigma$ contains at most one of $p_i$ and $-p_i$ for each $i\in m$ (but possibly neither). In particular, for each $\sigma\in H_m$ there is a unique string $x\in\br{0,1,*}^m$ where $x_i=0$ if $p_i\in \sigma$, $x_i=1$ if $-p_i\in \sigma$, and $x_i=*$ if neither $\pm p_i\in\sigma$. Note that the dimension of $\sigma$ is the number of non-null ($*$) positions in the corresponding $x$, and so the dimension of any simplex in $H_m$ is at most $m-1$. In fact, the set $\br{0,1,*}^m$ is in 1-to-1 correspondence with the simplices in $H_m$. The \emph{vertices} of $H_m$ are the strings $x\in\br{0,1,*}^m$ with a single non-null entry (i.e., corresponds to a point in $\br{\pm p_i}$) and the \emph{facets} of $H_m$ are the bit strings $x\in\br{0,1}^m$. It is straightforward to show that $H_m$ satisfies the conditions of a simplicial complex. In fact, it is a \emph{homogenenous} simplicial complex, meaning that every simplex is incident to some facet in $H_m$. Further, $H_m$ yields a \emph{simplicial structure}
on the $(m-1)$-dimensional sphere $S^{m-1}\subseteq \RR^m$: each point $\pm p_i$ has unit length, and for a simplex in $H_m$ one can take the convex hull of $\sigma$ on the sphere.

We now connect the $m$-dimensional hyperoctahedral complex to the hypercube complex. We recall the definition of the hypercube complex. Consider $\ZZ_2^m$ with its standard generating set $S\coloneqq\br{e_i}$. The $m$-dimensional hypercube is the complex of \emph{standard cosets of $S$ in $\ZZ_2^m$},
\begin{equation}
    \br{z+\standard{J}\Bigmid J\subseteq S,\; z\in\ZZ_2^m}.
\end{equation}
In particular, a coset $z+\standard{J}$ is a $\abs{J}$-dimensional subcube of $\ZZ_2^m$. The duality between the hyperoctahedron $H_m=\br{\br{0,1,*}^m}$ and the hypercube $\ZZ_2^m$ is given by
\begin{align}
    x\in\br{0,1,*}^m &\mapsto z+\standard{J} \text{ where } \left(\begin{array}{cc}
        z_i = x_i \text{ if } x_i\neq * \\
        z_i = 0 \text{ otherwise }
    \end{array}\right)
    \text{ and } J\coloneqq\br{e_i\Bigmid x_i=*},\\
    z+\standard{J}&\mapsto x\in\br{0,1,*}^m \text{ where } \left(\begin{array}{cc}
        x_i = z_i \text{ if } e_i\notin J \\
        x_i = * \text{ if }  e_i\in J
    \end{array}\right).
\end{align}
Under this correspondence, vertices of $H_m$ are equivalent to $(m-1)$-cubes in $\ZZ_2^m$ and facets of $H_m$ are equivalent to vertices in (elements of) $\ZZ_2^m$. 
In general, $\ell$-simplices in $H_m$ are equivalent to $(m-\ell-1)$-cubes in $\ZZ_2^m$. For instance, the empty simplex in $H_m$ corresponds to the \emph{entire} hypercube. See \cref{fig: dual view} for the correspondence between the 4 dimensional hypercube and hyperoctahedron.

We can define quantum RM codes using the structure of the hyperoctahedral complex. As there are $2^m$ facets in $H_m$, we will associate physical qubits with these $(m-1)$-dimensional simplices in $H_m$. Given a simplex $\sigma\in H_m$, its \emph{neighborhood} is defined at the set of facets that it is incident to, $\N(\sigma)\coloneqq\br{\rho\in H_M\mid \sigma\preceq \rho, \; \dim\rho=m-1}$. In the same way that we defined subcube operators, we can likewise define a \emph{simplex operator} $U_\sigma$ as the operator which acts as the single-qubit gate $U$ on the qubits in $\N(\sigma)$ and as identity otherwise.

\begin{definition}[Alternative definition of $QRM_m(q,r)$, \emph{cf.} \cref{def: quantum RM code hypercube}]
     Let $0\leq q\leq r\leq m$ be non-negative integers. The \emph{quantum Reed--Muller code} of order $(q,r)$ and length $2^m$, denoted by $QRM_m(q,r)$, is defined as the common $+1$ eigenspace of a Pauli stabilizer group $\mcS\coloneqq\langle S_X,S_Z\rangle$, with stabilizer generators given by
    \begin{align}
        S_X&\coloneqq \br{X_\sigma\Bigmid \sigma \text{ is a $(q-1)$-simplex}},  
        \\
        S_Z&\coloneqq \br{Z_\sigma\Bigmid \sigma \text{ is an $(m-r-2)$-simplex}},
    \end{align}
\end{definition}
This definition is equivalent to the ball code definition of the hypercube code family given in \cite{vasmer2022morphing}, where the central vertex of the ball code is given by the empty simplex.

\begin{figure}[b!]
    \centering
    \includegraphics{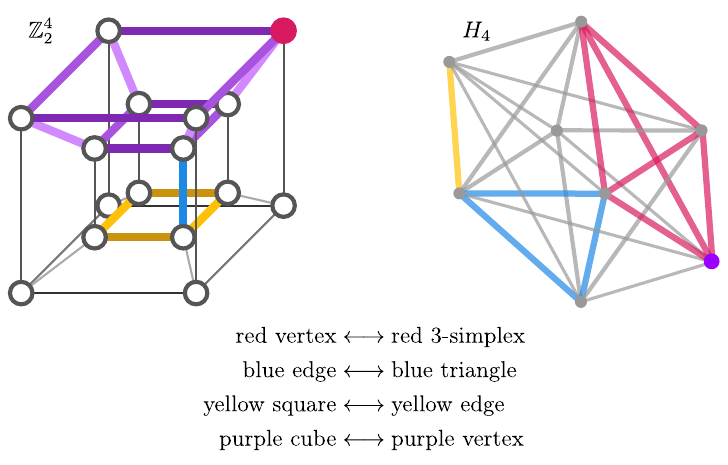}
    \caption{Equivalence between subcubes of $\ZZ_2^4$ and simplices in $H_4$.}
    \label{fig: dual view}
\end{figure}

\end{document}

%% file: macro.tex
\usepackage{booktabs}

\usepackage{amsmath,amsfonts,amsthm,mathrsfs,xspace,graphicx}
\usepackage[
colorlinks,bookmarks=true, citecolor=blue!90!black, urlcolor=blue!90!black,linkcolor=blue!90!black, pdfencoding=auto, psdextra]{hyperref}
\hypersetup{pdfauthor={Name}}
\usepackage{endnotes}
\usepackage{color}
\usepackage{float}
\usepackage[caption=false]{subfig}

\usepackage{xcolor}
\definecolor{since}{rgb}{0.5,0.5,0.5}
\definecolor{newred}{HTML}{ff382e}
\definecolor{newgreen}{HTML}{549641}
\definecolor{newblue}{HTML}{4c4cfc}
\definecolor{neworange}{HTML}{c98702}

\usepackage{mdframed}
\usepackage{bbm}
\usepackage{suffix} 
\usepackage{tabularx}
\usepackage{makecell}
\usepackage{amssymb,latexsym}
\usepackage[capitalize]{cleveref}
\usepackage{enumitem}
\usepackage{tikz}
\usetikzlibrary{quantikz2}

\usepackage{multirow}
\usepackage{endnotes}
\usepackage{amssymb,latexsym}
\usepackage{enumitem}
\usepackage{float}
\usepackage{setspace}
\usepackage{soul}
\usepackage[section]{placeins}
\usepackage[affil-it]{authblk} 
\usepackage{thmtools}
\usepackage{thm-restate}
\usepackage{stmaryrd}
\usepackage{esvect}
\usepackage{mathtools, physics}
\usepackage{complexity}
\usepackage{verbatim}
\usepackage{eczoo}

\usepackage{pdflscape}
\usepackage{adjustbox}

\usepackage{algorithm}
\usepackage{algpseudocode}

\makeatletter
\renewcommand*\env@matrix[1][*\c@MaxMatrixCols c]{%
  \hskip -\arraycolsep
  \let\@ifnextchar\new@ifnextchar
  \array{#1}}
\makeatother

\usepackage{amsthm}
\newtheorem{theorem}{Theorem}[section] 
\newtheorem*{theorem*}{Theorem}
\newtheorem{proposition}[theorem]{Proposition}
\newtheorem*{proposition*}{Proposition}
\newtheorem{lemma}[theorem]{Lemma}
\newtheorem*{lemma*}{Lemma}
\newtheorem{corollary}[theorem]{Corollary}
\newtheorem*{corollary*}{Corollary}

\newtheorem*{conjecture*}{Conjecture}

\theoremstyle{definition}
\newtheorem{fact}[theorem]{Fact}
\newtheorem*{fact*}{Fact}
\newtheorem{claim}[theorem]{Claim}
\newtheorem{definition}[theorem]{Definition}
\newtheorem*{definition*}{Definition}

\newtheorem{question}{Question}

\newtheorem{circuit}[theorem]{Logical Circuit}

\DeclareRobustCommand{\EndDef}{%
	\leavevmode\unskip\penalty9999 \hbox{}\nobreak\hfill
	\quad\hbox{$\lhd$}%
        \vspace{1em}
}

\theoremstyle{remark}
\newtheorem{remark}[theorem]{Remark}
\AtEndEnvironment{remark}{\EndDef}
\newtheorem*{remark*}{Remark}

\newclass{\QPCP}{QPCP}
\newclass{\QCPCP}{QCPCP}
\newclass{\QCMAcomp}{QCMA-complete}
\newclass{\sharpP}{\#P}

\DeclareMathOperator\im{im}
\DeclareMathOperator\clifford{Cl}
\newcommand{\Cl}[1]{\clifford^{({#1})}}

\DeclareMathOperator\evaluate{eval}

\DeclareMathOperator\CSS{CSS}
\DeclareMathOperator\N{N}

\DeclareMathOperator\unitary{U}
\DeclareMathOperator\dtc{DTC}

\DeclareMathOperator\had{H}
\DeclareMathOperator\eye{\mathbb{I}}
\DeclareMathOperator\phasegate{S}
\DeclareMathOperator{\supp}{supp}

\usepackage{accents}

\newcommand{\n}{\ensuremath{^{\otimes n}}}
\newcommand{\bits}[1]{\ensuremath{\{0,1\}^{{#1}}}}
\newcommand{\br}[1]{\ensuremath{\left\{{#1}\right\}}}

\newcommand{\bigmid}{\;\big\vert\;}
\newcommand{\Bigmid}{\;\,\Big\vert\,\;}
\newcommand{\biggmid}{\;\;\bigg\vert\;\;}

\newcommand{\standard}[1]{\ensuremath{\left\langle{#1}\right\rangle}}
\newcommand{\indicator}[1]{\mathbbold{1}_{#1}}
\newcommand{\tildeindicator}[1]{\widetilde{\mathbbold{1}}_{#1}}
\newcommand{\barindicator}[1]{{\mathbbold{1}}_{#1}}
\newcommand{\stabs}[1]{\mathcal{S}^{({#1})}}
\newcommand{\errors}[1]{\mathcal{N}^{({#1})}}
\newcommand{\logs}[1]{\mathcal{E}^{({#1})}}
\newcommand{\DTC}[1]{\dtc^{({#1})}}
\newcommand{\Z}[1]{\ensuremath{Z({#1})}}
\newcommand{\tildeZ}[1]{\ensuremath{\widetilde{Z}({#1})}}
\newcommand{\X}[1]{\ensuremath{X({#1})}}
\newcommand{\tildeX}[1]{\ensuremath{\widetilde{X}({#1})}}
\newcommand{\tildephase}{\ensuremath{\widetilde{\phasegate}}}

\newcommand{\subcubeeq}{\ensuremath{\sqsubseteq}}

\newcommand{\powerset}[1]{\ensuremath{\mathscr{P}({#1})}}

\newcommand{\CC}{\ensuremath{\mathbb{C}}}
\newcommand{\NN}{\ensuremath{\mathbb{N}}}
\newcommand{\FF}{\ensuremath{\mathbb{F}}}

\newcommand{\RR}{\ensuremath{\mathbb{R}}}
\newcommand{\ZZ}{\ensuremath{\mathbb{Z}}}

\newcommand{\mcB}{\ensuremath{\mathcal{B}}}
\newcommand{\mcC}{\ensuremath{\mathcal{C}}}
\newcommand{\mcD}{\ensuremath{\mathcal{D}}}
\newcommand{\mcE}{\ensuremath{\mathcal{E}}}
\newcommand{\mcF}{\ensuremath{\mathcal{F}}}
\newcommand{\mcG}{\ensuremath{\mathcal{G}}}

\newcommand{\mcI}{\ensuremath{\mathcal{I}}}
\newcommand{\mcJ}{\ensuremath{\mathcal{J}}}
\newcommand{\mcK}{\ensuremath{\mathcal{K}}}
\newcommand{\mcL}{\ensuremath{\mathcal{L}}}

\newcommand{\mcN}{\ensuremath{\mathcal{N}}}

\newcommand{\mcP}{\ensuremath{\mathcal{P}}}
\newcommand{\mcQ}{\ensuremath{\mathcal{Q}}}

\newcommand{\mcS}{\ensuremath{\mathcal{S}}}

\newcommand{\mcX}{\ensuremath{\mathcal{X}}}

\DeclareMathAlphabet{\mathbbold}{U}{bbold}{m}{n}

\newcommand{\since}[1]{{\color{since}\text{({#1})}}\hspace{2em}}

\newcommand\restr[2]{{
  \left.\kern-\nulldelimiterspace 
  #1 
  \right|_{#2} 
  }}
  

%% file: tables/code_table.tex
\begin{table}[p]
    \centering
\begin{tabular}{lll|lll|l}
$m$ & $q$ & $r$ & $[[n,$ & $\kappa,$ & $d]]$ & $k_{\max}$ \\ \hline
3          & 0          & 1          & 8          & 3          & 2          & 2           \\
4          & 0          & 1          & 16         & 4          & 2          & 3           \\
5          & 0          & 1          & 32         & 5          & 2          & 4           \\
5          & 0          & 2          & 32         & 15         & 2          & 2           \\
6          & 0          & 1          & 64         & 6          & 2          & 5           \\
6          & 0          & 2          & 64         & 21         & 2          & 2           \\
6          & 1          & 2          & 64         & 15         & 4          & 2           \\
7          & 0          & 1          & 128        & 7          & 2          & 6           \\
7          & 0          & 2          & 128        & 28         & 2          & 3           \\
7          & 1          & 2          & 128        & 21         & 4          & 2           \\
7          & 0          & 3          & 128        & 63         & 2          & 2           \\
8          & 0          & 1          & 256        & 8          & 2          & 7           \\
8          & 0          & 2          & 256        & 36         & 2          & 3           \\
8          & 1          & 2          & 256        & 28         & 4          & 3           \\
8          & 0          & 3          & 256        & 92         & 2          & 2           \\
8          & 1          & 3          & 256        & 84         & 4          & 2           \\
9          & 0          & 1          & 512        & 9          & 2          & 8           \\
9          & 0          & 2          & 512        & 45         & 2          & 4           \\
9          & 1          & 2          & 512        & 36         & 4          & 3           \\
9          & 0          & 3          & 512        & 129        & 2          & 2           \\
9          & 1          & 3          & 512        & 120        & 4          & 2           \\
9          & 2          & 3          & 512        & 84         & 8          & 2           \\
9          & 0          & 4          & 512        & 255        & 2          & 2           \\
10         & 0          & 1          & 1024       & 10         & 2          & 9           \\
10         & 0          & 2          & 1024       & 55         & 2          & 4           \\
10         & 1          & 2          & 1024       & 45         & 4          & 4           \\
10         & 0          & 3          & 1024       & 175        & 2          & 3           \\
10         & 1          & 3          & 1024       & 165        & 4          & 2           \\
10         & 2          & 3          & 1024       & 120        & 8          & 2           \\
10         & 0          & 4          & 1024       & 385        & 2          & 2           \\
10         & 1          & 4          & 1024       & 375        & 4          & 2          
\end{tabular}
    \caption{All quantum RM codes with $m \leq 10$ which admit at least transversal $T$ logic.}
    \label{tab:codeimplementations}
\end{table}

%% file: QRMLogic.bbl
\newcommand{\etalchar}[1]{$^{#1}$}
\begin{thebibliography}{DSRABR{\etalchar{+}}24}

\bibitem[AABA{\etalchar{+}}24]{acharya2024quantum}
Rajeev Acharya, Laleh Aghababaie-Beni, Igor Aleiner, Trond~I Andersen, Markus Ansmann, Frank Arute, Kunal Arya, Abraham Asfaw, Nikita Astrakhantsev, Juan Atalaya, et~al.
\newblock Quantum error correction below the surface code threshold.
\newblock {\em arXiv preprint arXiv:2408.13687}, 2024.

\bibitem[ABO97]{aharonov1997fault}
Dorit Aharonov and Michael Ben-Or.
\newblock Fault-tolerant quantum computation with constant error.
\newblock In {\em Proceedings of the twenty-ninth annual ACM symposium on Theory of computing}, pages 176--188, 1997.

\bibitem[ADP14]{anderson_fault-tolerant_2014}
Jonas~T. Anderson, Guillaume {Duclos-Cianci}, and David Poulin.
\newblock Fault-{{Tolerant Conversion}} between the {{Steane}} and {{Reed-Muller Quantum Codes}}.
\newblock {\em Phys. Rev. Lett.}, 113(8):080501, August 2014.
\newblock \href {https://doi.org/10.1103/PhysRevLett.113.080501} {\path{doi:10.1103/PhysRevLett.113.080501}}.

\bibitem[AK98]{AK98}
E.~F. Assmus, Jr. and J.~D. Key.
\newblock Polynomial codes and finite geometries.
\newblock In {\em Handbook of Coding Theory}, volume~II, pages 1269--1343. North-Holland, Amsterdam, 1998.

\bibitem[Ari09]{arikan2009channel}
Erdal Arikan.
\newblock Channel polarization: A method for constructing capacity-achieving codes for symmetric binary-input memoryless channels.
\newblock {\em IEEE Transactions on information Theory}, 55(7):3051--3073, 2009.

\bibitem[Ari10]{arikan2010survey}
Erdal Arikan.
\newblock A survey of {Reed-Muller} codes from polar coding perspective.
\newblock In {\em 2010 IEEE Information Theory Workshop on Information Theory (ITW 2010, Cairo)}, pages 1--5. IEEE, 2010.

\bibitem[Ass96]{assmus1996berman}
E.~F. Assmus, Jr.
\newblock On {B}erman's characterization of the {Reed-Muller} codes.
\newblock {\em Journal of Statistical Planning and Inference}, 56(1):17--21, 1996.

\bibitem[BEG{\etalchar{+}}24]{bluvstein2024logical}
Dolev Bluvstein, Simon~J Evered, Alexandra~A Geim, Sophie~H Li, Hengyun Zhou, Tom Manovitz, Sepehr Ebadi, Madelyn Cain, Marcin Kalinowski, Dominik Hangleiter, et~al.
\newblock Logical quantum processor based on reconfigurable atom arrays.
\newblock {\em Nature}, 626(7997):58--65, 2024.

\bibitem[BH12]{Bravyj2012magic}
Sergey Bravyi and Jeongwan Haah.
\newblock Magic-state distillation with low overhead.
\newblock {\em Phys. Rev. A}, 86:052329, Nov 2012.
\newblock URL: \url{https://link.aps.org/doi/10.1103/PhysRevA.86.052329}, \href {https://doi.org/10.1103/PhysRevA.86.052329} {\path{doi:10.1103/PhysRevA.86.052329}}.

\bibitem[BK05]{bravyi_universal_2005}
Sergey Bravyi and Alexei Kitaev.
\newblock Universal quantum computation with ideal {{Clifford}} gates and noisy ancillas.
\newblock {\em Phys. Rev. A}, 71(2):022316, February 2005.
\newblock \href {https://doi.org/10.1103/PhysRevA.71.022316} {\path{doi:10.1103/PhysRevA.71.022316}}.

\bibitem[BK13]{bravyi2013classification}
Sergey Bravyi and Robert K{\"o}nig.
\newblock Classification of topologically protected gates for local stabilizer codes.
\newblock {\em Phys. Rev. Lett.}, 110(17):170503, 2013.

\bibitem[BKS21]{beverland2021cost}
Michael~E Beverland, Aleksander Kubica, and Krysta~M Svore.
\newblock Cost of universality: A comparative study of the overhead of state distillation and code switching with color codes.
\newblock {\em PRX Quantum}, 2(2):020341, 2021.

\bibitem[BMD07]{Bombin2007}
H.~Bombin and M.~A. Martin-Delgado.
\newblock Topological computation without braiding.
\newblock {\em Phys. Rev. Lett.}, 98:160502, Apr 2007.
\newblock URL: \url{https://link.aps.org/doi/10.1103/PhysRevLett.98.160502}, \href {https://doi.org/10.1103/PhysRevLett.98.160502} {\path{doi:10.1103/PhysRevLett.98.160502}}.

\bibitem[BMT{\etalchar{+}}22]{beverland2022assessing}
Michael~E Beverland, Prakash Murali, Matthias Troyer, Krysta~M Svore, Torsten Hoefler, Vadym Kliuchnikov, Guang~Hao Low, Mathias Soeken, Aarthi Sundaram, and Alexander Vaschillo.
\newblock Assessing requirements to scale to practical quantum advantage.
\newblock {\em arXiv preprint arXiv:2211.07629}, 2022.

\bibitem[BW10]{bhaintwal2010generalized}
Maheshanand Bhaintwal and Siri~Krishan Wasan.
\newblock Generalized {Reed--Muller} codes over $\mathbb{Z}_q$.
\newblock {\em Designs, Codes and Cryptography}, 54(2):149--166, 2010.

\bibitem[CAB12]{campbell2012magic}
Earl~T Campbell, Hussain Anwar, and Dan~E Browne.
\newblock Magic-state distillation in all prime dimensions using quantum reed-muller codes.
\newblock {\em Physical Review X}, 2(4):041021, 2012.

\bibitem[Cam]{Campbellcolor}
Earl~T. Campbell.
\newblock The smallest interesting color code.
\newblock Blog post https://earltcampbell.com/2016/09/26/the-smallest-interesting-colour-code/, accessed on 10/4/2024.

\bibitem[CGK17]{CGK17}
Shawn~X. Cui, Daniel Gottesman, and Anirudh Krishna.
\newblock Diagonal gates in the {C}lifford hierarchy.
\newblock {\em Physical Review A}, 95(1):012329, January 2017.
\newblock arXiv:1608.06596 [quant-ph].
\newblock \href {https://doi.org/10.1103/PhysRevA.95.012329} {\path{doi:10.1103/PhysRevA.95.012329}}.

\bibitem[CH17a]{campbell2017unified}
Earl~T Campbell and Mark Howard.
\newblock Unified framework for magic state distillation and multiqubit gate synthesis with reduced resource cost.
\newblock {\em Physical Review A}, 95(2):022316, 2017.

\bibitem[CH17b]{campbell2017unifying}
Earl~T Campbell and Mark Howard.
\newblock Unifying gate synthesis and magic state distillation.
\newblock {\em Physical review letters}, 118(6):060501, 2017.

\bibitem[DSRABR{\etalchar{+}}24]{da2024demonstration}
MP~Da~Silva, C~Ryan-Anderson, JM~Bello-Rivas, A~Chernoguzov, JM~Dreiling, C~Foltz, JP~Gaebler, TM~Gatterman, D~Hayes, N~Hewitt, et~al.
\newblock Demonstration of logical qubits and repeated error correction with better-than-physical error rates.
\newblock {\em arXiv preprint arXiv:2404.02280}, 2024.

\bibitem[Got97]{gottesman1997stabilizer}
Daniel Gottesman.
\newblock {\em Stabilizer codes and quantum error correction}.
\newblock California Institute of Technology, 1997.

\bibitem[Got24]{Gottesman2024}
Daniel Gottesman.
\newblock {Surviving as a Quantum Computer in a Classical World}, 2024.
\newblock book draft.

\bibitem[HH18]{haah2018codes}
Jeongwan Haah and Matthew~B Hastings.
\newblock Codes and protocols for distilling {$T$}, controlled-{$S$}, and {T}offoli gates.
\newblock {\em Quantum}, 2:71, 2018.

\bibitem[HKB{\etalchar{+}}24]{hangleiter_fault-tolerant_2024}
Dominik Hangleiter, Marcin Kalinowski, Dolev Bluvstein, Madelyn Cain, Nishad Maskara, Xun Gao, Aleksander Kubica, Mikhail~D. Lukin, and Michael~J. Gullans.
\newblock Fault-tolerant compiling of classically hard {{IQP}} circuits on hypercubes.
\newblock (arXiv:2404.19005), April 2024.
\newblock \href {https://arxiv.org/abs/2404.19005} {\path{arXiv:2404.19005}}.

\bibitem[HLC21]{HLC21}
Jingzhen Hu, Qingzhong Liang, and Robert Calderbank.
\newblock Climbing the diagonal {C}lifford hierarchy.
\newblock {\em arXiv preprint arXiv:2110.11923}, October 2021.
\newblock arXiv:2110.11923 [quant-ph].
\newblock URL: \url{http://arxiv.org/abs/2110.11923}.

\bibitem[HLC22a]{HLC22}
Jingzhen Hu, Qingzhong Liang, and Robert Calderbank.
\newblock Designing the quantum channels induced by diagonal gates.
\newblock {\em Quantum}, 6:802, September 2022.
\newblock arXiv:2109.13481 [quant-ph].
\newblock \href {https://doi.org/10.22331/q-2022-09-08-802} {\path{doi:10.22331/q-2022-09-08-802}}.

\bibitem[HLC22b]{Hu2022designingquantum}
Jingzhen Hu, Qingzhong Liang, and Robert Calderbank.
\newblock Designing the {Q}uantum {C}hannels {I}nduced by {D}iagonal {G}ates.
\newblock {\em {Quantum}}, 6:802, September 2022.
\newblock \href {https://doi.org/10.22331/q-2022-09-08-802} {\path{doi:10.22331/q-2022-09-08-802}}.

\bibitem[HLC22c]{hu2022divisible}
Jingzhen Hu, Qingzhong Liang, and Robert Calderbank.
\newblock Divisible codes for quantum computation, 2022.
\newblock URL: \url{https://arxiv.org/abs/2204.13176}, \href {https://arxiv.org/abs/2204.13176} {\path{arXiv:2204.13176}}.

\bibitem[HSB{\etalchar{+}}18]{Hui2018polar}
Dennis Hui, Sara Sandberg, Yufei Blankenship, Mattias Andersson, and Leefke Grosjean.
\newblock Channel coding in {5G New Radio}: {A} tutorial overview and performance comparison with {4G LTE}.
\newblock {\em IEEE Vehicular Technology Magazine}, 13(4):60--69, 2018.
\newblock \href {https://doi.org/10.1109/MVT.2018.2867640} {\path{doi:10.1109/MVT.2018.2867640}}.

\bibitem[JOKY18]{jochym2018disjointness}
Tomas Jochym-O’Connor, Aleksander Kubica, and Theodore~J Yoder.
\newblock Disjointness of stabilizer codes and limitations on fault-tolerant logical gates.
\newblock {\em Physical Review X}, 8(2):021047, 2018.

\bibitem[KB15]{kubica2015universal}
Aleksander Kubica and Michael~E Beverland.
\newblock Universal transversal gates with color codes: A simplified approach.
\newblock {\em Physical Review A}, 91(3):032330, 2015.

\bibitem[KBK22]{koutsioumpas2022smallestcodetransversalt}
Stergios Koutsioumpas, Darren Banfield, and Alastair Kay.
\newblock The smallest code with transversal t, 2022.
\newblock URL: \url{https://arxiv.org/abs/2210.14066}, \href {https://arxiv.org/abs/2210.14066} {\path{arXiv:2210.14066}}.

\bibitem[KCP16]{Kumar2016qRM}
Santhosh Kumar, Robert Calderbank, and Henry~D. Pfister.
\newblock Reed-muller codes achieve capacity on the quantum erasure channel.
\newblock In {\em 2016 IEEE International Symposium on Information Theory (ISIT)}, pages 1750--1754, 2016.

\bibitem[Kit97]{kitaev1997quantum}
A~Yu Kitaev.
\newblock Quantum computations: algorithms and error correction.
\newblock {\em Russian Mathematical Surveys}, 52(6):1191, 1997.

\bibitem[KKM{\etalchar{+}}17]{kudekar2017}
Shrinivas Kudekar, Santhosh Kumar, Marco Mondelli, Henry~D. Pfister, Eren Şaşoǧlu, and Rüdiger~L. Urbanke.
\newblock {Reed–Muller} codes achieve capacity on erasure channels.
\newblock {\em IEEE Transactions on Information Theory}, 63(7):4298--4316, 2017.
\newblock \href {https://doi.org/10.1109/TIT.2017.2673829} {\path{doi:10.1109/TIT.2017.2673829}}.

\bibitem[KLM{\etalchar{+}}23]{kliuchnikov2023shorter}
Vadym Kliuchnikov, Kristin Lauter, Romy Minko, Adam Paetznick, and Christophe Petit.
\newblock Shorter quantum circuits via single-qubit gate approximation.
\newblock {\em Quantum}, 7:1208, 2023.

\bibitem[KLZ98]{knill1998resilient}
Emanuel Knill, Raymond Laflamme, and Wojciech~H Zurek.
\newblock Resilient quantum computation: error models and thresholds.
\newblock {\em Proceedings of the Royal Society of London. Series A: Mathematical, Physical and Engineering Sciences}, 454(1969):365--384, 1998.

\bibitem[LLZ22]{leverrier2022towards}
Anthony Leverrier, Vivien Londe, and Gilles Z{\'e}mor.
\newblock Towards local testability for quantum coding.
\newblock {\em Quantum}, 6:661, 2022.

\bibitem[LZ22]{leverrier2022quantum}
Anthony Leverrier and Gilles Z{\'e}mor.
\newblock Quantum {T}anner codes.
\newblock In {\em 2022 IEEE 63rd Annual Symposium on Foundations of Computer Science (FOCS)}, pages 872--883. IEEE, 2022.

\bibitem[MHU14]{mondelli2014polar}
Marco Mondelli, S~Hamed Hassani, and R{\"u}diger~L Urbanke.
\newblock From polar to {Reed-Muller} codes: {A} technique to improve the finite-length performance.
\newblock {\em IEEE Transactions on Communications}, 62(9):3084--3091, 2014.

\bibitem[MS77]{MS77}
F.~J. MacWilliams and N.~J.~A. Sloane.
\newblock {\em The Theory of Error-Correcting Codes}.
\newblock North-Holland mathematical library; v. 16. North-Holland Pub. Co., Amsterdam; New York, N.Y., 1977.

\bibitem[NH22]{nezami2022classification}
Sepehr Nezami and Jeongwan Haah.
\newblock Classification of small triorthogonal codes.
\newblock {\em Physical Review A}, 106(1):012437, 2022.

\bibitem[NJBG24]{nadkarni2024entanglement}
Priya~J Nadkarni, Praveen Jayakumar, Arpit Behera, and Shayan~Srinivasa Garani.
\newblock Entanglement-assisted quantum {Reed-Muller} tensor product codes.
\newblock {\em Quantum}, 8:1329, 2024.

\bibitem[PK22]{panteleev2022asymptotically}
Pavel Panteleev and Gleb Kalachev.
\newblock Asymptotically good quantum and locally testable classical ldpc codes.
\newblock In {\em Proceedings of the 54th Annual ACM SIGACT Symposium on Theory of Computing}, pages 375--388, 2022.

\bibitem[PR13]{paetznick_universal_2013}
Adam Paetznick and Ben~W. Reichardt.
\newblock Universal {{Fault-Tolerant Quantum Computation}} with {{Only Transversal Gates}} and {{Error Correction}}.
\newblock {\em Phys. Rev. Lett.}, 111(9):090505, August 2013.
\newblock \href {https://doi.org/10.1103/PhysRevLett.111.090505} {\path{doi:10.1103/PhysRevLett.111.090505}}.

\bibitem[PY15]{pastawski2015fault}
Fernando Pastawski and Beni Yoshida.
\newblock Fault-tolerant logical gates in quantum error-correcting codes.
\newblock {\em Physical Review A}, 91(1):012305, 2015.

\bibitem[RAC{\etalchar{+}}24]{reichardt_demonstration_2024}
Ben~W. Reichardt, David Aasen, Rui Chao, Alex Chernoguzov, Wim {van Dam}, John~P. Gaebler, Dan Gresh, Dominic Lucchetti, Michael Mills, Steven~A. Moses, Brian Neyenhuis, Adam Paetznick, Andres Paz, Peter~E. Siegfried, Marcus~P. {da Silva}, Krysta~M. Svore, Zhenghan Wang, and Matt Zanner.
\newblock Demonstration of quantum computation and error correction with a tesseract code.
\newblock (arXiv:2409.04628), September 2024.
\newblock \href {https://arxiv.org/abs/2409.04628} {\path{arXiv:2409.04628}}.

\bibitem[RBB{\etalchar{+}}24]{ryan-anderson_high-fidelity_2024}
C.~{Ryan-Anderson}, N.~C. Brown, C.~H. Baldwin, J.~M. Dreiling, C.~Foltz, J.~P. Gaebler, T.~M. Gatterman, N.~Hewitt, C.~Holliman, C.~V. Horst, J.~Johansen, D.~Lucchetti, T.~Mengle, M.~Matheny, Y.~Matsuoka, K.~Mayer, M.~Mills, S.~A. Moses, B.~Neyenhuis, J.~Pino, P.~Siegfried, R.~P. Stutz, J.~Walker, and D.~Hayes.
\newblock High-fidelity and {{Fault-tolerant Teleportation}} of a {{Logical Qubit}} using {{Transversal Gates}} and {{Lattice Surgery}} on a {{Trapped-ion Quantum Computer}}.
\newblock April 2024.
\newblock \href {https://arxiv.org/abs/2404.16728} {\path{arXiv:2404.16728}}.

\bibitem[RCNP20]{rengaswamy2020optimality}
Narayanan Rengaswamy, Robert Calderbank, Michael Newman, and Henry~D Pfister.
\newblock On optimality of {CSS} codes for transversal {$T$}.
\newblock {\em IEEE Journal on Selected Areas in Information Theory}, 1(2):499--514, 2020.

\bibitem[{\c{S}}a{\c{s}}12]{sasoglu2012polarization}
Eren {\c{S}}a{\c{s}}o{\u{g}}lu.
\newblock Polarization and polar codes.
\newblock {\em Foundations and Trends{\textregistered} in Communications and Information Theory}, 8(4):259--381, 2012.

\bibitem[SK05]{sarvepalli2005nonbinary}
Pradeep~Kiran Sarvepalli and Andreas Klappenecker.
\newblock Nonbinary quantum reed-muller codes.
\newblock In {\em Proceedings. International Symposium on Information Theory, 2005. ISIT 2005.}, pages 1023--1027. IEEE, 2005.

\bibitem[Ste99]{Steane1999}
A.M. Steane.
\newblock Quantum {Reed-Muller} codes.
\newblock {\em IEEE Transactions on Information Theory}, 45(5):1701--1703, 1999.
\newblock \href {https://doi.org/10.1109/18.771249} {\path{doi:10.1109/18.771249}}.

\bibitem[Ter15]{Terhal2015}
Barbara~M. Terhal.
\newblock Quantum error correction for quantum memories.
\newblock {\em Rev. Mod. Phys.}, 87:307--346, Apr 2015.
\newblock URL: \url{https://link.aps.org/doi/10.1103/RevModPhys.87.307}, \href {https://doi.org/10.1103/RevModPhys.87.307} {\path{doi:10.1103/RevModPhys.87.307}}.

\bibitem[VK22]{vasmer2022morphing}
Michael Vasmer and Aleksander Kubica.
\newblock Morphing quantum codes.
\newblock {\em PRX Quantum}, 3(3):030319, 2022.

\bibitem[WBB22]{Webster2022xpstabiliser}
Mark~A. Webster, Benjamin~J. Brown, and Stephen~D. Bartlett.
\newblock The {XP} {S}tabiliser {F}ormalism: a {G}eneralisation of the {P}auli {S}tabiliser {F}ormalism with {A}rbitrary {P}hases.
\newblock {\em {Quantum}}, 6:815, Sept. 2022.
\newblock \href {https://doi.org/10.22331/q-2022-09-22-815} {\path{doi:10.22331/q-2022-09-22-815}}.

\bibitem[YHH{\etalchar{+}}23]{ye_logical_2023}
Yangsen Ye, Tan He, He-Liang Huang, Zuolin Wei, Yiming Zhang, Youwei Zhao, Dachao Wu, Qingling Zhu, Huijie Guan, Sirui Cao, Fusheng Chen, Tung-Hsun Chung, Hui Deng, Daojin Fan, Ming Gong, Cheng Guo, Shaojun Guo, Lianchen Han, Na~Li, Shaowei Li, Yuan Li, Futian Liang, Jin Lin, Haoran Qian, Hao Rong, Hong Su, Shiyu Wang, Yulin Wu, Yu~Xu, Chong Ying, Jiale Yu, Chen Zha, Kaili Zhang, Yong-Heng Huo, Chao-Yang Lu, Cheng-Zhi Peng, Xiaobo Zhu, and Jian-Wei Pan.
\newblock Logical {{Magic State Preparation}} with {{Fidelity}} beyond the {{Distillation Threshold}} on a {{Superconducting Quantum Processor}}.
\newblock {\em Phys. Rev. Lett.}, 131(21):210603, November 2023.
\newblock \href {https://doi.org/10.1103/PhysRevLett.131.210603} {\path{doi:10.1103/PhysRevLett.131.210603}}.

\bibitem[YTC16]{yoder2016universal}
Theodore~J Yoder, Ryuji Takagi, and Isaac~L Chuang.
\newblock Universal fault-tolerant gates on concatenated stabilizer codes.
\newblock {\em Physical Review X}, 6(3):031039, 2016.

\bibitem[ZCC11]{zeng2011tranversality}
Bei Zeng, Andrew Cross, and Isaac~L. Chuang.
\newblock Transversality versus universality for additive quantum codes.
\newblock {\em IEEE Transactions on Information Theory}, 57(9):6272--6284, 2011.

\end{thebibliography}
